\pdfoutput=1
\RequirePackage{ifpdf}
\ifpdf % We~are running pdfTeX in pdf mode
\documentclass[pdftex]{sigma}%,draft
\else
\documentclass{sigma}
\fi

\usepackage[cal=boondoxo]{mathalfa}

\newcommand{\Id}{\mathbbmss{1}}

\newcommand{\rmh}{\textnormal{h}}
\newcommand{\rmk}{\textnormal{k}}

\newcommand{\rmp}{\textnormal{p}}
\newcommand{\rmq}{\textnormal{q}}

\newcommand{\op}[1]{\operatorname{#1}}
\input xy
\xyoption{all} %\tolerance=500

\usepackage{bbm}

\DeclareMathOperator{\GL}{GL}
\DeclareMathOperator{\Hom}{Hom}

\newcommand{\catname}[1]{\textnormal{\texttt{#1}}}

\font\black=cmbx10 \font\sblack=cmbx7 \font\ssblack=cmbx5 \font\blackital=cmmib10 \skewchar\blackital='177
\font\sblackital=cmmib7 \skewchar\sblackital='177 \font\ssblackital=cmmib5 \skewchar\ssblackital='177
\font\sanss=cmss10 \font\ssanss=cmss8 %scaled 900
\font\sssanss=cmss8 scaled 600 \font\blackboard=msbm10 \font\sblackboard=msbm7 \font\ssblackboard=msbm5
\font\caligr=eusm10 \font\scaligr=eusm7 \font\sscaligr=eusm5  \font\fraktur=eufm10
\font\sfraktur=eufm7 \font\ssfraktur=eufm5 
\font\bsymb=cmsy10 scaled\magstep2
\def\all#1{\setbox0=\hbox{\lower1.5pt\hbox{\bsymb
 \char"38}}\setbox1=\hbox{$_{#1}$} \box0\lower2pt\box1\;}
\def\exi#1{\setbox0=\hbox{\lower1.5pt\hbox{\bsymb \char"39}}
 \setbox1=\hbox{$_{#1}$} \box0\lower2pt\box1\;}

\def\tx#1{{\fam0\relax#1}}

\newfam\bifam
\textfont\bifam=\blackital \scriptfont\bifam=\sblackital \scriptscriptfont\bifam=\ssblackital

\newfam\blfam
\textfont\blfam=\black \scriptfont\blfam=\sblack \scriptscriptfont\blfam=\ssblack

\newfam\bbfam
\textfont\bbfam=\blackboard \scriptfont\bbfam=\sblackboard \scriptscriptfont\bbfam=\ssblackboard

\newfam\ssfam
\textfont\ssfam=\sanss \scriptfont\ssfam=\ssanss \scriptscriptfont\ssfam=\sssanss

\newfam\clfam
\textfont\clfam=\caligr \scriptfont\clfam=\scaligr \scriptscriptfont\clfam=\sscaligr

\newfam\frfam
\textfont\frfam=\fraktur \scriptfont\frfam=\sfraktur \scriptscriptfont\frfam=\ssfraktur

\def\hpb#1{\setbox0=\hbox{${#1}$}
 \copy0 \kern-\wd0 \kern.2pt \box0}
\def\vpb#1{\setbox0=\hbox{${#1}$}
 \copy0 \kern-\wd0 \raise.08pt \box0}

\def\pmb#1{\setbox0\hbox{${#1}$} \copy0 \kern-\wd0 \kern.2pt \box0}
\def\pmbb#1{\setbox0\hbox{${#1}$} \copy0 \kern-\wd0
 \kern.2pt \copy0 \kern-\wd0 \kern.2pt \box0}
\def\pmbbb#1{\setbox0\hbox{${#1}$} \copy0 \kern-\wd0
 \kern.2pt \copy0 \kern-\wd0 \kern.2pt
 \copy0 \kern-\wd0 \kern.2pt \box0}
\def\pmxb#1{\setbox0\hbox{${#1}$} \copy0 \kern-\wd0
 \kern.2pt \copy0 \kern-\wd0 \kern.2pt
 \copy0 \kern-\wd0 \kern.2pt \copy0 \kern-\wd0 \kern.2pt \box0}
\def\pmxbb#1{\setbox0\hbox{${#1}$} \copy0 \kern-\wd0 \kern.2pt
 \copy0 \kern-\wd0 \kern.2pt
 \copy0 \kern-\wd0 \kern.2pt \copy0 \kern-\wd0 \kern.2pt
 \copy0 \kern-\wd0 \kern.2pt \box0}

\mathchardef\za="710B %\alpha
\mathchardef\zb="710C %\beta
\mathchardef\zg="710D %\gamma
\mathchardef\zd="710E %\delta
\mathchardef\zve="710F %\epsilon
\mathchardef\zz="7110 %\zeta
\mathchardef\zh="7111 %\eta
\mathchardef\zvy="7112 %\theta
\mathchardef\zi="7113 %\iota
\mathchardef\zk="7114 %\kappa
\mathchardef\zl="7115 %\lambda
\mathchardef\zm="7116 %\mu
\mathchardef\zn="7117 %\nu
\mathchardef\zx="7118 %\xi
\mathchardef\zp="7119 %\pi
\mathchardef\zr="711A %\rho
\mathchardef\zs="711B %\sigma
\mathchardef\zt="711C %\tau
\mathchardef\zu="711D %\upsilon
\mathchardef\zvf="711E %\phi
\mathchardef\zq="711F %\chi
\mathchardef\zc="7120 %\psi
\mathchardef\zw="7121 %\omega
\mathchardef\ze="7122 %\varepsilon
\mathchardef\zy="7123 %\vartheta
\mathchardef\zf="7124 %\varomega
\mathchardef\zvr="7125 %\varrho
\mathchardef\zvs="7126 %\varsigma
\mathchardef\zf="7127 %\varphi
\mathchardef\zG="7000 %\Gamma
\mathchardef\zD="7001 %\Delta
\mathchardef\zY="7002 %\Theta
\mathchardef\zL="7003 %\Lambda
\mathchardef\zX="7004 %\Xi
\mathchardef\zP="7005 %\Pi
\mathchardef\zS="7006 %\Sigma
\mathchardef\zU="7007 %\Upsilon
\mathchardef\zF="7008 %\Phi
\mathchardef\zW="700A %\Omega
\mathchardef\zC="7009 %\Psi

\newcommand{\be}{\begin{equation}}
\newcommand{\ee}{\end{equation}}

\newcommand{\bea}{\begin{eqnarray}}
\newcommand{\eea}{\end{eqnarray}}
\def\*{{\textstyle *}}
\newcommand{\R}{{\mathbb R}}

\newcommand{\Z}{{\mathbb Z}}

\newcommand{\s}{{\textstyle *}}

\def\ul{\underline}

\def\la{\langle}

\def\xi{\tx{i}}

\def\cM{\cal M}

\def\s*{{\scriptstyle *}}

\def\cO{\mathcal{O}}
\def\cE{\mathcal{E}}
\def\cF{\mathcal{F}}

\def\cM{\mathcal{M}}
\def\cU{\mathcal{U}}
\def\ul{\underline}
\newcommand{\beas}{\begin{eqnarray*}}
\newcommand{\eeas}{\end{eqnarray*}}

\newcommand{\0}{\otimes}
\newcommand{\id}{\op{id}}

\newcommand{\z}{\Z_2^n}

\numberwithin{equation}{section}

\newtheorem{Theorem}{Theorem}[section]
\newtheorem{Proposition}[Theorem]{Proposition}
 { \theoremstyle{definition}
\newtheorem{Definition}[Theorem]{Definition}
\newtheorem{Example}[Theorem]{Example}
\newtheorem{Remark}[Theorem]{Remark} }

\begin{document}
%\allowdisplaybreaks

\newcommand{\arXivNumber}{2011.01012}

\renewcommand{\PaperNumber}{060}

\FirstPageHeading

\ShortArticleName{Linear ${\mathbb Z}_2^n$-Manifolds and Linear Actions}

\ArticleName{Linear $\boldsymbol{{\mathbb Z}_2^n}$-Manifolds and Linear Actions}

\AuthorNameForHeading{A.J.~Bruce, E.~Ibargu\"engoytia and N.~Poncin}

 \Author{Andrew James BRUCE, Eduardo IBARGU\"ENGOYTIA and Norbert PONCIN}

\Address{Department of Mathematics, University of Luxembourg,\\ Maison du Nombre, 6, avenue de la Fonte, L-4364 Esch-sur-Alzette, Luxembourg}
\Email{\href{mailto:andrewjamesbruce@googlemail.com}{andrewjamesbruce@googlemail.com}, \href{mailto:eduardo.ibarguengoytia@uni.lu}{eduardo.ibarguengoytia@uni.lu}, \newline
\hspace*{14.0mm}\href{mailto:norbert.poncin@uni.lu}{norbert.poncin@uni.lu}}

\ArticleDates{Received November 05, 2020, in final form May 30, 2021; Published online June 16, 2021}

\Abstract{We establish the representability of the general linear ${\mathbb Z}_2^n$-group and use the restricted functor of points~-- whose test category is the category of ${\mathbb Z}_2^n$-manifolds over a~single topological point~-- to define its smooth linear actions on ${\mathbb Z}_2^n$-graded vector spaces and linear ${\mathbb Z}_2^n$-manifolds. Throughout the paper, particular emphasis is placed on the full faithfulness and target category of the restricted functor of points of a number of categories that we are using.}


\Classification{58A50; 58C50; 14A22; 14L30; 13F25; 16L30; 17A70}

\section{Introduction}

In order to be able to deal with the technical details of vector bundles and related structures in~the category of $\z$-manifolds (for $n=1$ see~\cite{Balduzzi:2011}), we need some foundational results on $\z$-Lie groups and their smooth linear actions on linear $\z$-manifolds. However, the proofs of some folklore results, i.e., results that we tended to accept somewhat hands-waving, are often not at~all obvious in the $\z$-context. The present paper, beyond its supposed applications, intrinsic interest and the beauty of some of its developments, raises the question of the scientific value of~``results'' that are partially based on speculations.

Loosely speaking, $\Z_{2}^{n}$-manifolds ($\z=\Z_2^{\times n}$) are ``manifolds'' for which the structure sheaf has a $\Z_{2}^{n}$-grading and the commutation rules for the local coordinates comes from the standard scalar product (see~\cite{BruceGrabowski, Bruce:2019b, Bruce:2018, Bruce:2019,Covolo:2016,Covolo:2016a,Covolo:2016b, Covolo:2016c, Poncin:2016} for details). This is not just a trivial or straightforward generalization of the notion of a standard supermanifold, as one has to deal with formal coordinates that anticommute with other formal coordinates, but are themselves \emph{not} nilpotent. Due to the presence of formal variables that are not nilpotent, formal power series are used rather than polynomials. Recall that for standard supermanifolds all functions are polynomial in the Grassmann odd variables. The theory of $\Z_2^n$-geometry is currently being developed and many foundational questions remain. For completeness, we include Appendix~\ref{app:ZZnGeom} in which the foundations of $\Z_2^n$-geometry are given. In this paper, we examine the relation between $\Z_2^n$-graded vector spaces and linear $\Z_2^n$-manifolds, and then we use this to define linear actions of $\Z_2^n$-Lie groups.

In the literature on supergeometry, the symbol $\R^{p|q}$ has two distinct, but related meanings. First, we have the notion of a $\Z_2$-graded, or super, vector space with $p$ even and $q$ odd dimensions, i.e., the real vector space $\mathbf{R}^{p|q}=\R^p \bigoplus \R^q$. Secondly, we have the locally ringed space $\R^{p|q}=\big( \R^p, C^\infty_{\R^p}[\zx]\big)$, where $\zx^i$ ($i\in\{1, \dots, q\}$) are the generators of a Grassmann algebra. The difference can be highlighted by identifying the points of these objects. The $\Z_2$-graded vector space has as its underlying topological space $\R^{p+q}$ (we just forget the ``superstructure''), while for the locally ringed space the topological space is $\R^p$. There are several ways of showing that these two notions are deeply tied. In particular, the category of finite dimensional super vector spaces is equivalent to the category of ``linear supermanifolds'' (see~\cite{Bernstein:2013,Kostant:1977, Leites:1980, Leites:1991,Sanchez-Valenzuela:1988}).

In this paper, we will show that the categories of finite dimensional $\z$-graded vector spaces $\mathbf{V}$ and linear $\z$-manifolds $V$ are isomorphic. We~do this by explicitly constructing a ``manifoldification'' functor $\mathcal{M}$, which associates a linear $\z$-manifold to every finite dimensional $\Z_2^n$-graded vector space, and a ``vectorification'' functor, which is the inverse of the previous functor. It~turns out that working in a coordinate-independent way $(\mathbf{V}, V)$ is much more complex than working with canonical coordinates $\big(\mathbf{R}^{p|\ul q},\R^{p|\ul q}\big)$.

Throughout this article, a special focus is placed on functors of points. The functor of points has been used informally in physics as from the very beginning. It~is actually of importance in~contexts where there is no good notion of point as in super- and $\Z_2^n$-geometry and in algebraic geometry. For instance, homotopical algebraic geometry~\cite{TVI, TVII} and its generalisation that goes under the name of homotopical algebraic geometry over differential operators~\cite{BPP:KTR, BPP:HAC}, are completely based on the functor of points approach. In this paper, we are particularly interested in functors of $\zL$-points, i.e., functors of points from appropriate locally small categories $\tt C$ to a~functor category whose source is not the category ${\tt C}^{\op{op}}$ but the category ${\tt G}$ of $\z$-Grassmann algebras $\zL$. However, functors of points that are restricted to the very simple test category ${\tt G}$ are fully faithful only if we replace the target category of the functor category by a subcategory of the category of sets.

More precisely, closely related to the above isomorphism of supervector spaces and linear supermanifolds is the so-called ``even rules''. Loosely this means including extra odd parameters to render everything even and in doing so one removes copious sign factors (see for example~\cite[Section~1.7]{Deligne:1999}). We~will establish an analogue of the even rules in our higher graded setting which we will refer to as the ``zero degree rules'' (see Definition~\ref{def:Vfun}). To address this we will make extensive use of $\Z_2^n$-Grassmann algebras $\zL$, $\Lambda$-points and the Schwarz--Voronov embedding, which is a fully faithful functor of points $\mathcal{S}$ from $\z$-manifolds to a functor category with source~${\tt G}$ and the category of Fr\'echet manifolds (see for
example~\cite{Hamilton:82}) over commutative Fr\'echet algebras as target (see~\cite{Bruce:2019b}). We~show that the zero degree rules functor $\mathcal F$, understood as an assignment of a functor from ${\tt G}$ to the category of modules over commutative (Fr\'echet) algebras, given a (finite dimensional) $\Z_2^n$-graded vector space, is fully faithful (see Theorem~\ref{trm:FFullyFaith} and Proposition~\ref{FunPtsFinVecFF2}). The ``zero degree rules'' allow one to identity a finite dimensional $\Z_2^n$-graded vector space, considered as a~functor, with the functor of points of its ``manifoldification''. In other words, the composite $\mathcal{S}\circ\mathcal{M}$ and~$\mathcal{F}$ can be viewed as functors between the same categories and are naturally isomorphic. This identification is fundamental when describing linear group actions on $\Z_2^n$-graded vector spaces and linear $\Z_2^n$-manifolds.

Another important part of this work is the category of $\z$-Lie groups and its fully faithful functor of points valued in a functor category with ${\tt G}$ as source category and Fr\'echet Lie groups over commutative Fr\'echet algebras as target category. We~define the general linear $\Z_2^n$-group as a functor in this functor category and show that it is representable, i.e., is a genuine $\Z_2^n$-manifold (see Theorem~\ref{thm:GLRep}). This leads to interesting insights into the computation of the inverse of an invertible degree zero $\z$-graded square matrix of dimension $p|\ul q$ with entries in a~$\z$-commutative algebra. Furthermore, the approach using $\Lambda$-points and the zero rules allows us to construct a canonical smooth linear action of the general linear $\Z_2^n$-group on $\Z_2^n$-graded vector spaces and linear $\Z_2^n$-manifolds. All these notions, in particular the equivalence between the definitions of a smooth linear action as natural transformation and as $\z$-morphism, are carefully and explicitly explained in the main text.

We remark that many of the statements in this paper are not surprising in themselves. However, due to the subtleties of $\Z_2^n$-geometry, many of the proofs are much more involved than the analogue statements in supergeometry. The main source of difficulty is that one has to deal with formal power series in non-zero degree coordinates, rather than polynomials as one does in supergeometry. This forces one to work with infinite dimensional objects and the $\mathcal{J}$-adic topology ($\mathcal{J}$ is the ideal generated by non-zero degree elements). Many of the ``categorical'' proofs are significantly more involved than the proofs for supermanifolds. In general, there is a~lot of work to establish the form of natural transformations as we have non-nilpotent elements of non-zero degree. While the ethos of the proofs may be standard, they are not, in general, simple or routine checks due to the aforementioned subtleties.

\textbf{Motivation from physics:} $\Z_{2}^{n}$-gradings ($n \geq 2$) can be found in the theory of parastatistics (see for example~\cite{Druhl:1970,Green:1953,Greenberg:1965,Yang:2001}) and in relation to an alternative approach to supersymmetry~\cite{Aizawa:2018,Aizawa:2016,Tolstoy:2013}. ``Higher graded'' generalizations of the super Schr\"{o}dinger algebra (see~\cite{Aizawa:2017}) and the super Poincar\'{e} algebra (see~\cite{Bruce:2019a}) have appeared in the literature. Furthermore, such gradings appear in the theory of mixed symmetry tensors as found in string theory and some formulations of~supergravity (see~\cite{Bruce:2018a}). It~must also be pointed out that quaternions and more general Clifford algebras can be understood as $\Z_2^n$-graded $\Z_2^n$-commutative algebras~\cite{Albuquerque:1999, Albuquerque:2002}. Finally, \emph{any} ``sign rule'' can be interpreted in terms of a $\Z_{2}^{n}$-grading (see~\cite{Covolo:2016}).

 \textbf{Background:} For various sheaf-theoretical notions we will draw upon Hartshorne~\cite[Chapter~II]{Hartshorne:1977} and Tennison~\cite{Tennison:1975}. There are several good introductory books on the theory of super\-ma\-ni\-folds including Bartocci, Bruzzo and~Hern\'{a}ndez Ruip\'{e}rez~\cite{Bartocci:1991}, Carmeli, Caston and~Fioresi~\cite{Carmeli:2011}, Deligne and~Morgan~\cite{Deligne:1999} and Varadarajan~\cite{Varadarajan:2004}. For categorical notions we will be based on Mac Lane~\cite{MacLane1998}. We~will make extensive use of the constructions and statements found in our earlier publications~\cite{Bruce:2019b,Bruce:2018,Bruce:2019}.

\section[Z2n-graded vector spaces and Linear Z2n-manifolds]
{$\boldsymbol{\Z_{2}^{n}}$-graded vector spaces and Linear $\boldsymbol{\Z_2^n}$-manifolds}

\subsection[Z2n-graded vector spaces and the zero degree rules]%\label{SubSec:VecSpLin}
{$\boldsymbol{\Z_2^n}$-graded vector spaces and the zero degree rules }
When dealing with linear superalgebra one encounters the so-called \emph{even rules} (see~\cite[Section~1.8]{Carmeli:2011},~\cite[Section~1.7]{Deligne:1999} and~\cite[pp.~123--124]{Varadarajan:2004}, for example). Very informally, the idea is to remove sign factors by allowing extra parameters to render the situation completely even. The idea has been applied in physics since the early days of supersymmetry. More precisely, let
\[
V(A)=(A\otimes V)_0
\]
be the even part of the extension of scalars in a (real) super vector space $V$, from the base field~$\R$ to a supercommutative algebra $A\in\tt SAlg$ (in the even rules that we are about to describe, it actually suffices to use supercommutative Grassmann algebras $A=\R[\zvy_1,\dots, \zvy_N]$: the $\zvy_i$ are then the extra parameters mentioned before). The main result in even rules states, roughly, that defining a morphism $\phi\colon V\0 V\to V$ is equivalent to defining it functorially on the even part of~$V$ after extension of scalars, i.e., is equivalent to defining a functorial family of morphisms
\[
\phi(A)\colon\ V(A)\times V(A)\to V(A)
\]
(indexed by $A\in\tt SAlg$). More precisely, there is a $1:1$ correspondence between parity respecting $\R$-linear maps $\phi\colon V_1\otimes\dots\otimes V_n\to V$ and functorial families
\[
\phi(A)\colon\ V_1(A)\times \dots \times V_n(A)\to V(A)
\]
($A\in\tt SAlg$) of $A_0$-multilinear maps.

We now proceed to generalise this theorem to the $\Z_2^n$-setting. We~will work with the category $\Z_2^n\tt GrAlg$ of $\Z_2^n$-Grassmann algebras rather than the category $\Z_2^n\tt Alg$ of all $\Z_2^n$-commutative algebras.

Let ${V} =\bigoplus_{i = 0}^{N} {V}_{\gamma_i}$ be a (real) $\Z_2^n$-graded vector space, i.e., a (real) vector space with a~direct sum decomposition over $i\in\{0,\dots,N\}$ (we say that the vectors of $V_{\zg_i}$ are of degree \mbox{$\zg_i\in\z$}). The category of $\Z_2^n$-graded vector spaces (not necessarily finite dimensional) we denote as $\Z_2^n\catname{Vec}$. Morphisms in this category are degree preserving linear maps. We~denote the category of~modu\-les over commutative algebras as $\tt AMod$ (see Appendix~\ref{app:CatModAlg}).

To $V$ we associate a functor
\begin{equation*}
{V}(-)\in{\tt Fun}_0\big(\Z_2^n\catname{Pts}^{\textnormal{op}},{\tt AMod}\big)
\end{equation*}
in the category of those functors whose value on any $\Z_2^n$-Grassmann algebra $\zL\in\Z_2^n{\tt Pts}^{\op{op}}$ is a $\zL_0$-module, and of those natural transformations that have $\zL_0$-linear $\zL$-components. The functor $V(-)$ is essentially the tensor product functor $-\0 V$. It~is built in the following way. First, for every $\Z_2^n$-Grassmann algebra $\Lambda$, we define
\[
{V}(\Lambda) :=(\Lambda \otimes {V})_0\in\zL_0\tt Mod,
\]
where the tensor product is over $\R$. Secondly, for any $\Z_2^n$-algebra morphism $\varphi^* \colon \Lambda \rightarrow \Lambda^\prime$, we define
\[
{V}(\varphi^*) := \big(\varphi^* \otimes \Id_{{V}}\big)_0,
\]
where the {\small RHS} is the restriction of $\zf^*\0\Id_{ V}$ to the degree 0 part of $\zL\0 V$, so that ${V}(\zf^*)$ is an~$\tt AMod$-morphism
\begin{gather}\label{CatAMod}
{V}(\zf^*)\colon\ {V}(\zL)\to {V}(\zL'),
\end{gather}
whose associated algebra morphism is the restriction $(\zf^*)_0\colon \zL_0\to\zL'_0$. It is clear that $V(-)$ respects compositions and identities and is thus a functor, as announced.

We thus get an assignment
\[
\cF\colon\ \Z_2^n{\tt Vec}\ni V\mapsto \cF( V):= V(-)\in{\tt Fun}_0\big(\Z_2^n{\tt Pts}^{\op{op}}, \tt AMod\big).
\]
The map $\cF$ is essentially $-\0\bullet$ and is itself a functor. It~associates to any grading respecting linear map $\phi \colon {V} \rightarrow {W}$ and any $\Z_2^n$-Grassmann algebra $\zL$, a $\zL_0$-linear map
\[
\phi_\Lambda := (\Id_\Lambda \otimes \phi)_0\colon\ {V}(\Lambda) \rightarrow {W}(\Lambda).
\]
The family $\cF(\phi):=\phi_-$ is a natural transformation from $\cF( V)$ to $\cF( W)$. Since $\cF$ respects compositions and identities, it is actually a functor valued in the restricted functor category ${\tt Fun}_0\big(\Z_2^n{\tt Pts}^{\op{op}},\tt AMod\big)$.
\begin{Definition}\label{def:Vfun}
The functor
\[
\mathcal{F} \colon\ \Z_2^n\catname{Vec} \longrightarrow {\tt Fun}_0\big(\Z_2^n\catname{Pts}^\textnormal{op}, \catname{AMod}\big)
\]
is referred to as the {\it zero degree rules functor}.
\end{Definition}
\begin{Theorem}\label{trm:FFullyFaith}
The zero degree rules functor
\[
\mathcal{F} \colon\ \Z_2^n\catname{Vec} \longrightarrow {\tt Fun}_0\big(\Z_2^n\catname{Pts}^\textnormal{op}, \catname{AMod}\big)
\]
is fully faithful, i.e., for any pair of $\Z_2^n$-graded vector spaces ${V}$ and ${W}$, the map
\[
\mathcal{F}_{ V, W} \colon \ \Hom_{\Z_2^n\catname{Vec}}({V} , {W}) \longrightarrow \Hom_{{\tt Fun}_0(\Z_2^n\catname{Pts}^\textnormal{op}, \catname{AMod})} (\mathcal{F}({V}) ,\mathcal{F}({W}) ),
\]
is a bijection.
\end{Theorem}

This result is the $\Z_2^n$-counterpart of the $1:1$ correspondence mentioned above.

\begin{proof}
We show first that the map $\cF_{ V, W}$ is injective. Let $\phi,\psi\colon V\to W$ be two degree preserving linear maps, and assume that $\cF(\phi)=\phi_-=\psi_-=\cF(\psi)$, so that, for any $\zL\in\Z_2^n{\tt Pts}^{\op{op}}$ and any~$\zl\0 v\in V(\zL)$, we have
\begin{gather}\label{AssInj}
\zl\0 \phi(v)=\phi_\zL(\zl\0 v)=\psi_\zL(\zl\0 v)=\zl\0 \psi(v).
\end{gather}
Notice now that
\[
V(\zL)=(\zL\0 V)_0=\bigoplus_{i=0}^N\zL_{\zg_i}\0 V_{\zg_i}
\]
and let $\zL$ be the Grassmann algebra
\begin{gather}\label{Lambda1}
\zL_1:=\R[[\zvy_1,\dots,\zvy_N]]
\end{gather}
that has exactly one generator $\zvy_j$ in each non-zero degree $\zg_j\in\Z_2^n$ ($N=2^n-1$). For any $v_j\in V_{\zg_j}$, equation~\eqref{AssInj} implies that $\zvy_j \otimes \zvf(v_j)=\zvy_j \otimes\psi(v_j)$, so that $\phi$ and $\psi$ coincide on $ V_{\zg_j}$, for all $j\in\{1,\dots,N\}$. For $v_0\in V_0:= V_{\zg_0}$ and $\zl=1$, the same equation shows that $\phi$ and $\psi$ coincide also on $ V_0$.

To prove surjectivity, we consider an arbitrary natural transformation $\Phi_-\colon V(-)\to W(-)$ and will define a degree 0 linear map $\phi\colon V\to W$, such that $\cF(\phi)=\phi_- =\Phi_-$, i.e., such that, for any $\zL\in\Z_2^n{\tt Pts}^{\op{op}}$, we have
\begin{gather*}%\label{CondSurj1}
\phi_\zL=\Phi_\zL
\end{gather*}
on $ V(\zL)$.
Since an element of $ V(\zL)$ (uniquely) decomposes into a sum over $i\in\{0,\dots, N\}$ of~(not uniquely defined) finite sums of decomposable tensors $\zl_i\0 v_i$, with (not uniquely defined) factors~$\zl_i$ and $v_i$ of degree $\zg_i$, it suffices to show that
\begin{gather}\label{CondSurj2}
\phi_\zL(\zl_i\0 v_i)=\Phi_\zL(\zl_i\0 v_i),
\end{gather}
for all $i\in\{0,\dots,N\}$.

Further, it suffices to prove condition~\eqref{CondSurj2} for $\zL_1$ (see~\eqref{Lambda1}) and for the tensors $\zvy_i\0 v_i$ ($\zvy_0:=1$, $v_i\in V_{\zg_i}$, $i\in\{0,\dots,N\}$). The observation follows from naturality of $\Phi$. Indeed, assume that~\eqref{CondSurj2} is satisfied for $\zL_1$ and the decomposable tensors just mentioned (assumption~($\star$)). For any fixed $i\in\{1,\dots,N\}$ (resp., $i=0$), and for $\zL$, $\zl_i$ and $v_i$ as above, let $\zf^*\colon \zL_1\to \zL$ be the $\Z_2^n$-algebra map defined by $\zf^*(\zvy_i)=\zl_i$, $\zf^*(\zvy_j)=0$ for $j\neq i$, $j\neq 0$, and $\zf^*(\zvy_0)=\zf^*(1)=1$ (resp., $\zf^*(\zvy_j)=0$ for all $j\neq 0$, and $\zf^*(\zvy_0)=\zf^*(1)=1$). For $i\in\{1,\dots, N\}$, when applying the naturality condition
\begin{center}
\leavevmode
\begin{xy}
(0,20)*+{ {V}(\Lambda_1)}="a"; (40,20)*+{{W}(\Lambda_1)}="b";
(0,0)*+{{V}(\Lambda)}="c"; (40,0)*+{{W}(\Lambda)}="d";%
{\ar "a";"b"}?*!/_3mm/{\Phi_{\Lambda_1}};%
{\ar "a";"c"}?*!/^6mm/{{V}(\zf^*)};%
{\ar "b";"d"}?*!/_6mm/{{W}(\zf^*)};%
{\ar "c";"d"}?*!/^3mm/{\Phi_{\Lambda}};%
\end{xy}
\end{center}
to $\zvy_i\0 v_i$, we get clockwise
\[
W(\zf^*)(\Phi_{\zL_1}(\zvy_i\0 v_i))=W(\zf^*)(\zvy_i\0\phi(v_i))=\zf^*(\zvy_i)\0\phi(v_i)=\zl_i\0 \zvf(v_i)=\phi_\zL(\zl_i\0 v_i),
\]
in view of ($\star$), whereas anticlockwise we obtain
\[
\Phi_\zL(V(\zf^*)(\zvy_i\0 v_i))=\Phi_\zL(\zl_i\0 v_i).
\]
Hence, condition~\eqref{CondSurj2} holds for $i\in\{1,\dots,N\}$. For $i=0$, the above naturality condition yields $1\0\phi(v_0)=\Phi_\zL(1\0 v_0)$, when applied to $1\0 v_0$. In view of the $\zL_0$-linearity of the $\zL$-components of the natural transformations considered, we get now
\[
\phi_\zL(\zl_0\0 v_0)=\zl_0\,\phi_\zL(1\0 v_0)=\zl_0(1\0\phi(v_0))=\Phi_\zL(\zl_0\0 v_0).
\]
Finally, condition~\eqref{CondSurj2} holds for an arbitrary $\zL$, if it holds for $\zL_1$.

Surjectivity now reduces to constructing a $\Z_2^n\tt Vec$-morphism $\zvf\colon V\to W$ that satisfies~\eqref{CondSurj2} for $\zL_1$ and decomposable tensors of the type $\zvy_i\0 v_i$ ($i\in\{0,\dots,N\}$).

We first build $\phi(v_j)\in W_{\zg_j}$ linearly in $v_j\in V_{\zg_j}$ for an arbitrarily fixed $j\in\{0,\dots,N\}$. We~set again $\zvy_0=1\in\zL_{1,0}$. Since $\Phi_{\zL_1}(\zvy_j\0v_j)\in (\zL_1\0 W)_0$, it reads
\[
\Phi_{\zL_1}(\zvy_j\0v_j)=\sum_{i=0}^N\sum_{k=1}^{M_i} \zl_i^k\0 w_i^k,
\]
where $M_i\in \mathbb{N}$, $\zl_i^k\in\zL_{1,\zg_i}$ and $w_i^k\in W_{\zg_i}$. When setting
\[
\mathcal{A}_i=\bigg\{\za\in\mathbb{N}^{\times N}\colon \sum_{\ell=1}^N\za_\ell\zg_\ell=\zg_i\bigg\}
\]
and
\[
\zl_i^k=\sum_{\za\in\mathcal{A}_i}r_{\za,i}^k\,\zvy^\za\qquad\big(r_{\za,i}^k\in\R\big),
\]
where we used the standard multi-index notation, we get
\[
\Phi_{\zL_1}(\zvy_j\0v_j)=\sum_{i=0}^N\sum_{\za\in\mathcal{A}_i}\zvy^\za\0\bigg(\sum_{k=1}^{M_i}r_{\za,i}^k w_i^k\bigg)=:\sum_{i=0}^N\sum_{\za\in\mathcal{A}_i}\zvy^\za\0 w_{\za,i}\qquad
(w_{\za,i}\in W_{\zg_i}).
\]
Denoting the canonical basis of $\R^N$ by $(e_\ell)_\ell$ and decomposing the {\small RHS} with respect to the values of $|\za|=\za_1+\dots+\za_N\in\mathbb{N}$, we obtain
\begin{gather}\label{Defphi1}
\Phi_{\zL_1}(\zvy_j\0v_j)=w_{0,0}+\sum_{i=1}^N \zvy_i\0 w_{e_i,i}+\sum_{i=0}^N\sum_{\za\in\mathcal{A}_i\colon |\za|\ge 2}\zvy^\za\0 w_{\za,i}.
\end{gather}

Let now $\zf^*_{r_0}$ ($r_0\in\R$, $r_0>0$ and $r_0\neq 1$) be the $\Z_2^n$-algebra endomorphism of $\zL_1$ that is defined by $\zf^*_{r_0}(\zvy_k)=r_0\zvy_k$ if $k\ne 0$ and by $\zf^*_{r_0}(\zvy_0)=1$. It follows from the naturality condition
\begin{center}
\leavevmode
\begin{xy}
(0,20)*+{ {V}(\Lambda_1)}="a"; (40,20)*+{{W}(\Lambda_1)}="b";
(0,0)*+{{V}(\Lambda_1)}="c"; (40,0)*+{{W}(\Lambda_1)}="d";%
{\ar "a";"b"}?*!/_3mm/{\Phi_{\Lambda_1}};%
{\ar "a";"c"}?*!/^6mm/{{V}(\zf^*_{r_0})};%
{\ar "b";"d"}?*!/_6mm/{\ {W}(\zf^*_{r_0})};%
{\ar "c";"d"}?*!/^3mm/{\Phi_{\Lambda_1}};%
\end{xy}
\end{center}
that
\[
W(\zf^*_{r_0})(\Phi_{\zL_1}(\zvy_j\0v_j))=w_{0,0}+\sum_{i=1}^N \zvy_i\0(r_0 w_{e_i,i})+\sum_{i=0}^N\sum_{\za\in\mathcal{A}_i\colon |\za|\ge 2}\zvy^\za\0\big(r_0^{|\za|} w_{\za,i}\big)
\]
and
\[
\Phi_{\zL_1}\big(V(\zf^*_{r_0})(\zvy_j\0v_j)\big)=r_0^{1-\zd_{j0}}w_{0,0}+\!\sum_{i=1}^N \zvy_i\0\big(r_0^{1-\zd_{j0}} w_{e_i,i}\big)+\sum_{i=0}^N\sum_{\za\in\mathcal{A}_i\colon |\za|\ge 2}\!\!\!\!\!\!\!\zvy^\za\0 \big(r_0^{1-\zd_{j0}} w_{\za,i}\big),
\]
where $\zd_{j0}$ is the Kronecker symbol, coincide. As all the monomials in $\zvy$ in the {\small RHS}-s of the two last equations are different, we get,
\begin{enumerate}\itemsep=0pt
 \item[(1)] if $j\neq 0$: $w_{0,0}=0$ and $w_{\za,i}=0$, for all $i\in\{0,\dots,N\}$ and all $\za\in\mathcal{A}_i\colon |\za|\ge 2$, and,
 \item[(2)] if $j=0$: $w_{e_i,i}=0$, for all $i\in\{1,\dots,N\}$, and $w_{\za,i}=0$, for all $i\in\{0,\dots,N\}$ and all $\za\in\mathcal{A}_i\colon |\za|\ge 2$.
 \end{enumerate}
Equation~\eqref{Defphi1} thus yields
\begin{gather}\label{Defphi2}
\Phi_{\zL_1}(\zvy_j\0 v_j)=\sum_{i=1}^{N}\zvy_i\0 w_{e_i,i}\quad(j\neq 0)\qquad\text{and}\qquad \Phi_{\zL_1}(1\0 v_0)=w_{00} .
\end{gather}

If $j\neq 0$, a new application of naturality, now for the $\Z_2^n$-algebra endomorphism $\zf^*_{R_0}$ ($R_0\in\R$, $R_0\neq 1$) of $\zL_1$ that is defined by $\zf^*_{R_0}(\zvy_i)= R_0\zvy_i$ ($i\neq 0,i\neq j$), $\zf^*_{R_0}(\zvy_j)=\zvy_j$ and $\zf^*_{R_0}(\zvy_0)=1$, leads to
\[
\zvy_j\0 w_{e_j,j}+\sum_{i\neq j} \zvy_i\0(R_0w_{e_i,i})=\zvy_j\0 w_{e_j,j}+\sum_{i\neq j}\zvy_i\0 w_{e_i,i},
\]
so that
\begin{gather}
\label{Defphi3}
\Phi_{\zL_1}(\zvy_j\0 v_j)=\zvy_j\0 w_{e_j,j}\qquad(j\neq 0).
\end{gather}

The vectors $w_{00}\in W_{0}$ (see~\eqref{Defphi2}) and $w_{e_j,j}\in W_{\zg_j}$ ($j\neq 0$) (see~\eqref{Defphi3}) are well-defined and depend obviously linearly on $v_0$ and $v_j$, respectively. Hence, setting $\phi(v_0)=w_{00}$ and $\phi(v_j)=w_{e_j,j}$ ($j\neq 0$), we define a degree 0 linear map from $V$ to $W$. Moreover, since~\eqref{CondSurj2} is clearly satisfied for $\zL_1$ and the $\zvy_i\0 v_i$ ($i\in\{0,\dots,N\}$), it is satisfied for any $\zL$, which completes the proof of surjectivity.
\end{proof}

Since
\[
\cF\colon\ \Z_2^n{\tt Vec}\to {\tt Fun}_0\big(\Z_2^n{\tt Pts}^{\op{op}},\tt AMod\big)
\]
is fully faithful, it is essentially injective, i.e., it is injective on objects up to isomorphism. It~follows that $\Z_2^n{\tt Vec}$ can be viewed as a full subcategory of the target category of $\cF$.

The above considerations lead to the following definition.
\begin{Definition}%\label{def:RepAMod}
A functor
\[
\mathcal{V}\in{\tt Fun}_0\big(\Z_2^n\catname{Pts}^\textnormal{op}, \catname{AMod}\big)
\]
is said to be \emph{representable}, if there exists $V\in\Z_2^n\tt Vec$, such that $\cF(V)$ is naturally isomorphic to~$\mathcal{V}$.
\end{Definition}

As $\cF$ is essentially injective, a {\it representing object} $V$, if it exists, is unique up to isomorphism. We~therefore refer sometimes to $V$ as ``the'' representing object.

\subsection[Cartesian Z2n-graded vector spaces and Cartesian Z2n-manifolds]
{Cartesian $\boldsymbol{\Z_2^n}$-graded vector spaces and Cartesian $\boldsymbol{\Z_2^n}$-manifolds}
\newcommand{\mbf}{\mathbf}
\newcommand{\Ci}{C^\infty}
\newcommand{\iho}{\ul{\op{Hom}}}
In the literature, the space $\R^{p|\ul{q}}$ is viewed, either as the trivial $\Z_2^n$-manifold
\[
\R^{p|\ul{q}}=\big(\R^p,\Ci_{\R^p}[[\zx]]\big)
\]
with canonical $\z$-graded formal parameters $\zx$, or as the Cartesian $\Z_2^n$-graded vector space
\[
\mbf{R}^{p|\ul{q}}=\mbf{R}^p\oplus\bigoplus_{j=1}^N\mbf{R}^{q_j},
\]
where $\mbf{R}^p$ (resp., $\mbf{R}^{q_j}$) is the term of degree $\zg_0=0\in\Z_2^n$ (resp., $\zg_j\in\Z_2^n$). Observe that we use the notation $\mbf{R}^\bullet$ (resp., $\R^\bullet$), when $\R^\bullet$ is viewed as a vector space (resp., as a manifold). It~can happen that we write $\R^\bullet$ for both, the vector space and the manifold, however, in these cases, the meaning is clear from the context. Further, we set $q_0=p$, $\mathbf{q}=(q_0,q_1,\dots,q_N)$, and $|\mathbf{q}|=\sum_iq_i$. When embedding $\mbf{R}^{q_i}$ ($i\in\{0,\dots,N\}$) into $\mbf{R}^{p|\ul{q}}$, we identify each vector of the canonical basis of $\mbf{R}^{q_i}$ with the corresponding vector of the canonical basis of $\mbf{R}^{|\mathbf{q}|}$. We~denote this basis by
\[
\big(e^i_k\big)_{i,k}\qquad (i\in I=\{0,\dots,N\},\, k\in K_i=\{1,\dots,q_i\})
\]
and assign of course the degree $\zg_i$ to every vector $e^i_k$. We~can now write
\[
\mbf{R}^{p|\ul q}=\bigoplus_{i=0}^N\mbf{R}^{q_i}=\bigoplus_{(i,k)\in I\times K_i}\mbf{R}\, e^i_k.
\]
The dual space of $\mbf{R}^{p|\ul q}$ is defined by
\[
\big(\mbf{R}^{p|\ul q}\big)^\vee = \ul{\op{Hom}}\big(\mbf{R}^{p|\ul q},\mbf{R}\big)=\bigoplus_{i=0}^N\iho_{\zg_i}\big(\mbf{R}^{p|\ul q},\mbf{R}\big),
\]
where $\iho$ is the internal Hom of $\Z_2^n\tt Vec$, i.e., the $\Z_2^n$-graded vector space of all linear maps, and where $\iho_{\zg_i}$ is the vector space of all degree $\zg_i$ linear maps. We~sometimes write $\iho_{\Z_2^n\tt Vec}$ instead of $\iho$. The dual basis of $\big(e^i_k\big)_{i,k}$ is defined as usual by
\[
\ze^k_i\big(e^j_\ell\big)=\zd^j_i\zd^k_\ell,
\]
so that $\ze^k_i$ is a linear map of degree $\zg_i$ and
\[
\big(\mbf{R}^{p|\ul q}\big)^\vee=\bigoplus_{(i,k)\in I\times K_i}\mbf{R}\,\ze^k_i.
\]
Let us finally mention that any $\Z_2^n$-vector $x\in\mbf{R}^{p|\ul q}$ reads $x=\sum_{j,\ell}x^\ell_j e^j_\ell$ and that
\begin{gather}\label{Ass1}
\ze^k_i(x)=x^k_i,
\end{gather}
as usual.

Notice now that if $M$ is a smooth $m$-dimensional real manifold and $(U,\zf)$ is a chart of $M$, the coordinate map $\zf$ sends any point $x\in M$ to $\zf(x)=(x^1,\dots,x^m)\in\R^m$, so that
\begin{gather}\label{Ass2}
\zf^i(x)=x^i.
\end{gather}
Hence, what we refer to as coordinate function $x^i\in\Ci(U)$ is actually the function $\zf^i$. Equations~\eqref{Ass1} and~\eqref{Ass2} suggest to associate to any $\Z_2^n$-graded vector space $\mbf{R}^{p|\ul q}$ a $\Z_2^n$-manifold $\R^{p|\ul q}$ with coordinate functions $\ze^k_i$. In other words, the associated $p|\ul q$-dimensional $\Z_2^n$-manifold will be the locally $\Z_2^n$-ringed space
\[
\R^{p|\ul q}=\big(\R^p,\cO_{\R^{p|\ul q}}\big)=\big(\R^p, \Ci_{\R^p}[[\ze_1^1,\dots,\ze_N^{q_N}]]\big),
\]
where $\Ci_{\R^p}$ is the standard function sheaf of $\R^p$, where the degree $\zg_j$ linear maps $\ze_j^1,\dots,\ze_j^{q_j}$ ($j\in\{1,\dots,N\}$) are interpreted as coordinate functions or formal parameters of degree $\zg_j$, and where the degree $0$ linear maps $\ze_0^1,\dots,\ze_0^{p}$ are viewed as coordinates in $\R^{p}$. We~often set
\begin{gather}\label{Equiv}
\zx^\ell_j:=\ze^\ell_j\quad (j\neq 0)\qquad\text{and}\qquad x^\ell := \ze^\ell_0.
\end{gather}

\begin{Remark}
In the following, we denote the coordinates of $\R^{p|\ul q}$ by
\[
\big(x^\ell,\zx_j^\ell\big)=\big(x^a,\zx^A\big)=(u^{\mathfrak{a}}),
\]
if we wish to make a distinction between the coordinates of degree $0,\zg_1,\dots,\zg_N$, if we distinguish between zero degree coordinates and non-zero degree ones, or if we consider all coordinates together.
\end{Remark}

We refer to the category of $\Z_2^n$-graded vector spaces $\mbf{R}^{p|\ul q}$ ($p, q_1,\dots,q_N\in\mathbb N$) and degree 0 linear maps, as the category $\Z_2^n\tt CarVec$ of Cartesian $\Z_2^n$-vector spaces. As just mentioned, the interpretation of the dual basis as coordinates leads naturally to a map
\[
\mathcal{M}\colon\ \Z_2^n{\tt CarVec}\ni\mbf{R}^{p|\ul q}\mapsto \R^{p|\ul q}\in\Z_2^n\tt Man,
\]
where $\Z_2^n\tt Man$ is the category of $\Z_2^n$-manifolds and corresponding morphisms. This map can easily be extended to a functor. Indeed, if $\mbf{L}\colon \mbf{R}^{p|\ul q}\to \mbf{R}^{r|\ul s}$ is a morphism in $\Z_2^n\tt CarVec$ $\big($it is canonically represented by a block diagonal matrix $\mbf{L}\in\op{gl}\big(r|\ul{s}\times p|\ul{q},\mbf{R}\big)\big)$, its dual ($\Z_2^n$-transpose) $\mbf{L}^\vee\colon \big(\mbf{R}^{r|\ul s}\big)^\vee\to \big(\mbf{R}^{p|\ul q}\big)^\vee$ $\big($which is represented by the standard transpose $^t\mbf{L}\in\op{gl}\big(p|\ul{q}\times r|\ul{s},\mbf{R}\big)\big)$ is also a degree 0 linear map. If we set
\[
\mbf{L}=\big(\mbf{L}^{\ell i}_{i k}\big),
\]
where $i\in I$, $\ell\in\{1,\dots, s_i\}$ label the row and $i\in I$, $k\in\{1,\dots,q_i\}$ label the column, we get
\[
\mbf{L}^\vee\big(\ze'^\ell_i\big)=\sum_{k=1}^{q_i} \mbf{L}^{\ell i}_{i k}\ze^k_i,
\]
where $\big(\ze'^\ell_i\big)_{i,\ell}$ is the basis of $\big(\mbf{R}^{r|\ul s}\big)^\vee$. When using notation~\eqref{Equiv}, we obtain
\begin{gather}\label{Pull1}
L^*\big(x'^\ell\big):=\mbf{L}^\vee\big(x'^\ell\big)=\sum_{k=1}^{p} \mbf{L}^{\ell 0}_{0 k}\,x^k\in \cO_{\R^{p|\ul q}}^{0}(\R^p)\qquad(\ell\in\{1,\dots, r\})
\end{gather}
and
\begin{gather}\label{Pull2}
L^*\big(\zx'^{\ell}_j\big):=\mbf{L}^\vee\big(\zx'^\ell_j\big)=\sum_{k=1}^{q_j} \mbf{L}^{\ell j}_{j k}\,\zx^k_j\in\cO_{\R^{p|\ul q}}^{\zg_j}(\R^p)\qquad(j\neq 0,\,\ell\in\{1,\dots,s_j\}).
\end{gather}
These pullbacks define a $\Z_2^n$-morphism $L\colon \R^{p|\ul q}\to \R^{r|\ul s}$. This is the searched $\Z_2^n$-morphism $\cM(\mbf L)\colon \cM\big(\mbf{R}^{p|\ul q}\big)\to \cM\big(\mbf{R}^{r|\ul s}\big)$. Since $\cM(\mbf{L})$ is defined interpreting the standard transpose $^t\mbf{L}$ as pullback $(\cM(\mbf{L}))^*$ of coordinates, we have \[
(\cM(\mbf{M}\circ\mbf{L}))^*\simeq\, ^t\mbf{L}\circ\,^t\mbf{M}\simeq(\cM(\mbf{L}))^*\circ(\cM(\mbf{M}))^*=(\cM(\mbf{M})\circ\cM(\mbf{L}))^*,
\]
so that $\cM$ respects composition. Further, it obviously respects identities. Hence, we defined a~functor $\cM$.

\newcommand{\cV}{\mathcal{V}}

The pullbacks~\eqref{Pull1} and~\eqref{Pull2} are actually linear homogeneous $\z$-functions, i.e., homogeneous $\z$-functions in
\begin{gather}\label{LinFun}
\cO^{\op{lin}}_{\R^{p|\ul q}}(\R^p):=\bigg\{\sum_{k=1}^{p} r_k\,x^k+\sum_{j=1}^N\sum_{k=1}^{q_j} r^j_k\,\zx^k_j\colon r_k, r^j_k\in\R\bigg\}=\big(\mbf{R}^{p|\ul q}\big)^\vee\subset\cO_{\R^{p|\ul q}}(\R^{p}),
\end{gather}
where the last equality is obvious because of equation~\eqref{Equiv}. Hence, the functor $\cM$ is valued in~the subcategory $\Z_2^n\tt{CarMan}\subset\Z_2^n\tt Man$ of Cartesian $\Z_2^n$-manifolds $\R^{p|\ul q}$ ($p,q_1,\dots,q_N\in\mathbb N$) and~$\Z_2^n$-morphisms whose coordinate pullbacks are global linear functions of the source manifold that have the appropriate degree:
\[
\mathcal{M}\colon\ \Z_2^n{\tt CarVec}\to \Z_2^n\tt CarMan.
\]
The inverse ``vectorification functor'' $\cV$ of this ``manifoldification functor'' $\cM$ is readily defined: to a Cartesian $\Z_2^n$-manifold $\R^{p|\ul q}$ we associate the Cartesian $\Z_2^n$-vector space $\mbf{R}^{p|\ul q}$, and to a~linear $\Z_2^n$-morphism we associate the degree 0 linear map that is defined by the transpose of~the block diagonal matrix coming from the morphism's linear pullbacks. It~is obvious that $\cV\circ\cM=\cM\circ\cV=\op{id}$.

\begin{Proposition}\label{IsoCatCar}
We have an isomorphism of categories
\begin{gather}\label{FunEquivCar}
\mathcal{M}\colon\ \Z_2^n{\tt CarVec}\rightleftarrows\Z_2^n\tt CarMan\ \, {:}\,\cV
\end{gather}
between the full subcategory $\Z_2^n{\tt CarVec}\subset\Z_2^n\tt{Vec}$ of Cartesian $\Z_2^n$-vector spaces $\mbf{R}^{p|\ul q}$ and the subcategory $\Z_2^n{\tt Car}{\tt Man}\subset\Z_2^n{\tt Man}$ of Cartesian $\Z_2^n$-manifolds $\R^{p|\ul q}$ and $\Z_2^n$-morphisms with linear coordinate pullbacks.
\end{Proposition}

\begin{Remark} Let us stress that the $\Z_2^n$-vector space of linear $\z$-functions
\[
\big(\mbf{R}^{p|\ul q}\big)^\vee\simeq\cO_{\R^{p|\ul q}}^{\op{lin}}(\R^{p})\subset \cO_{\R^{p|\ul q}}(\R^p)
\]
is of course not an algebra. In the case $p=0$, we get
\[
\cO^{\op{lin}}_{\R^{0|\ul q}}(\{\star\})=\zL^{\op{lin}}\subset\cO_{\R^{0|\ul q}}(\{\star\})=\zL,
\]
where $\{\star\}$ denotes the 0-dimensional base manifold $\R^0$ of the $\z$-point $\R^{0|\ul q}$, where $\zL=\R[[\zvy^1_1,\dots$, $\zvy^{q_N}_N]]$ is the $\z$-Grassmann algebra that corresponds to $\R^{0|\ul q}$, and where $\zL^{\op{lin}}$ is the $\z$-vector space of homogeneous degree 1 polynomials in the $\zvy_1^1,\dots,\zvy_N^{q_N}$ (with vanishing term $\zL_0^{\op{lin}}$ of $\z$-degree zero).
\end{Remark}

We close this subsection with some observations regarding the functor of points.

The Yoneda functor of points of the category $\z\tt Man$ is the fully faithful embedding
\[
\mathcal{Y}\colon\ \z{\tt Man}\to {\tt Fun}\big(\z{\tt Man}^{\op{op}},{\tt Set}\big).
\]
In~\cite{Bruce:2019b}, we showed that $\mathcal Y$ remains fully faithful for appropriate restrictions of the source and target of the functor category, as well as of the {\it resulting} functor category. More precisely, we~proved that the functor
\begin{gather}\label{SVYoneda}
\mathcal{S}\colon\ \z{\tt Man}\to{\tt Fun}_0\big(\z{\tt Pts}^{\op{op}},\tt A(N)FM\big)
\end{gather}
is fully faithful. The category $\tt A(N)FM$ is the category of (nuclear) Fr\'echet manifolds over a (nuc\-lear) Fr\'echet algebra, and the functor category is the category of those functors that send a~$\z$-Grassmann algebra $\zL$ to a (nuclear) Fr\'echet $\zL_0$-manifold, and of those natural transformations that have $\zL_0$-smooth $\zL$-components. For any $M\in\z\tt Man$ and any $\R^{0|\ul m}\simeq\zL\in\z{\tt Pts}^{\op{op}}$, we have
\[
M(\zL):=\mathcal{S}(M)(\zL)=\mathcal{Y}(M)(\zL)=\op{Hom}_{\z\tt Man}\big(\R^{0|\ul m},M\big).
\]

On the other hand, the Yoneda functor of points of the category $\z\tt CarVec$ is the fully faithful embedding
\[
\ul\bullet\colon\ \z{\tt CarVec}\ni\mbf{R}^{p|\ul q}\mapsto \ul{\mbf{R}}^{p|\ul q}:=\op{Hom}_{\z\tt Vec}\big({-},\mbf{R}^{p|\ul q}\big)\in{\tt Fun}\big(\z{\tt CarVec}^{\op{op}},\tt Set\big).
\]
The value of this functor on $\mbf{R}^{0|\ul m}\simeq\R^{0|\ul m}\simeq\zL$, is the {\it subset}
\begin{align}
\ul{\mbf{R}}^{p|\ul q}(\zL)&=\op{Hom}_{\z\tt Vec}\big(\mbf{R}^{0|\ul m},\mbf{R}^{p|\ul q}\big)\simeq\op{Hom}_{\z\tt CarMan}\big(\R^{0|\ul m},\R^{p|\ul q}\big)\nonumber
\\
&\subset\R^{p|\ul q}(\zL)=\mathcal{S}\big(\R^{p|\ul q}\big)(\zL)
=\op{Hom}_{\z\tt Man}\big(\R^{0|\ul m},\R^{p|\ul q}\big) \nonumber \\
&\simeq\bigoplus_{i=0}^{N}\bigoplus_{k=1}^{q_i}\cO_{\R^{0|\ul m},\zg_i}(\{\star\})=\bigoplus_{i=0}^{N}\bigoplus_{k=1}^{q_i}\zL_{\zg_i} =\bigoplus_{i=0}^{N}\bigoplus_{k=1}^{q_i}\zL_{\zg_i}\0\mbf{R} e^i_k =\bigoplus_{i=0}^{N}\;\zL_{\zg_i}\!\0\bigoplus_{k=1}^{q_i}\mbf{R} e^i_k\nonumber
\\
&=
\big(\zL\0\mbf{R}^{p|\ul q}\big)_0=\cF\big(\mbf{R}^{p|\ul q}\big)(\zL)=\mbf{R}^{p|\ul q}(\zL)\in\zL_0\tt Mod.
\label{FunPtsYon1}
\end{align}
More precisely,
\begin{gather}\label{FunPtsYon2}
\ul{\mbf{R}}^{p|\ul q}(\zL)=\op{Hom}_{\z\tt Vec}\big(\mbf{R}^{0|\ul m},\mbf{R}^{p|\ul q}\big)\simeq \bigoplus_{i=0}^{N}\bigoplus_{k=1}^{q_i}\cO^{\op{lin}}_{\R^{0|\ul m},\zg_i}(\{\star\})=\big(\zL^{\op{lin}}\0\mbf{R}^{p|\ul q}\big)_0\in\tt Set.
\end{gather}
Remark that, if we denote the coordinates of $\R^{p|\ul q}$ compactly by $(u^{\mathfrak a})$, the bijection in equation~\eqref{FunPtsYon2} sends a degree 0 linear map $\mbf{L}$ to the linear pullbacks $L^*(u^{\mathfrak a})$ of the corresponding $\z$-morphism $L=\cM(\mbf{L})$.

\begin{Remark}
If we restrict the functor $\ul{\mbf{R}}^{p|\ul q}$ $\big($resp., the functor $\cF\big(\mbf{R}^{p|\ul q}\big)\big)$ from $\z{\tt CarVec}^{\op{op}}\simeq\z{\tt CarMan}^{\op{op}}$ (resp., from $\z\tt Pts^{\op{op}}$) to the joint subcategory $\z\tt Car Pts^{\op{op}}$ of $\z$-points and $\z$-morphisms with linear coordinate pullbacks, the restricted Hom functor $\ul{\mbf{R}}^{p|\ul q}$ is actually a subfunctor of the restricted tensor product functor $\cF\big(\mbf{R}^{p|\ul q}\big)$. This observation clarifies the relationship between the fully faithful ``functor of points'' $\cF(\bullet)(-)=\bullet(-)$ of the full subcategory $\z{\tt CarVec}\subset\z\tt Vec$ and its standard fully faithful Yoneda functor of points $\ul\bullet(-)$.
\end{Remark}

Indeed, we observed already that the values of the Hom functor on $\Z_2^n$-points are subsets of the values of the tensor product functor. Further, on morphisms, the values of $\op{Hom}\big({-},\mbf{R}^{p|\ul q}\big)$ are restrictions of the values of $\big({-}\0\mbf{R}^{p|\ul q}\big)_0$. Indeed, if
\[
\op{Hom}_{\z\tt CarPts}\big(\R^{0|\ul n},\R^{0|\ul m}\big)\ni L\simeq \mathcal{V}(L)=\mbf{L}\in\op{Hom}_{\z\tt Vec}\big(\mbf{R}^{0|\ul n},\mbf{R}^{0|\ul m}\big),
\]
the morphisms $\ul{\mbf{R}}^{p|\ul q}(\mbf{L})$ and $\cF\big(\mbf{R}^{p|\ul q}\big)(L)$ are defined on $\ul{\mbf{R}}^{p|\ul q}(\zL)$ and its {\it supset} $\,\cF\big(\mbf{R}^{p|\ul q}\big)(\zL)$, respectively. When interpreting an element $\mbf{K}$ of the first as an element of the second, we use the identifications
\[
\mbf{K}\simeq K\simeq (K^*(u^{\mathfrak a}))_{\mathfrak a}\in\bigoplus_{i=0}^N\bigoplus_{k=1}^{q_i}\zL_{\zg_i}.
\]
Similar identifications are of course required when $\zL$ is replaced by $\zL'$. We~thus get
\[
\ul{\mbf{R}}^{p|\ul q}(\mbf{L})(\mbf{K})=\op{Hom}_{\z\tt Vec}\big(\mbf{L},\mbf{R}^{p|\ul q}\big)(\mbf{K})=\mbf{K}\circ\mbf{L}\simeq K\circ L\simeq \big(L^*(K^*(u^{\mathfrak a}))\big)_{\mathfrak a}.
\]
On the other hand, we have
\[
\cF\big(\mbf{R}^{p|\ul q}\big)(L^*)\big((K^*(u^{\mathfrak{a}}))_{\mathfrak{a}}\big) =\big(L^*(K^*(u^{\mathfrak{a}}))\big)_{\mathfrak{a}},
\]
since
\[
\big(L^*\0\Id_{\mbf{R}^{p|\ul q}}\big)_0\simeq \bigoplus_{i=0}^N\bigoplus_{k=1}^{q_i}L^*,
\]
when read through the isomorphism
\[
\bigoplus_{i=0}^N\;\zL_{\zg_i}\!\0\bigoplus_{k=1}^{q_i}\mbf{R}\,e^i_k\simeq \bigoplus_{i=0}^N\bigoplus_{k=1}^{q_i}\zL_{\zg_i}.
\]
This completes the proof of the subfunctor-statement.

\subsection[Finite dimensional Z2n-graded vector spaces and linear Z2n-manifolds]
{Finite dimensional $\boldsymbol{\Z_2^n}$-graded vector spaces and linear $\boldsymbol{\Z_2^n}$-manifolds}

In this subsection, we extend equivalence~\eqref{FunEquivCar} in a coordinate-free way.

\newcommand{\bbN}{\mathbb{N}}

\subsubsection[Finite dimensional Z2n-graded vector spaces]
{Finite dimensional $\boldsymbol{\Z_2^n}$-graded vector spaces}

We focus on the full subcategory $\Z_2^n{\tt FinVec}\subset\Z_2^n\tt Vec$ of finite dimensional $\Z_2^n$-graded vector spaces, i.e., of $\Z_2^n$-vector spaces $V$ of finite dimension
\[
p|\ul q\qquad \big(p\in\bbN,\, \ul{q}=(q_1,\dots,q_N)\in\bbN^{\times N}\big).
\]
Clearly
\[
\Z_2^n{\tt CarVec}\subset\Z_2^n \tt FinVec
\]
is a full subcategory.

Above, we already used the canonical basis of $\mbf{R}^{p|\ul q}$, i.e., the basis
\[
e^i_k=\,^t(0\dots 0;\dots;0\dots 1\dots 0;\dots;0\dots 0),
\]
where $1$ sits in position $k$ of block $i$. If
\[
(b^i_k)_{i,k}\qquad (i\in I=\{0,\dots,N\},\, k\in K_i=\{1,\dots,q_i\},\, \deg(b^i_k)=\zg_i\in\Z_2^n)
\]
is a basis of $V$, the degree respecting linear map
\begin{gather*}%\label{CanIsoVs}
b\colon\ V\ni v=\sum_{i,k}v^k_ib^i_k\mapsto \sum_{i,k}v^k_ie^i_k=\,^t\big(v_0^1,\dots,v_N^{q_N}\big)=:v^I\in\mbf{R}^{p|\ul q}
\end{gather*}
maps a basis to a basis and is thus an isomorphism of $\Z_2^n$-vector spaces.

We already discussed extensively the functor of points $\cF= \cF(\bullet)(-)=\bullet(-)$ of $\z\tt Vec$. Since $\z\tt FinVec$ is a full subcategory of $\z\tt Vec$, the functor $\cF$ remains fully faithful when restricted to $\z\tt FinVec$:

\begin{Proposition}\label{FunPtsFinVecFF}
The functor of points $\cF\colon \z{\tt FinVec}\to{\tt Fun}_0\big(\z{\tt Pts}^{\op{op}},{\tt AMod}\big)$ of the category $\z{\tt FinVec}$ is fully faithful.
\end{Proposition}

\begin{Remark} Later on, we consider linear $\Z_2^n$-manifolds and denote them sometimes using the same letter $V$ as for $\Z_2^n$-vector spaces. We~often disambiguate the concept considered by writing~$\mbf{V}$ in the vector space case.
\end{Remark}

\newcommand{\sfl}{\textsf{L}}

\subsubsection[Linear Z2n-manifolds]{Linear $\boldsymbol{\Z_2^n}$-manifolds}

In this subsection, we investigate the category of linear $\z$-manifolds, linear $\z$-functions of its objects, as well as its functor of points.

\medskip\noindent
{\bf Linear $\boldsymbol{\z}$-manifolds and their morphisms.} A $\Z_2^n$-manifold of dimension $p|\ul q$ is a locally $\Z^n_2$-ringed space $M := \big(|M|, \cO_M\big)$ that is locally isomorphic to $\R^{p|\ul q}$.
\begin{Definition}%\label{def:LinMan}
A {\it linear $\Z_2^n$-manifold} of dimension $p|\ul q$ is a locally $\Z_2^n$-ringed space $\textsf{L}=\big(|\textsf{L}|,\cO_\sfl\big)$ that is globally isomorphic to $\R^{p|\ul{q}}$, i.e., it is a $\Z_2^n$-manifold such that there exists a $\Z_2^n$-diffeomorphism
\[
\rmh \colon\ \textsf{L} \longrightarrow \R^{p|\ul{q}}.
\]
The diffeomorphism $\rmh$ is referred to as a \emph{linear coordinate map} or a \emph{linear one-chart-atlas}.
\end{Definition}

We now mimic classical differential geometry and say that two linear one-chart-atlases are linearly compatible, if their union is a ``linear two-chart-atlas''. In other words:

\begin{Definition}
Two linear coordinate maps $\rmh_1, \rmh_2 \colon \sfl \rightarrow \R^{p|\ul{q}}$ are said to be \emph{linearly compa\-tible}, if the $\z$-morphisms
\[
h_2\circ h_1^{-1},\,\rmh_1 \circ \rmh^{-1}_2 \colon\ \R^{p|\ul{q}} \longrightarrow \R^{p|\ul{q}}
\]
have linear coordinate pullbacks, i.e., if they are $\z\tt CarMan$-morphisms.
\end{Definition}

Linear compatibility is an equivalence relation on linear one-chart-atlases. There is a \mbox{$1:1$} correspondence between equivalence classes of linear one-chart-atlases and maximal linear atla\-ses, i.e., the unions of all linear one-chart-atlases of an equivalence class. For simplicity, we refer to a maximal linear atlas as a \emph{linear atlas}.

Just as a classical smooth manifold is a set that admits an atlas, or, better, a set endowed with an equivalence class of atlases, a linear $\Z_2^n$-manifold is a locally $\z$-ringed space $\sfl$ equipped with a linear atlas $(\sfl,\rmh_\za)_\za$.

We continue working in analogy with differential geometry and define a linear $\Z_2^n$-morphism between linear $\z$-manifolds as a locally $\Z_2^n$-ringed space morphism, or, equivalently, a $\z$-morphism, with linear coordinate form:
\begin{Definition}\label{def:LinMorpTen}
Let $\sfl$ and $\sfl'$ be two linear $\z$-manifolds of dimension $p|\ul q$ and $r|\ul s$, respectively. A $\Z_2^n$-morphism $\zvf\colon \sfl\to \sfl'$ is {\it linear}, if there exist linear coordinate maps
\[
\rmh\colon\ \sfl\to\R^{p|\ul q}\qquad\text{and}\qquad\rmk\colon\ \sfl'\to\R^{r|\ul s}
\]
in the linear atlases of $\sfl$ and $\sfl'$, such that the $\z$-morphism
\[
\rmk \circ \phi \circ \rmh^{-1}\colon\ \mathbb{R}^{p|\ul{q}} \rightarrow \mathbb{R}^{r|\mathbf{s}}
\]
has linear coordinate pullbacks.
\end{Definition}

It follows that any linear coordinate map $\rm h$ of the linear atlas of a linear $\z$-manifold $\sf L$, {\it is}~a~{\it linear} $\z$-morphism between the linear $\z$-manifolds $\sf L$ and $\R^{p|\ul q}$. This justifies the name ``linear coordinate map''. Further, the inverse ${\rm h}^{-1}$ of $\rm h$ is a linear $\z$-morphism.

\begin{Proposition}\label{prop:IndCoordMap}
If $\zvf\colon \sf L\to \sf L'$ is a linear $\z$-morphism, then, for any linear coordinate maps $(\sf L,\rmh')$ and $(\sf L',\rmk')$ of the linear atlases of $\sf L$ and $\sf L'$, respectively, the $\z$-morphism $\rmk'\circ\zvf\circ\rmh'^{-1}$ has linear coordinate pullbacks.
\end{Proposition}
\begin{proof}
We use the notations of Definition~\ref{def:LinMorpTen} and Proposition~\ref{prop:IndCoordMap}. Since
\[
\rmk'\circ\zvf\circ\rmh'^{-1}=\big(\rmk'\circ\rmk^{-1}\big)\circ \big(\rmk\circ \zvf\circ \rmh^{-1}\big)\circ\big(\rmh\circ\rmh'^{-1}\big)
\]
and each parenthesis of the {\small RHS} has linear pullbacks, their composite has linear pullbacks as~well.
\end{proof}

{\sloppy\begin{Proposition} Linear $\Z_2^n$-manifolds and linear $\z$-morphisms form a subcategory $\Z_2^n\catname{LinMan}\subset\z\tt Man$ of the category of $\z$-manifolds. Further, Cartesian $\z$-manifolds and $\z$-morphisms with linear coordinate pullbacks form a full subcategory $\z{\tt CarMan}\subset\z{\tt LinMan}$.
\end{Proposition}

}

\begin{proof}
If $\zvf\colon \sfl\to \sfl'$ and $\psi\colon \sfl'\to \sfl''$ are linear $\z$-morphisms, the composite $\z$-morphism is linear as well. Indeed, if $\rmk\circ\zvf\circ\rmh^{-1}$ and $\rmq\circ\psi\circ\rmp^{-1}$ have linear pullbacks, then $\rmq\circ\psi\circ\rmk^{-1}$ has linear pullbacks and so has $\rmq\circ(\psi\circ \zvf)\circ\rmh^{-1}$. Further, for any linear $\z$-manifold $\sfl$, the $\z$-identity map $\id_\sfl$ is linear, as for any linear coordinate map $\rmh$, we have $\rmh\circ\id_\sfl\circ\rmh^{-1}=\id_{\R^{p|\ul q}}$. The second statement is obvious.\end{proof}

\medskip\noindent
{\bf Sheaf of linear $\boldsymbol{\z}$-functions.}

\begin{Definition}
Let $\textsf{L}\in\z\tt LinMan$ be of dimension $p|\ul q$ and let $|U|\subset|\sfl|$ be open. A $\z$-function $f\in\cO_\sfl(|U|)$ is a {\it linear $\z$-function}, if there exists a linear coordinate map $\rmh\colon \sfl\to \R^{p|\ul q}$, such that
\[
(\rmh^*)^{-1}(f)\in\cO^{\op{lin}}_{\R^{p|\ul q}}\big(|\rmh|(|U|)\big).
\]
We denote the subset of all linear $\z$-functions of $\cO_\sfl(|U|)$ by $\cO_\sfl^{\op{lin}}(|U|)$. \end{Definition}

The subset $\cO_{\R^{p|\ul q}}^{\op{lin}}\big(|\rmh|(|U|)\big)$ is defined in the obvious way. If $f\in\cO_\sfl^{\op{lin}}(|U|)$, then for any chart $(\sfl,\rmh')$ of the linear atlas of $\sfl$, we have
\[
(\rmh'^*)^{-1}(f)\in\cO^{\op{lin}}_{\R^{p|\ul q}}\big(|\rmh'|(|U|)\big).
\]
This follows from the equation
\[
(\rmh'^*)^{-1}(f)=\big(\rmh\circ\rmh'^{-1}\big)^*\big((\rmh^*)^{-1}(f)\big)
\]
and the compatibility of the two charts.

As $\cO_{\sfl}^{\op{lin}}(|U|)\subset\cO_{\sfl}(|U|)$ is visibly closed for linear combinations, it is a vector subspace of~$\cO_\sfl(|U|)$. Hence, the intersection
\[
\cO_{\sfl,\zg_i}^{\op{lin}}(|U|):=\cO_{\sfl}^{\op{lin}}(|U|)\cap\cO_{\sfl,\zg_i}(|U|)\subset\cO_\sfl(|U|)
\]
is also a vector subspace. We~thus get vector subspaces $\cO^{\op{lin}}_{\sfl,\zg_i}(|U|)\subset\cO^{\op{lin}}_\sfl(|U|)$, so their direct sum over $i$ is a vector subspace as well. Since any $f\in\cO^{\op{lin}}_\sfl(|U|)$ reads uniquely as
\[
f=\sum_{i=0}^Nf_i\qquad \big(f_i\in\cO_{\sfl,\zg_i}(|U|)\big),
\]
we get
\[
(\rmh^*)^{-1}f_0+\sum_j(\rmh^*)^{-1}f_j=(\rmh^*)^{-1}f=\sum_\ell r_\ell\, x^\ell+\sum_j\sum_\ell r^j_\ell\,\zx^\ell_j.
\]
As $(\rmh^*)^{-1}$ is $\z$-degree preserving, we find that $f_i\in\cO^{\op{lin}}_{\sfl,\zg_i}(|U|)$, so that
\[
\cO^{\op{lin}}_\sfl(|U|)=\bigoplus_i\cO^{\op{lin}}_{\sfl,\zg_i}(|U|)\in\z\tt Vec.
\]

\begin{Remark}\label{RemLinFct} Observe that:
\begin{enumerate}\itemsep=0pt
\item[$(i)$] For any open subset $|U|\subset|\sfl|$ and any linear coordinate map $\rmh\colon \sfl\to\R^{p|\ul q}$, the map
\[
\rmh^*\colon\ \cO^{\op{lin}}_{\R^{p|\ul q}}(|h|(|U|))\to \cO^{\op{lin}}_\sfl(|U|)
\]
is an isomorphism of $\z$-vector spaces of dimension $p|\ul q$.
\item[$(ii)$] The restriction maps and the gluing property of $\cO_\sfl$ endow $\cO^{\op{lin}}_\sfl$ with a sheaf of $\z$-vector spaces structure.
\item[$(iii)$] A $\z$-morphism $\zvf\colon \sfl\to\sfl'$ between linear $\z$-manifolds is itself linear, if and only if $\zvf^*$ is a degree respecting linear map
\[
\zvf^*\colon\ \cO^{\op{lin}}_{\sfl'}(|\sfl'|)\to\cO^{\op{lin}}_{\sfl}(|\sfl|).
\]
\end{enumerate}\end{Remark}

It is straightforward to check the first two statements. For the third one, let $\sfl$ (resp., $\sfl'$) be of dimension $p|\ul q$ (resp., $r|\ul s$) and denote the coordinates of the corresponding Cartesian $\z$-manifold by $u^{\mathfrak a}=\big(x^a,\zx^A\big)$ $\big($resp., $v^{\mathfrak b}=\big(y^b,\zh^B\big)\big)$. The morphism $\zvf$ is linear, if and only if there exist linear coordinate maps $(\sfl,\rmh)$ and $(\sfl',\rmk)$, such that $\rmk\circ\zvf\circ\,\rmh^{-1}$ has linear coordinate pullbacks, i.e., such that
\begin{gather}\label{LinA}
\big((\rmh^*)^{-1}\circ\zvf^*\circ\rmk^*\big)(v^{\mathfrak b})\in\cO^{\op{lin}}_{\R^{p|\ul q}}(\R^p).
\end{gather}
On the other hand, in view of the first item of the previous remark, the condition
\[
\zvf^*\big(\cO^{\op{lin}}_{\sfl'}(|\sfl'|)\big)\subset\cO^{\op{lin}}_{\sfl}(|\sfl|)
\]
of the third item is equivalent to asking that
\begin{gather}\label{LinB}
(\zvf^*\circ\rmk^*)\bigg(\sum_{\mathfrak{b}}r_{\mathfrak{b}}v^{\mathfrak{b}}\bigg)\in \rmh^*\big(\cO^{\op{lin}}_{\R^{p|\ul q}}(\R^p)\big).
\end{gather}
The conditions~\eqref{LinA} and~\eqref{LinB} are visibly equivalent.

\medskip\noindent
{\bf Functor of points of $\boldsymbol{\z{\tt LinMan}}$.} We~start with the following

\begin{Proposition}\label{FL0Mod}
For any linear $\z$-manifold $\sf{L}$ $($of dimension $p|\ul q)$ and any $\z$-Grassmann algebra $\zL\simeq \R^{0|\underline{m}}$, the set
\[
{\sf{L}}(\zL):=\op{Hom}_{\z {\tt Man}}\big(\R^{0|\underline{m}},{\sf L}\big)\simeq \op{Hom}_{\z\tt Alg}(\cO_{\sf L}(|{\sf L}|),\zL)
\]
of $\zL$-points of $\,\sf L$ admits a unique Fr\'echet $\Lambda_0$-module structure, such that, for any chart ${\rm h}\colon {\sf L}\to \R^{p|\ul q}$ of the linear atlas of $\,\sf{L}$, the induced map
\[
{\rm h}_\zL\colon\ {\sf{L}}(\zL)\ni {\rm x}^*\mapsto {\rm x}^*\circ{\rm h}^*\in\R^{p|\ul q}(\zL)
\]
is a Fr\'echet $\zL_0$-module isomorphism.
\end{Proposition}
The definition of the category $\tt FAMod$ of Fr\'echet modules over Fr\'echet algebras can be found in Appendix~\ref{app:CatModAlg}. In the preceding proposition, it is implicit that the (unital) Fr\'echet algebra morphism that is associated to ${\rm h}_\zL$ is $\id_{\zL_0}$.

\begin{proof}
Let $\Lambda \in \Z_2^n\catname{GrAlg}$. In view of the fundamental theorem of $\z$-morphisms, there is a~$1:1$ correspondence between the $\zL$-points $\rm{x}^*$ of $\R^{p|\ul{q}}$ and the $(p+|\ul{q}|)$-tuples
\[
\mathrm{x}^* \simeq \big( {x}^a_\Lambda , \, \zx^A_\Lambda\big)\in {\Lambda}_{0}^{\times p} \times \Lambda_{\gamma_1}^{\times q_1} \times \dots \times \Lambda_{\gamma_N}^{\times q_N}
\]
(we used this correspondence already in equation~\eqref{FunPtsYon1}). Indeed, the algebra $\zL$ is the $\Z_2^n$-commutative nuclear Fr\'echet $\R$-algebra of global $\Z_2^n$-functions of some $\R^{0|\ul{m}}$ (in particular, the degree zero term $\zL_0$ of $\zL$ is a commutative nuclear Fr\'echet algebra). Hence, all its homogeneous subspaces $\zL_{\zg_i}$ ($i\in\{0,\dots, N\}$, $\zg_0=0$) are nuclear Fr\'echet vector spaces. Since any product (resp., any countable product) of nuclear (resp., Fr\'echet) spaces is nuclear (resp., Fr\'echet), the set $\R^{p|\ul{q}}(\zL)$ of $\zL$-points of $\R^{p|\ul{q}}$ is a nuclear Fr\'echet space. The latter statements can be found in~\cite{Bruce:2018}. The Fr\'{e}chet $\Lambda_{0}$-module structure on $\R^{p|\ul{q}}(\zL)$ is then defined by
\begin{gather}\label{ActionCompWise}
 \triangleleft\colon\ \Lambda_{0} \times \R^{p|\ul{q}}(\Lambda)\ni (\mathrm{a}, \mathrm{x}^*)\mapsto {\rm a}\triangleleft{\rm x}^*:=\big(\mathrm{a} \cdot x_\Lambda^a, \, \mathrm{a} \cdot \zx_\Lambda^{A} \big)\in\R^{p|\ul{q}}(\Lambda) .
\end{gather}
Since this action (which is compatible with addition in $\zL_0$ and addition in $\R^{p|\ul q}(\zL)$) is defined using the continuous associative multiplication $\cdot\colon \zL_{\zg_i}\times\zL_{\zg_j}\to\zL_{\zg_i+\zg_j}$ of the Fr\'echet algebra $\zL$, it is (jointly) continuous.

We now define the $\zL_0$-module structure on ${\sf L}(\zL)$. Observe first that, for any chart map ${\rm h}\colon {\sf L}\rightleftarrows\R^{p|\ul q}\ {:}\, {\rm h}^{-1}$ of the linear atlas of $\sf L$, the induced maps ${\rm h}_\zL\colon {\sf L}(\zL)\rightleftarrows\R^{p|\ul q}(\zL)\ {:}\,\big({\rm h}^{-1}\big)_\zL$ are inverse maps: $\big({\rm h}^{-1}\big)_\zL=({\rm h}_\zL)^{-1}=:{\rm h}^{-1}_\zL$. For $K\in\mathbb{N}\setminus\{0\}$, $k\in\{1,\dots,K\}$, ${\rm a}^k\in\zL_0$, and ${\rm y}_k^*\in{\sf L}(\zL)$, we set
\[
\sum_k {\rm a}^k\star {\rm y}_k^*:={\rm h}_\zL^{-1}\bigg(\sum_k {\rm a}^k\triangleleft{\rm h}_\zL({\rm y}_k^*)\bigg)\in{\sf L}(\zL).
\]
This defines a $\zL_0$-module structure on ${\sf L}(\zL)$ that makes ${\rm h}_\zL$ a $\zL_0$-module isomorphism. The $\zL_0$-module structures ${\sf L}(\zL)_{\rm h}$ and ${\sf L}(\zL)_{\rm k}$ that are implemented by $\rm h$ and another chart $\rm k$ of the linear atlas, respectively, are related by the $\zL_0$-module isomorphism
\[
{\rm k}_\zL^{-1}\circ {\rm h}_\zL\colon\ {\sf L}(\zL)_{\rm h}\to {\sf L}(\zL)_{\rm k}.
\]
Hence, the $\zL_0$-module structure on ${\sf L}(\zL)$ is well-defined.

In order to get a Fr\'echet structure on the real vector space ${\sf L}(\zL)$ that we just defined, we need a countable and separating family of seminorms $(p_n)_{n\in\mathbb{N}}$, such that any sequence in ${\sf L}(\zL)$ that is Cauchy for every $p_n$, converges for every $p_n$ to a fixed vector (i.e., a vector that does not depend on $n$). We~define this family (of course) by transferring to ${\sf L}(\zL)$ the analogous family $(\zr_n)_{n\in\mathbb{N}}$ of the Fr\'echet vector space $\R^{p|\ul q}(\zL)$ (see~\cite[Theorem 14]{Bruce:2018}). In other words, for each ${\rm y}^*\in{\sf L}(\zL)$, we set
\[
p_n({\rm y}^*):=\zr_n({\rm h}_\zL({\rm y}^*))\in\R_+.
\]
It is straightforwardly checked that $(p_n)_{n\in\mathbb{N}}$ is a countable family of seminorms that has the required properties. Moreover, the vector space isomorphism ${\rm h}_\zL$ is an isomorphism of Fr\'echet vector spaces, i.e., a continuous linear map with a continuous inverse. We~show that ${\rm h}_\zL$ is continuous for the seminorm topologies implemented by the $p_n$ and the $\zr_n$, i.e., that, for all $n\in\mathbb{N}$, there exist $m\in\mathbb{N}$ and $C>0$, such that
\[
\zr_n({\rm h}_\zL({\rm y}^*))\le C\, p_m({\rm y}^*),
\]
for all ${\rm y}^*\in{\sf L}(\zL)$. This requirement is of course satisfied. Hence, the composite ${\rm k}_\zL^{-1}\circ\, {\rm h}_\zL$ of~isomorphisms of Fr\'echet spaces is an isomorphism of Fr\'echet spaces, so that the Fr\'echet space structure on ${\sf L}(\zL)$ is well-defined.

The $\zL_0$-module structure and the Fr\'echet vector space structure on ${\sf L}(\zL)$ combine into a~Fr\'echet $\zL_0$-module structure, if they are compatible, i.e., if the $\zL_0$-action \begin{gather}\label{L0AL}
\star\colon\ \zL_0\times {\sf L}(\zL)\ni({\rm a},{\rm y}^*)\mapsto {\rm h}_\zL^{-1}({\rm a}\triangleleft{\rm h}_\zL({\rm y}^*))\in {\sf L}(\zL)
\end{gather}
is continuous. The condition is obviously satisfied as this action is the composite of the continuous maps $\id\times {\rm h}_\zL$, $\triangleleft$ and ${\rm h}_\zL^{-1}$. Further, the map ${\rm h}_\zL$ is clearly a Fr\'echet $\zL_0$-module isomorphism, for any $\rm h$ in the linear atlas of $\sf L$.

There is obviously no other Fr\'echet $\zL_0$-module structure on ${\sf L}(\zL)$ with that property. Indeed, if there were, it would be isomorphic to the Fr\'echet $\zL_0$-module structure on $\R^{p|\ul q}(\zL)$, hence isomorphic to the Fr\'echet $\zL_0$-module structure that we just constructed. \end{proof}

In the following, we denote the $\zL_0$-action $\star$ on ${\sf L}(\zL)$ by simple juxtaposition, i.e., we write~${\rm a}{\rm y^*}$ instead of ${\rm a}\star{\rm y}^*$.

To proceed, we need some preparation.

Let ${\tt Fun}_0\big(\z{\tt Pts}^{\op{op}},{\tt FAMod}\big)$ be the category of functors $F$, whose values $F(\zL)$ are Fr\'echet $\zL_0$-modules, and of natural transformations $\zb$, whose $\zL$-components $\zb_\zL$ are continuous $\zL_0$-linear maps. We~already used above the category ${\tt Fun}_0\big(\z{\tt Pts}^{\op{op}},{\tt AFM}\big)$ of functors, whose values are Fr\'echet $\zL_0$-manifolds, and of natural transformations, whose components are $\zL_0$-smooth maps.

{\sloppy\begin{Proposition}%\label{CatAMod2}
The category ${\tt Fun}_0\big(\z{\tt Pts}^{\op{op}},{\tt FAMod}\big)$ is a subcategory of the category ${\tt Fun}_0\big(\z{\tt Pts}^{\op{op}},{\tt AFM}\big)$.
\end{Proposition}

}

\begin{proof}
Observe first that composition of natural transformations (resp., identities of functors) is (resp., are) induced by composition (resp., identities) in the target category of the functors considered, which is (resp., are) in both target categories the standard set-theoretical composition (resp., identities). Hence, composition and identities are the same in both functor categories. However, we still have to show that objects (resp., morphisms) of the first functor category are objects (resp., morphisms) of the second.

Let $F$ be a functor with target $\tt FAMod$. Since a Fr\'echet $\zL_0$-module (i.e., a Fr\'echet vector space with a (compatible) continuous $\zL_0$-action) is clearly a Fr\'echet $\zL_0$-manifold, the functor $F$ sends $\z$-Grassmann algebras $\zL$ to Fr\'echet $\zL_0$-manifolds $F(\zL)$. Let now $\zf^*\colon \zL\to \zL'$ be a morphism of $\z$-algebras. As $F(\zf^*)\colon F(\zL)\to F(\zL')$ is a morphism between Fr\'echet modules over the Fr\'echet algebras $\zL_0$ and $\zL'_0$, respectively, it is continuous and it has an associated continuous (unital, $\R$-) algebra morphism $\psi\colon \zL_0\to\zL'_0$, such that
\begin{gather}\label{QL0L}
F(\zf^*)({\rm a}{\rm v}+{\rm a}'{\rm v}')=\psi({\rm a})F(\zf^*)({\rm v})+\psi({\rm a}')F(\zf^*)({\rm v}'),
\end{gather}
for all ${\rm a,a'}\in\zL_0$ and all ${\rm v,v'}\in F(\zL)$. We must show that $F(\zf^*)$ is a morphism between Fr\'echet manifolds over $\zL_0$ and $\zL'_0$, respectively, i.e., we must show that $F(\zf^*)$ is smooth and has first order derivatives that are linear in the sense of~\eqref{QL0L} (see~\cite{Bruce:2019b}). Since, for any $t\in\R$, we have $\psi(t)=t\psi(1)=t$, it follows from~\eqref{QL0L} that
\[
{\rm d}_{\rm x}F(\zf^*)({\rm v}):=\lim_{t\to 0}\frac{1}{t}(F(\zf^*)({\rm x}+t{\rm v})-F(\zf^*)({\rm x}))=F(\zf^*)({\rm v})
\]
and
\[
{\rm d}^{k+1}_{\rm x}F(\zf^*)({\rm v}_1,\dots,{\rm v}_{k+1})=0,
\]
for any ${\rm x}, {\rm v}, {\rm v}_1,\dots, {\rm v}_{k+1}\in F(\zL)$ and any $k\ge 1$. Hence, all derivatives exist everywhere and are (jointly) continuous. This implies that $F(\zf^*)$ has the required properties, so that $F$ is a functor with target $\tt AFM$.

As for morphisms, let $\zh\colon F\to G$ be a natural transformation between functors valued in~$\tt FAMod$. Its $\zL$-components $\zh_\zL\colon F(\zL)\to G(\zL)$ are continuous and $\zL_0$-linear maps. Repeating the proof given in the preceding paragraph for $F(\zf^*)$, we obtain that $\zh_\zL$ is $\zL_0$-smooth, i.e., is smooth and has $\zL_0$-linear first order derivatives. Therefore, the morphism $\zh$ of the functor category with target $\tt FAMod$ is a morphism of the functor category with target $\tt AFM$.
\end{proof}

Since
\[
\z{\tt LinMan}\subset\z{\tt Man}\qquad\text{and}\qquad{\tt Fun}_0\big(\z{\tt Pts}^{\op{op}},{\tt FAMod}\big)\subset{\tt Fun}_0\big(\z{\tt Pts}^{\op{op}},{\tt AFM}\big)
\]
are subcategories, we expect that:

\begin{Proposition}
The functor
\begin{gather*}%\label{FuPtMa}
\mathcal{S}\colon\ \z{\tt Man}\to{\tt Fun}_0\big(\z{\tt Pts}^{\op{op}},\tt AFM\big)
\end{gather*}
$($see equation~\eqref{SVYoneda}$)$ restricts to a functor
\begin{gather*}
%\label{FunPtsLinMan}
\mathcal{S}\colon\ \z{\tt LinMan}\to{\tt Fun}_0\big(\z{\tt Pts}^{\op{op}},\tt FAMod\big).
\end{gather*}
\end{Proposition}

\begin{proof}
We have to explain why $\mathcal{S}$ sends linear $\z$-manifolds and linear $\z$-morphisms to objects and morphisms, respectively, of the target subcategory.

Let ${\sf L}\in\z{\tt LinMan}$. The functor $\mathcal{S}({\sf L})$ is an object of the functor category with target $\tt AFM$. Since composition and identities are the same in both target categories, it suffices to show that, for any $\z$-Grassmann algebra $\zL$, the value $\mathcal{S}({\sf L})(\zL)={\sf L}(\zL)$ is a Fr\'echet $\zL_0$-module and that, for any $\z$-algebra morphism $\zf^*\colon \zL\to \zL'$, the morphism
\[
{\sf L}(\zf^*)\colon\ {\sf L}(\zL)\ni {\rm y}^*\mapsto \zf^*\circ{\rm y}^*\in {\sf L}(\zL')
\]
is a morphism of the category $\tt FAMod$. The first of the preceding conditions holds in view of~Pro\-position~\ref{FL0Mod}. We~start proving the second condition for ${\sf L}=\R^{p|\ul q}$. Since $\R^{p|\ul q}(\zf^*)$ is a~mor\-phism of $\tt AFM$, it is smooth, hence, continuous. Further, omitting the summation symbols and using our standard notation, we get
\begin{align*}
\R^{p|\ul q}(\zf^*)\big({\rm a}^k \triangleleft{\rm x}_k^*\big)& =\R^{p|\ul q}(\zf^*)\big({\rm a}^k\cdot x_{\zL,k}^a,{\rm a}^k\cdot\zx_{\zL,k}^A\big)=\big(\zf^*\big({\rm a}^k\big)\cdot \zf^*(x_{\zL,k}^a),\zf^*\big({\rm a}^k\big)\cdot\zf^*\big(\zx_{\zL,k}^A\big)\big)\nonumber\\
 & = \zf^*\big({\rm a}^k\big)\,\triangleleft\;\R^{p|\ul q}(\zf^*)({\rm x}^*_k) .%\label{CatAMod3}
\end{align*}
It now suffices to recall that the $\z$-algebra morphism $\zf^*$ is the pullback $\zf^*_{\star}$ over the whole base manifold $\{\star\}$ of a $\z$-morphism $\zf\colon \R^{0|\ul m'}\to \R^{0|\ul m}$, and that all pullbacks of $\z$-morphisms are continuous, so that the restriction $\zf^*\colon \zL_0\to\zL'_0$ is a continuous algebra morphism. We~are now able to prove that the second condition holds also for an arbitrary linear $\z$-manifold $\sf L$. Indeed, since $\zf^*\colon \zL\to \zL'$ is a morphism of $\z$-algebras, the map ${\sf L}(\zf^*)\colon {\sf L}(\zL)\to {\sf L}(\zL')$ is a morphism of~$\tt AFM$, hence, it is continuous. Recall now that any chart ${\rm h}\colon {\sf L}\to \R^{p|\ul q}$ is a $\z$-morphism, so that $\mathcal{S}({\rm h})\colon {\sf L}(-)\to \R^{p|\ul q}(-)$ is a natural transformation ${\rm h}_-$ with $\zL$-components ${\rm h}_\zL\colon {\sf L}(\zL)\to\R^{p|\ul q}(\zL)$ that are Fr\'echet $\zL_0$-module isomorphisms in view of Proposition~\ref{FL0Mod}. Naturality of ${\rm h}_-$ implies that
\[
{\rm h}_{\zL'}\circ{\sf L}(\zf^*)=\R^{p|\ul q}(\zf^*)\circ {\rm h}_\zL,
\]
and, due to invertibility, that
\[
{\sf L}(\zf^*)={\rm h}_{\zL'}^{-1}\circ\R^{p|\ul q}(\zf^*)\circ {\rm h}_\zL.
\]
Definition~\eqref{L0AL} yields
\[
{\sf L}(\zf^*)\big({\rm a}^k\star {\rm y}^*_k\big)=\big({\rm h}_{\zL'}^{-1}\circ\R^{p|\ul q}(\zf^*)\circ {\rm h}_\zL\big)\big({\rm h}_\zL^{-1}\big({\rm a}^k\triangleleft {\rm h}_\zL({\rm y}^*_k)\big)\big)=\zf^*\big({\rm a}^k\big)\star{\sf L}(\zf^*)({\rm y}^*_k)
\]
(we used our standard notation). Hence, the functor $\mathcal{S}({\sf L})$ is an object of the functor category with target $\tt FAMod$.

As for morphisms, we consider a linear $\z$-morphism
\[
\zvf\colon\ {\sf L}\to{\sf L}'
\]
and will prove that $\mathcal{S}(\zvf)$, which is a natural transformation $\zvf_-$ of the functor category with target $\tt AFM$, i.e., a natural transformation with $\zL_0$-smooth $\zL$-components $\zvf_\zL$, has actually continuous (but this results from $\zL_0$-smoothness) $\zL_0$-linear components.

Let $p|\ul q$ (resp., $r|\ul s$) be the dimension of $\sf L$ (resp., of $\sf L'$). We~first discuss the case of a linear $\z$-morphism
\[
\Phi\colon\ \R^{p|\ul q}\to \R^{r|\ul s}
\]
between the corresponding Cartesian $\z$-manifolds with canonical coordinates $\big(x^a,\zx^A\big)$ and $\big(y^b,\zh^B\big)$, respectively. We~know from~\cite{Bruce:2019b} that, if the $\z$-morphism (resp., the linear $\z$-morphism) $\Phi$ reads
\begin{subequations}
\begin{align}
& \Phi^*\big(y^b\big)=\sum_{|\za|\ge 0}\Phi_\za^b(x)\, \zx^\za\qquad
\bigg(\text{resp.,}\;=\sum_{a}{\mbf L}^b_a x^a\bigg),\label{coeffy}
\\
& \Phi^*\big(\zh^B\big)=\sum_{|\za|> 0}\Phi_\za^B(x)\, \zx^\za\qquad
\bigg(\text{resp.,}\;=\sum_A {\mbf L}^B_A\zx^A\bigg)\label{coeffh}
\end{align}
\end{subequations}
(where the right-hand sides have the appropriate degree and where the coefficients $\mbf L^{\ast}_\ast$ are real numbers), then the $\zL$-component $\Phi_\zL$ associates to the $\zL$-point ${\rm x}^*\simeq \big(x^a_\Lambda;\zx^A_\Lambda\big) =\big(x^a_{||},\mathring{x}_\Lambda^a;\zx^A_\Lambda \big)$ of~$\R^{p|\ul{q}}(\zL)$, the $\zL$-point ${\rm x}^*\circ\Phi^*\simeq \big(y_\zL^b;\zh_\zL^B\big)$ of $\R^{r|\ul{s}}(\zL)$ that is given by
\begin{subequations}
\begin{align}
& y^b_\Lambda = \sum_{|\za|\ge 0}\sum_{|\zb|\ge 0}\frac{1}{\zb!}\,\big(\partial_{x}^\zb\Phi_\za^b\big)\big(x_{||}\big)\,\mathring{x}_\zL^\zb\,\zx_\zL^\za\qquad
\bigg(\text{resp.,}\;=\sum_{a}{\mbf L}^b_a x_\zL^a\bigg), \label{eqn:smoothzero}
\\
& \zh^B_\Lambda =\sum_{|\za|> 0}\sum_{|\zb|\ge 0}\frac{1}{\zb!}\,\big(\partial_x^\zb\Phi_\za^B\big)\big(x_{||}\big)\,\mathring{x}_\zL^\zb\,\zx_\zL^\za\qquad
\bigg(\text{resp.,}\;=\sum_A {\mbf L}^B_A\zx_\zL^A\bigg). \label{eqn:smoothnonzero}
\end{align}
\end{subequations}
Here, we used the obvious decomposition $\Lambda = \R \times \mathring{\Lambda}$ and wrote $x^a_\Lambda = (x^a_{||}, \mathring{x}^a_{\Lambda})$. The particular linear versions of equations~\eqref{eqn:smoothzero} and~\eqref{eqn:smoothnonzero} (in parentheses), show that the component $\Phi_\zL$ is $\zL_0$-linear, as needed.

In the general case of a linear $\z$-morphism $\phi\colon {\sf L}\to{\sf L}'$, the $\z$-morphism $\Phi:={\rm k}\circ\phi\circ{\rm h}^{-1}\colon \R^{p|\ul q}\to\R^{r|\ul s}$ has linear coordinate pullbacks $\Phi^*\big(y^b\big)$ and $\Phi^*\big(\zh^B\big)$ (and is thus a linear $\z$-morphism), for any charts $\rm h$ and $\rm k$ of $\sf L$ and ${\sf L}'$, respectively. Since $\zvf={\rm k}^{-1}\circ\Phi\circ{\rm h}$, we have $\zvf_\zL={\rm k}_\zL^{-1}\circ\Phi_\zL\circ{\rm h}_\zL$ and, in view of Proposition~\ref{FL0Mod} and the result of the preceding paragraph, all three factors of the {\small RHS} are $\zL_0$-linear.

Finally, the natural transformation $\mathcal{S}(\zvf)$ is a natural transformation of the functor category with target $\tt FAMod$.\end{proof}

\begin{Theorem}\label{FunPtsLinManFF}
The functor of points
\[
\mathcal{S}\colon\ \z{\tt LinMan}\to{\tt Fun}_0\big(\z{\tt Pts}^{\op{op}},\tt FAMod\big)
\]
of the category $\z\tt LinMan$ is fully faithful.
\end{Theorem}

\begin{proof} We need to prove that the map
\[
\mathcal{S}_{\sf L,L'}\colon\ \op{Hom}_{\z\tt LinMan}({\sf L,L'})\ni\phi\mapsto \zvf_-\in\op{Hom}_{{\tt Fun}_0(\z{\tt Pts}^{\op{op}},\tt FAMod)}({\sf L}(-),{\sf L'}(-))
\]
is bijective, for all linear $\z$-manifolds $\sf L,\ L'$.

Since $\mathcal{S}$ is the restriction of the fully faithful functor $\mathcal{S}\colon {\z\tt Man}\to{\tt Fun}_0\big(\z{\tt Pts}^{\op{op}},{\tt AFM}\big)$, the map $\mathcal{S}_{\sf L,L'}$ is injective.

To prove that $\mathcal{S}_{\sf L,L'}$ is also surjective, it actually suffices to show that the property holds for Cartesian $\z$-manifolds. Indeed, in this case, if $\zh\colon {\sf L}(-)\to{\sf L}'(-)$ is a natural transformation of ${\tt Fun}_0\big(\z{\tt Pts}^{\op{op}},\tt FAMod\big)$, then $ {\rm k}_-\circ \zh\circ{\rm h}^{-1}_-$ is a natural transformation in the same category from $\R^{p|\ul q}(-)$ to $\R^{r|\ul s}(-)$, and this transformation is implemented by a linear $\z$-morphism $\zf\colon \R^{p|\ul q}\to \R^{r|\ul s}$. It~follows that
\[
\zh={\rm k}_-^{-1}\circ\zf_-\circ{\rm h}_-=({\rm k}^{-1}\circ\zf\circ{\rm h})_-,
\]
where the latter composite is a linear $\z$-morphism.

Let now $H\colon \R^{p|\ul q}(-)\to \R^{r|\ul s}(-)$ be a natural transformation of ${\tt Fun}_0\big(\z{\tt Pts}^{\op{op}},\tt FAMod\big)$, hence, a natural transformation of ${\tt Fun}_0\big(\z{\tt Pts}^{\op{op}},\tt AFM\big)$. We~know from~\cite{Bruce:2019b} that $H$ is implemented by a~$\z$-morphism $\Phi\colon \R^{p|\ul q}\to \R^{r|\ul s}$, but we still have to prove that this morphism is linear. It~follows from equations~\eqref{eqn:smoothzero} and~\eqref{eqn:smoothnonzero} that $H_\zL=\Phi_\zL$ is given by
\begin{subequations}
\begin{align}
& y^b_\Lambda = \sum_{|\za|\ge 0}\sum_{|\zb|\ge 0}F^b_{\za\zb}\big(x_{||}\big)\,\mathring{x}_\zL^\zb\,\zx_\zL^\za, \label{eqn:smoothzero2}\\
& \zh^B_\Lambda =\sum_{|\za|> 0}\sum_{|\zb|\ge 0}F^B_{\za\zb}\big(x_{||}\big)\,\mathring{x}_\zL^\zb\,\zx_\zL^\za, \label{eqn:smoothnonzero2}
\end{align}
\end{subequations}
where we set
\begin{gather}\label{ExplForm}
F^\ast_{\za\zb}(x):=\frac{1}{\zb !}\partial_x^\zb\Phi_\za^\ast\in\Ci(\R^p)
\end{gather}
(the $\Phi^*_\za\in\Ci(\R^p)$ are the coefficients of the coordinate pullbacks by $\Phi$, see equations~\eqref{coeffy} and~\eqref{coeffh}), and where the {\small RHS}-s have of course the same $\z$-degree as the corresponding coordinates of $\R^{r|\ul s}$. Since $H_\zL$ is $\zL_0$-linear, we have
\[
\sum_{\za}\sum_{\zb}F^\ast_{\za\zb}\big(r\, x_{||}\big)\,r^{|\za|+|\zb|}\mathring{x}_\zL^\zb\,\zx_\zL^\za=r\sum_{\za}\sum_{\zb}F^\ast_{\za\zb}\big(x_{||}\big)\,\mathring{x}_\zL^\zb\,\zx_\zL^\za,
\]
i.e.,
\[
r^{|\za|+|\zb|}F^\ast_{\za\zb}\big(r\, x_{||}\big)=r\, F^\ast_{\za\zb}\big(x_{||}\big) ,
\]
for any $r\in\R_{>0}\subset\zL_0$, any $\za$, $\zb$ and for any $x_{||}\in\R^p$. When deriving with respect to $r$, we obtain
\[
r^{|\za|+|\zb|-1}\bigg((|\za|+|\zb|)F^\ast_{\za\zb}\big(r x_{||}\big)+r\sum_{a=1}^{p}x_{||}^a\big(\partial_{x^a_{||}} F^*_{\za\zb}\big)\big(r x_{||}\big)\bigg)=F_{\za\zb}^\ast\big(x_{||}\big),
\]
so that setting $r=1$ yields
{\samepage\begin{gather}\label{EVF}
\sum_{a=1}^{p}x_{||}^a\partial_{x^a_{||}} F^*_{\za\zb}=(1-n)F_{\za\zb}^\ast\big(x_{||}\big)\qquad
(n:=|\za|+|\zb|\in\mathbb{N}),
\end{gather}
again for all $\za$, $\zb$ and all $x_{||}\in\R^p$.

}

Recall now that Euler's homogeneous function theorem states that, if $F\in C^1(\R^p\setminus\{0\})$, then, for any $\zn\in\R$, we have
\[
\sum_{a=1}^{p}x^a\partial_{x^a}F=\zn F(x),\qquad\forall x\in\R^p\setminus\{0\}
\]
is equivalent to
\[
F(r x)=r^\zn F(x),\qquad\forall r>0,\quad\forall x\in\R^p\setminus\{0\}.
\]
In view of~\eqref{EVF}, we thus get
\begin{gather}\label{Homogeneity}
F_{\za\zb}^\ast\big(rx_{||}\big)=r^{1-n}F_{\za\zb}^\ast\big(x_{||}\big),\qquad
\forall r>0,\quad \forall x_{||}\in\R^p,
\end{gather}
 where we could extend the equality from $\R^p\setminus\{0\}$ to $\R^p$ due to continuity. If $r$ tends to $0^+$, the limit of the {\small LHS} is $F_{\za\zb}^\ast(0)\in\R$ and, for $n=0$ (resp., $n=1$; resp., $n\ge 2$), the limit of the {\small RHS} is $0$ $\big($resp., $F_{\za\zb}^\ast\big(x_{||}\big)$; resp., $+\infty\cdot F_{\za\zb}^\ast\big(x_{||}\big)\big)$.

In the case $n\ge 2$, we conclude that
\begin{gather}\label{LinForm1}
F_{\za\zb}^*\big(x_{||}\big)=0,\qquad
\forall x_{||}\in\R^p,\quad \forall \za,\zb\colon\ |\za|+|\zb|\ge 2.
\end{gather}

For $n=0$, we get
\[
F^*_{00}(0)=0.
\]
Observe that $\za=\zb=0$ is only possible in equation~\eqref{eqn:smoothzero2}. Differentiating~\eqref{Homogeneity}, in the case $n=0$, with respect to any component $x_{||}^a$ of $x_{||}$ and simplifying by $r$, we obtain
\[
\big(\partial_{x^a_{||}}F_{00}^b\big)\big(rx_{||}\big)=\partial_{x_{||}^a}F^b_{00}\big(x_{||}\big),
\]
and taking the limit $r\to 0^+$, we get
\[
\partial_{x_{||}^a}F^b_{00}\big(x_{||}\big)=\partial_{x_{||}^a}F^b_{00}(0)=:\mbf{L}^b_a\in\R.
\]
Integration yields
\begin{gather}\label{LinForm2}
F^b_{00}\big(x_{||}\big)=\sum_a \mbf{L}^b_a x_{||}^a,\qquad
\forall x_{||}\in\R^p,\quad \forall b,
\end{gather}
as $F_{00}^b(0)=0$.

In the remaining case $n=|\za|+|\zb|=1$, we have necessarily $\za=0$ and $\zb=e_a$, or $\za=e_A$ and $\zb=0$ (the $e_*$ are of course the vectors of the canonical basis of $\R^p$ and $\R^{|\ul q|}$, respectively). For $\z$-degree reasons, the first (resp., second) possibility is incompatible with equation~\eqref{eqn:smoothnonzero2} (resp., equation~\eqref{eqn:smoothzero2}). Hence, the only terms in~\eqref{eqn:smoothzero2} that still need being investigated are the terms $(\za,\zb)=(0,e_a)$. It~follows from equation~\eqref{Homogeneity} and its limit $r\to 0^+$ (see above) that $F_{0\,e_a}^b\big(x_{||}\big)=\mbf{K}^b_a$, where we set $\mbf{K}^b_a:=F^b_{0\,e_a}(0)\in\R$. However, equations~\eqref{ExplForm} and~\eqref{LinForm2} imply that
\[
\mbf{K}^b_a=F^b_{0\,e_a}\big(x_{||}\big)=\partial_{x_{||}^a}F^b_{00}\big(x_{||}\big)=\mbf{L}^b_a,
\]
so that
\begin{gather}\label{LinForm3}
F_{0\,e_a}^b\big(x_{||}\big)=\mbf{L}^b_a\in\R,\qquad
\forall x_{||}\in\R^p, \quad \forall a,b.
\end{gather}
In equation~\eqref{eqn:smoothnonzero2}, the only terms that still need being investigated are the terms $(\za,\zb)=(e_A,0)$. Using again the limit $r\to 0^+$ of equation~\eqref{Homogeneity}, we find
\begin{gather}\label{LinForm4}
F^B_{e_A 0}\big(x_{||}\big)=\mbf{L}^B_A,\qquad\forall x_{||}\in\R^p,\quad\forall A,B,
\end{gather}
where we wrote $\mbf{L}^B_A$ instead of $F^B_{e_A 0}(0)$.

When combining now the results of equations~\eqref{LinForm1}, \eqref{LinForm2}, \eqref{LinForm3}, and~\eqref{LinForm4}, we see that equations~\eqref{eqn:smoothzero2} and~\eqref{eqn:smoothnonzero2} reduce to \begin{gather*}
y^b_\Lambda = \sum_a\mbf{L}^b_a\, \big(x_{||}^a+\mathring{x}_{\zL}^a\big)\qquad\text{and}\qquad
\zh^B_\Lambda =\sum_A\mbf{L}^B_A \, \zx_\zL^A\;
%\label{eqn:lin}
\end{gather*}
and that the $\z$-morphism $\Phi$ that induces the natural transformation $H$ is defined by the coordinate pullbacks
\[
\Phi^*\big(y^b\big)= \sum_a\mbf{L}^b_a\, x^a\qquad\text{and}\qquad \Phi^*\big(\zh^B\big)=\sum_A\mbf{L}^B_A \, \zx^A,
\]
 i.e., that $\Phi$ is linear (see~\eqref{coeffy},~\eqref{coeffh},~\eqref{eqn:smoothzero}, and~\eqref{eqn:smoothnonzero}).\end{proof}

\subsubsection[Isomorphism between finite dimensional Z2n-graded vector spaces and linear Z2n-manifolds]
{Isomorphism between finite dimensional $\boldsymbol{\Z_2^n}$-graded vector spaces\\ and linear $\boldsymbol{\Z_2^n}$-manifolds}

In this subsection, we extend the isomorphism
\[
\cM\colon\ \Z_2^n{\tt CarVec}\rightleftharpoons\Z_2^n{\tt CarMan}\ \,{:}\, \cV
\]
of Proposition~\ref{IsoCatCar} between the full subcategories $\Z_2^n{\tt CarVec}\subset\Z_2^n{\tt FinVec}$ and $\Z_2^n{\tt CarMan}\subset\Z_2^n{\tt LinMan}$, to an isomorphism
\[
\mathcal{M}\colon\ \Z_2^n{\tt FinVec}\rightleftharpoons\Z_2^n{\tt LinMan}\ \,{:}\, \mathcal{V}.
\]

\noindent\textbf{$\boldsymbol{\Z_2^n}$-symmetric tensor algebra.} We~start with some remarks on tensor and $\Z_2^n$-symmetric tensor algebras over a (finite dimensional) $\Z_2^n$-vector space (see~\cite{Manin1988} and~\cite{BP13}).

Let
\[
V=\bigoplus_{i=0}^NV_{i}:=\bigoplus_{i=0}^NV_{\zg_i}\in\Z_2^n{\tt FinVec}
\]
be of dimension $p|\ul{q}$. The $\Z_2^n$-symmetric tensor algebra of $V$ is defined exactly as in the non-graded case, as the quotient of the $\Z_2^n$-graded associative unital tensor algebra of $V$ by the homogeneous ideal
\[
\bar{I}=\big(v_i \otimes v_j - (-1)^{\langle \zg_i, \zg_j \rangle } v_j \otimes v_i\colon v_i \in V_i,\, v_j \in V_j \big).
\]

More precisely, for $k\ge 2$, we have
\begin{gather*}%\label{UsefulDecomp1}
V^{\0 k}=\bigoplus_{i_1,\dots, i_k=0}^NV_{i_1}\0\cdots\0V_{i_k}=\bigoplus_{i_1\le\cdots\le i_k}V_{i_1,\dots,i_k}:=\bigoplus_{i_1\le\cdots\le i_k}\bigg(\bigoplus_{\zs\in\op{Perm}}V_{\zs_{i_1}}\0\cdots\0V_{\zs_{i_k}}\bigg),
\end{gather*}
where Perm is the set of all permutations of $i_1\le\cdots\le i_k$. For instance, if $n=1$, i.e., in the standard super case, the space $V^{\0 3}$ is the direct sum of the tensor products whose three factors have the subscripts $000$, $001$, $010$, $011$, $100$, $101$, $110$, $111$. The notation we just introduced means that we write
\begin{gather*}%\label{UsefulDecomp2}
V^{\0 3}=V_{000}\oplus V_{001}\oplus V_{011}\oplus V_{111},
\end{gather*}
where we used the lexicographical order and where
\[
V_{000}=V_0\0V_0\0V_0,\qquad
V_{001}=V_0\0V_0\0V_1\oplus V_0\0V_1\0V_0\oplus V_1\0V_0\0V_0,\quad
\text{et cetera}.
\]

Further, as we are dealing with formal power series in this paper, we define the $\Z_2^n$-graded tensor algebra of $V$ by
\[
\overline{T}V:=\Pi_k V^{\0 k},
\]
where $\Pi_k$ means that we consider not only finite sums of tensors of different tensor degrees, but full sequences of such tensors. The vector space structure on such sequences is obvious and~the algebra structure is defined exactly as in the standard case. Indeed, for $T^k\in V^{\0 k}$ and $U^\ell\in V^{\0 \ell}$, we have $T^k\0 U^\ell\in V^{\0(k+\ell)}$ and we just extend this tensor product by linearity. In~other words,~if
\[
T=\sum_{k=0}^\infty T^k\in\overline{T}V\qquad\text{and}\qquad U=\sum_{\ell=0}^{\infty}U^\ell\in\overline{T}V,
\]
we set
\begin{gather}\label{TenMult}
T\0 U=\sum_k\sum_\ell T^k\0 U^\ell=\sum_m\sum_{k+\ell=m}T^k\0 U^\ell\in \overline{T}V.
\end{gather}
It is clear that the just defined tensor multiplication endows $\overline{T}V$ with a $\Z_2^n$-graded algebra structure. Indeed, since
\[
V^{\0 k}=\bigoplus_{i_1\le\cdots\le i_k}V_{i_1,\dots, i_k}=\bigoplus_{p=0}^N\!\bigoplus_{
\substack{i_1\le\cdots\le i_k \\ \sum_j\zg_{i_j}=\zg_p}}V_{i_1,\dots, i_k}=:\bigoplus_{p=0}^N\,\big(V^{\0 k}\big)_p
\]
is visibly a $\Z_2^n$-graded vector space, the space $\overline{T}V$ is itself $\Z_2^n$-graded:
\[
\overline{T}V=\Pi_k\bigoplus_{p=0}^N\,\big(V^{\0 k}\big)_p=\bigoplus_{p=0}^N\Pi_k\,\big(V^{\0 k}\big)_p=:\bigoplus_{p=0}^N\,\big(\overline{T}V\big)_p.
\]
Now, if $T\in\big(\overline{T}V\big)_p$ and $U\in\big(\overline{T}V\big)_q$, we have $T^k\in\big(V^{\0 k}\big)_p$ and $U^\ell\in\big(V^{\0 \ell}\big)_q$, so that $T\0 U\in\big(\overline{T}V\big)_{p+q}$ (where $p+q$ means $\zg_p+\zg_q$), which shows that $\overline{T}V$ is a $\Z_2^n$-graded (associative unital) algebra (over $\R$), as announced.

The ideal $\bar I$ is homogeneous with respect to the decomposition
\[
\overline{T}V=\Pi_k\bigoplus_{i_1\le\cdots\le i_k}V_{i_1,\dots, i_k},\qquad\text{i.e.,}\qquad \bar I=\Pi_{k(\ge 2)}\bigoplus_{i_1\le\cdots\le i_k}\big(V_{i_1,\dots, i_k}\cap \bar I\big).
\]
Therefore, the $\Z_2^n$-symmetric tensor algebra of $V$ is given by
\begin{align}
\bar{S}V&=\Pi_k \bigoplus_{i_1\le\cdots\le i_k} V_{i_1,\dots,i_k}/\big(V_{i_1,\dots,i_k}\cap \bar I\big)=:\Pi_k \bigoplus_{i_1\le\cdots\le i_k} V_{i_1}\odot\cdots\odot V_{i_k}\nonumber
\\ \label{GraSymTenAlg}
&=\bigoplus_{p=0}^N\;\,\Pi_k\! \bigoplus_{\begin{subarray}{c}{i_1\le\cdots\le i_k}\\ {\sum_j\zg_{i_j}=\zg_p}\end{subarray}} V_{i_1}\odot\cdots\odot V_{i_k},
\end{align}
see~\cite{BP13}. We~denote by $\odot$ the $\Z_2^n$-commutative multiplication that is induced on $\bar{S}V$ by the multiplication $\0$ of $\overline{T}V$. By definition, we have, for $[T]\in (\bar S V)_{\zg_i}$ and $[U]\in (\bar S V)_{\zg_j}$ (obvious notation),
\[
[T]\odot[U]=[T\0 U]=(-1)^{\la \zg_i,\zg_j\rangle}[U]\odot[T].
\]
For instance, if $v_i\in V_i\subset(\bar{S}V)_{\zg_i}$, $v_j\in V_j\subset(\bar{S}V)_{\zg_j}$ and if $i\le j$, we get
\begin{gather}\label{GraCom}
v_i\odot v_j=[v_i\0 v_j] =\big[(-1)^{\la\zg_i,\zg_j\rangle}v_j\0v_i\big] =(-1)^{\la\zg_i,\zg_j\rangle}v_j\odot v_i\in V_i\odot V_j.
\end{gather}

Notice further that, if $i<j$, the linear map
\begin{gather}\label{CanIsoTen}
\iota\colon\ V_i\0 V_j\ni T\mapsto [T]\in V_i\odot V_j\qquad\ (\iota\colon\ V_i\0 V_j\ni v_i\0 v_j\mapsto v_i\odot v_j\in V_i\odot V_j)
\end{gather}
is a vector space isomorphism. Indeed, if $[T]=0$, the representative $T$ is a vector in $(V_i\0 V_j\oplus V_j\0 V_i)\cap\bar I$ and is therefore a finite sum of generators of $\bar I$:
\begin{gather}\label{Inj}
(-1)^{\la \zg_i,\zg_j\rangle}\sum_kv^k_j\0 v^k_i=\sum_kv^k_i\0 v^k_j-T\in (V_i\0 V_j)\cap(V_j\0 V_i)=\{0\}.
\end{gather}
It follows that the {\small LHS} of
equation~\eqref{Inj} vanishes; hence, the first term of the {\small RHS} vanishes, due to the isomorphism $V_i\0 V_j\simeq V_j\0 V_i$, and thus $T$ vanishes as well. In order to show that $\iota$ is also surjective, consider an arbitrary vector in $V_i\odot V_j$. It~reads
\[
[T]=\bigg[\sum_k v_i^k\0 v_j^k+\sum_\ell w_j^\ell\0 w_i^\ell\bigg].
\]
The image by $\iota$ of
\[
\sum_k v_i^k\0 v_j^k+(-1)^{\la\zg_i,\zg_j\rangle}\sum_\ell w_i^\ell\0 w_j^\ell\in V_i\0 V_j
\]
is the corresponding class. This class coincides with $[T]$, since the difference of the representatives is a vector of $\bar I$.

It follows that, for $n=2$ for instance, we have in particular
\begin{align}
V_{00}\odot V_{00}\odot V_{01}\odot V_{10}\odot V_{10}\odot V_{10}\odot V_{11}&\simeq \odot^2 V_{00}\0 V_{01}\0\odot^3V_{10}\0 V_{11}\nonumber
\\
&\simeq\vee^2 V_{00}\0 V_{01}\0\wedge^3V_{10}\0 V_{11},\label{IsoGraSymTenPro}
\end{align}
where $\vee$ (resp., $\wedge$) is the symmetric (resp., antisymmetric) tensor product. Moreover, if the (finite dimensional) vector space $V$ has dimension $q_0|q_1$, $q_2$, $q_3$, we denote the vectors of its basis (in accordance with the notation we adopted earlier in this text) by $b^i_{j}$, where $i\in\{0,1,2,3\}$ refers to the degrees $00$, $01$, $10$, $11$ and where $j\in\{1,\dots,q_i\}$. The basis of the $\Z_2^n$-symmetric tensor product~\eqref{IsoGraSymTenPro} is then made of the tensors
\[
b^0_{j_1}\vee b^0_{j_2}\0 b^1_{j_3}\0 b^2_{j_4}\wedge b^2_{j_5}\wedge b^2_{j_6}\0 b^3_{j_7}
\]
$(j_1\le j_2\;\text{and}\;j_4<j_5<j_6)$, which can also be written
\[
b^0_{j_1}\odot b^0_{j_2}\odot b^1_{j_3}\odot b^2_{j_4}\odot b^2_{j_5}\odot b^2_{j_6}\odot b^3_{j_7}
\]
$(j_1\le j_2\;\text{and}\;j_4<j_5<j_6)$ (see~\eqref{CanIsoTen}). More generally, the basis of $V_{i_1}\odot\cdots\odot V_{i_k}$ ($i_1\le\cdots\le i_k$) is made of the tensors
\begin{gather}\label{BasGraSymTenAlg}
b^{i_1}_{j_1}\odot\cdots\odot b^{i_k}_{j_k}
\end{gather}
($j_\ell\le j_{\ell+1}$ (resp., $<$), if $i_\ell = i_{\ell+1}$ and $\la\zg_{i_\ell},\zg_{i_{\ell+1}}\rangle$ even (resp., odd)). To refer to the previous condition regarding the $j$-s, we write in the following $j_1\lhd\cdots\lhd j_k$.

Observe also that
\[
S^kV=\bigoplus_{i_1\le\cdots\le i_k}V_{i_1}\odot\cdots\odot V_{i_k}=S^k\bigoplus_i V_i=\bigg(\bigotimes_i SV_{i}\bigg)^k,
\]
as well as that, in order to define a linear map on $V_{i_1}\odot\cdots\odot V_{i_k}$ (see~\eqref{IsoGraSymTenPro}), it suffices to define a $k$-linear map on $V_{i_1}\times\cdots\times V_{i_k}$ that is $\Z_2^n$-commutative in the variables $i_\ell=\cdots=i_m$.

\medskip\noindent
{\bf Manifoldification functor.} If $V$ is a $\Z_2^n$-graded vector space, its dual $V^\vee$ is defined by
\[
V^\vee:=\ul{\op{Hom}}(V,\R)=\bigoplus_{i=0}^N\ul{\op{Hom}}_{\zg_i}(V,\R) =\bigoplus_{i=0}^N\op{Hom}(V_i,\R)=\bigoplus_{i=0}^N(V_i)^\vee\in\Z_2^n\tt Vec.
\]
More explicitly, we consider the space of $\R$-linear maps from $V$ to $\R$ of any $\Z_2^n$-degree. It~is clear that the linear maps of degree $\zg_i$ are the linear maps from $V_i$ to $\R$ (that vanish in any other degree). Hence,
\[
\big(V^\vee\big)_i=(V_i)^\vee=:V_i^\vee.
\]
It follows that, if $V$ is finite dimensional of dimension $p|\ul q$, its dual $V^\vee$ has the same dimension. Moreover, any basis $(b^i_k)_{i,k}$ ($i\in\{0,\dots, N\}$ and $k\in\{1,\dots, q_i\}$, where we set $q_0:= p$) of $V$ defines a dual basis $(\zb^k_i)_{i,k}$ of $V^\vee$.

Let now $V\in\Z_2^n\tt FinVec$ be of dimension $p|\ul q$. We~set
\[
V^\vee_*:=\bigoplus_{j=1}^N V^\vee_{j}\in\Z_2^n{\tt FinVec}\qquad
\big(\op{dim}\big(V^\vee_*\big)=0|\ul q\big).
\]

\begin{Proposition}\label{NonCanIso}
If $V$ is a $\Z_2^n$-graded vector space of dimension $p|\ul q$, there is a non-canonical isomorphism of $\Z_2^n$-commutative associative unital $\R$-algebras
\begin{gather*}%\label{Z2nSymAlg}
\flat\colon\ \bar{S}\big(V^\vee_*\big)\stackrel{\sim}{\longrightarrow}\R[[\zx]],
\end{gather*}
where $\R[[\zx]]$ is the global function algebra of $\R^{0|\ul q}$. \end{Proposition}

\begin{proof} As usual, we ordered the $\Z_2^n$-degrees lexicographically, so that the $\zx_j^\ell$-s are ordered unambiguously. We~have
\[
\R[[\zx]]=\Pi_\za\R\,\zx^\za,
\]
where the multi-index $\za$ has components $\za_j^\ell\in\mathbb{N}$ (resp., $\za_j^\ell\in\{0,1\}$), if $\la\zg_j,\zg_j\rangle$ is even (resp., odd).

On the other hand, it follows from equations~\eqref{GraSymTenAlg} and~\eqref{BasGraSymTenAlg} that, choosing a basis $\big(b^j_\ell\big)_{j,\ell}$ of $V_*$ (defined similarly as $V^\vee_*$) and denoting its dual basis by $(\zb^\ell_j)_{j,\ell}$, leads to
\begin{gather}\label{SBarBasis}
\bar S\big(V^\vee_*\big)=\Pi_k\bigoplus_{j_1\le\cdots\le j_k}\bigoplus_{\ell_1\lhd\cdots\lhd \ell_k}\R\,\zb_{j_1}^{\ell_1}\odot\cdots\odot \zb_{j_k}^{\ell_k}=\Pi_k\bigoplus_{|\za|=k}\R\,\zb^\za=\Pi_\za\R\,\zb^\za,
\end{gather}
where $\za_j^\ell\in\mathbb{N}$ (resp., $\za_j^\ell\in\{0,1\}$), if $\la\zg_j,\zg_j\rangle$ is even (resp., odd).

In view of~\eqref{TenMult} and~\eqref{GraCom}, the multiplications of $\R[[\zx]]$ and $\bar S\big(V^\vee_*\big)$ are exactly the same, so that the two $\Z_2^n$-commutative algebras are canonically isomorphic, once a basis of $V_*$ has been chosen.
\end{proof}

\begin{Remark}
We denoted the isomorphism by $\flat$ to remind us of its dependence on the basis~$\big(b^j_\ell\big)_{j,\ell}$.
\end{Remark}

We are now prepared to define the linear $\Z_2^n$-manifold associated to a finite dimensional $\Z_2^n$-vector space. From here we denote the vector space by $\mbf{V}$ instead of $V$ and reserve the notation~$V$ for the manifold $V:=\cM(\mbf{V})$.

Hence, let $\mbf{V}\in\Z_2^n\tt FinVec$ be of dimension $p|\ul q$. The $p$-dimensional vector space $\mbf{V}_0$ of degree~$0$ is of course a smooth manifold of dimension~$p$, as well as a linear $\Z_2^n$-manifold $V_0$ of dimension~$p|\ul 0$. On the other hand, the algebra $\bar{S}\big(\mbf{V}^\vee_*\big)$ is a sheaf of $\Z_2^n$-commutative associative unital $\R$-algebras over $\{\star\}$, i.e., it is a $\Z_2^n$-ringed space with underlying topological space $\{\star\}$, and, in view of~Proposition~\ref{NonCanIso}, this space is (non-canonically) globally isomorphic to $\R^{0|\ul q}=(\{\star\},\R[[\zx]])$. Hence, the space $\big(\{\star\},\bar{S}\big(\mbf{V}_*^\vee\big)\big)$ is a linear $\Z_2^n$-manifold $V_>$ of~dimension $0|\ul q$. Finally, the product $V=V_0\times V_>$ is a $\Z_2^n$-manifold of dimension $p|\ul q$, with base manifold $V_0\times\{\star\}\simeq V_0$ and function sheaf $\cO_{V}$ that is, for any open subset $\zW\subset V_0\simeq \R^p$, given~by
\begin{align}
\cO_V(\zW)&=\cO_{V_0\times V_>}(\zW\times\{\star\})= \Ci_{V_0}(\zW)\widehat{\0}_\R\,\cO_{V_>}(\{\star\})\simeq \Ci(\zW)\widehat{\0}_\R\,\R[[\zx]]\nonumber
\\
&= \Ci(\zW)[[\zx]]=\cO_{\R^{p|\ul q}}(\zW)\label{FunShV}
\end{align}
(since $\zW$ and $\{\star\}$ are $\Z_2^n$-chart domains; for more information about the problem with the function sheaf of product $\Z_2^n$-manifolds, we refer the reader to~\cite{Bruce:2019}). In particular, the $\Z_2^n$-algebras $\cO_V(V_0)$ and $\cO_{\R^{p|\ul q}}(\R^p)$ are isomorphic (see also Definition~13 of product $\Z_2^n$-manifolds in~\cite{Bruce:2019}), so that the $\Z_2^n$-manifolds $V$ and $\R^{p|\ul q}$ are diffeomorphic (given what has been said above, the diffeomorphism is implemented by the choice of a basis of $\mbf{V}$). Finally $V\in\Z_2^n\tt LinMan$ $\big(\op{dim} V=p|\ul q\big)$. We~define the manifoldification functor $\cM$ on objects by
\[%\label{ManFunObj}
\cM(\mbf{V})=V.
\]

We now define $\cM$ on morphisms. A degree zero linear map $\mbf{L}\colon \mbf{V}\to\mbf{W}$ between finite dimensional vector spaces (of dimensions $p|\ul q$ and $r|\ul s$, respectively) is a family of linear maps $\mbf{L}_i\colon \mbf{V}_i\to\mbf{W}_i$ ($i\in\{0,\dots, N\}$). We~denote the transpose maps by $^t\mbf{L}_i\colon \mbf{W}_i^\vee\to \mbf{V}_i^\vee$.

The linear map $\mbf{L}_0\colon \mbf{V}_0\to\mbf{W}_0$ is of course a smooth map $L_0\colon V_0\to W_0$, where $V_0,W_0$ are the vector spaces $\mbf{V}_0,\mbf{W}_0$ viewed as smooth manifolds. The map $L_0$ can also be interpreted as $\Z_2^n$-morphism $L_0\colon V_0\to W_0$ between the $\Z_2^n$-manifolds $V_0,W_0$ (which are of dimension zero in all non-zero degrees). The base morphism of $L_0$ is $L_0$ itself and, for any open subset $\zW\subset W_0$, the pullback $(L_0)^*_\zW$ is the (unital) algebra morphism $-\circ L_0|_\zw\colon \Ci(\zW)\to \Ci(\zw)$ $\big(\zw:=L_0^{-1}(\zW)\big)$ that extends the transpose $^t\mbf{L}_0(-)= -\circ \mbf{L}_0$.

The linear maps $^t\mbf{L}_j\colon \mbf{W}_j^\vee\to \mbf{V}_j^\vee$ ($j\in\{1,\dots, N\}$) define a linear map
\[
{\bar S}\big(^t\mbf{L}\big)\colon\ {\bar S}\big(\mbf{W}_*^\vee\big)\to {\bar S}\big(\mbf{V}_*^\vee\big).
\]
Observe first that to define such a map, it suffices to define a linear map in each tensor degree~$k$, hence, it suffices to define a linear map
\[
\big(^t\mbf{L}\big)^{\odot k}_{j_1\dots j_k}\colon\ \mbf{W}_{j_1}^\vee\odot\cdots\odot\mbf{W}_{j_k}^\vee \to\mbf{V}_{j_1}^\vee\odot\cdots\odot\mbf{V}_{j_k}^\vee,
\]
for any $j_1\le\dots \le j_k$ ($j_a\in\{1,\dots,N\}$). Since the $k$-linear maps
\begin{gather*}
\big(^t\mbf{L}\big)^{\times k}_{j_1\dots j_k}\colon\ \mbf{W}_{j_1}^\vee\times\cdots\times\mbf{W}_{j_k}^\vee\ni \big(\zw^1_{j_1},\dots,\zw^k_{j_k}\big)\mapsto
\\ \hphantom{\big(^t\mbf{L}\big)^{\times k}_{j_1\dots j_k}\colon\ }
^t\mbf{L}_{j_1}\big(\zw^1_{j_1}\big)\odot\cdots\odot\, ^t\mbf{L}_{j_k}\big(\zw^k_{j_k}\big)\in\mbf{V}_{j_1}^\vee\odot\cdots\odot\mbf{V}_{j_k}^\vee
\end{gather*}
are $\Z_2^n$-commutative in the variables $j_\ell=\cdots=j_m$, they define the degree zero linear maps $\big(^t\mbf{L}\big)^{\odot k}_{j_1\dots j_k}$ $\big($we set $\big(^t\mbf{L}\big)^{\odot\, 0}=\id_\R\big)$ and thus the degree zero linear map $\bar S\big(^t\mbf{L}\big)$ that we are looking for. In view of our definitions, the latter is a (unital) $\Z_2^n$-algebra morphism between the global function algebras of the $\Z_2^n$-manifolds $W_>$ and $V_>$, and it therefore defines a unique $\Z_2^n$-morphism $L_>\colon V_>\to W_>$. The base morphism of $L_>$ is the identity $c\colon \{\star\}\to\{\star\}$.

We thus get a $\Z_2^n$-morphism
\begin{gather}\label{ML}
\cM(\mbf{L}):=L:=L_0\times L_>\colon\ \cM(\mbf{V})=V=V_0\times V_>\to \cM(\mbf{W})=W=W_0\times W_>,
\end{gather}
with base map $L_0\times c\simeq L_0$ and pullback ($\zW$ open subset of $W_0$, $\zw:=L_0^{-1}(\zW)$)
\begin{gather}\label{PullProd}
L^*_\zW\colon\ \cO_{W}(\zW)=\Ci_{W_0}(\zW)\widehat{\otimes}_\R\,\bar S\big(\mbf{W}^\vee_*\big)\to \cO_{V}(\zw)=\Ci_{V_0}(\zw)\widehat{\otimes}_\R\,\bar S\big(\mbf{V}^\vee_*\big),
\end{gather}
which is fully defined by $(-\circ L_0|_\zw)\otimes \bar S\big(^t\mbf{L}\big)$.

We must now prove that the $\Z_2^n$-morphism $\cM(\mbf{L})=L$ is a morphism of $\Z_2^n\tt LinMan$, i.e., that in linear coordinates it has linear coordinate pullbacks. As said above, the linear coordinate map $\rmk\colon W\to\R^{r|\ul s}$ is the product of the linear coordinate maps $\rmk_0\colon W_0\to \R^{r|\ul 0}$ and $\rmk_>\colon W_>\to \R^{0|\ul s}$. The first of these coordinate maps is implemented by a basis $b_W$ of $\mbf{W}_0$ and its global pullback $b_W^*\colon \Ci(\R^r)\to \Ci_{W_0}(W_0)$ sends a coordinate function $y^\ell\in \Ci(\R^r)$ to
\[
b_W^*\big(y^\ell\big)=y^\ell\circ b_W=\zb_W^\ell\in\Ci_{W_0}(W_0),
\]
where $\zb_W$ is the dual basis (observe that $b^*_W$ extends the transpose of $b_W$ viewed as vector space isomorphism). Similarly, it is clear from Proposition~\ref{NonCanIso} that the global pullback $\flat_W^{-1}$ of the second coordinate map sends a coordinate function $\zh_j^\ell\in\R[[\zh]]$ to
\[
\flat^{-1}_W\big(\zh_j^\ell\big)=\zb_j^\ell\in\bar S\big(\mbf{W}^\vee_*\big),
\]
where $\big(\zb_j^\ell\big)_{j,\ell}$ is the dual of a basis of $\mbf{W_*}$. Based on what we just said and on the statement~\eqref{PullProd}, we get that the coordinate pullbacks in the linear coordinate expression of $L$ are
\[
\big(b^*_V\big)^{-1}\big(\big(b_W^*\big(y^\ell\big)\big)\circ L_0\big) =\big(b^*_V\big)^{-1}\big(^t\mbf{L}_0\big(\zb_W^{\ell}\big)\big)
=\big(b^*_V\big)^{-1}\bigg(\sum_k(\mbf{L}_0)^\ell_k\zb_V^k\bigg)
=\sum_k(\mbf{L}_0)^\ell_kx^k
\]
and \[
\flat_V\big(^t\mbf{L}_j\big(\flat_W^{-1}\big(\zh_j^{\ell}\big)\big)\big)
=\flat_V\big(^t\mbf{L}_j\big(\zb_j^{\ell}\big)\big) =\flat_V\bigg(\sum_k(\mbf{L}_j)^\ell_k\zb^k_j\bigg)=\sum_k(\mbf{L}_j)^\ell_k\zx_j^k,
\]
where the notations are self-explanatory. Hence, $\cM(\mbf{L})\colon \cM(\mbf{V})\to\cM(\mbf{W})$ is a morphism of $\Z_2^n\tt LinMan$.

Since $\cM(\mbf{L})$ is essentially the transpose of $\mbf L$, we have defined a functor
\[
\cM\colon\ \Z_2^n{\tt FinVec}\to\Z_2^n\tt LinMan
\]
and this functor coincides on $\Z_2^n\tt CarVec$ with the functor $\cM$ that we defined earlier.

We already mentioned that the $\z$-diffeomorphism, say $\rmh$, between $V=\cM(\mbf{V})$ and $\R^{p|\ul q}$ is implemented by a basis $(b^i_k)_{i,k}$ of $\mbf{V}$. Now we can explain this observation in more detail. Indeed, the basis chosen provides a $\z$-vector space isomorphism $\mbf{b}\colon \mbf{V}\to\mbf{R}^{p|\ul q}$, hence, the image $\cM(\mbf{b})\colon \cM(\mbf{V})\to \cM\big(\mbf{R}^{p|\ul q}\big)$ is a $\z$-diffeomorphism (it is even an isomorphism of $\z{\tt LinMan}$), say $b\colon V\to \R^{p|\ul q}$. The diffeomorphism $b=\cM(\mbf{b})$ is a special case of the map $L=\cM(\mbf{L})$ of $\z{\tt LinMan}$, whose construction has been described above. It~is almost obvious from the penultimate paragraph that the diffeomorphism $\rmh$ coincides with the diffeomorphism $b$. Indeed, the diffeomorphism $\rmh$ is the product of two $\z$-diffeomorphisms $\rmh_0\colon V_0\to\R^{p|\ul 0}$ and $\rmh_>\colon V_>\to\R^{0|\ul q}$ (see $\rmk$ in the penultimate paragraph). The same holds for $b$, which is defined as $b=b_0\times b_>$, where $b_0\colon V_0\to\R^{p|\ul 0}$ and $b_>\colon V_>\to \R^{0|\ul q}$ (see~\eqref{ML} and $\eqref{PullProd}$). The map $\rmh_0$ is canonically induced by the basis $(b^0_k)_k$ of $\mbf{V}_0$, and so is $b_0\,;$ hence
$\rmh_0=b_0$. The $\z$-diffeomorphism $b_>$ is defined by the corresponding $\z$-algebra isomorphism
\[
\bar{S}(^t\mbf{b})\colon\ \bar{S}\big(\big(\mbf{R}^{0|\ul q}\big)^\vee\big)\to \bar{S}\big(\mbf{V}_*^\vee\big),
\]
where the source algebra is $\Pi_\za\R\,\ze^\za=\R[[\zx]]$. As seen above, this algebra morphism is fully defined by the transposes $^t\mbf{b}_j\colon (\mbf{R}^{q_j})^\vee\to\mbf{V}_j^\vee$ and their action on the basis $(\ze^\ell_j)_\ell$. The action is
\[
^t\mbf{b}_j\big(\ze_j^\ell\big)=\ze_j^\ell\circ \mbf{b}_j=\zb_j^\ell,
\]
since the image of any $v_j=\sum_{k}v^k_j\, b_k^j\in\mbf{V}_j$ by the two maps is $v_j^\ell$. It follows that
\begin{gather}\label{GlobDiffcM1}
\bar{S}\big(^t\mbf{b}\big)=\flat^{-1}.
\end{gather}
This yields $b_>=\rmh_>$. Finally, we get
\begin{gather}\label{GlobDiffcM2}
\rmh=b=\cM(\mbf{b}).
\end{gather}

\noindent
{\bf Vectorification functor.} In this subsection, we define the vectorification functor
\[
\cV\colon\ \z{\tt LinMan}\to\z{\tt FinVec}.
\]

If $\sfl\in\z{\tt LinMan}$ has dimension $p|\ul q$, we set
\begin{gather}\label{DefVecFun}
\cV(\sfl):=\mbf{L}:=\big(\cO_\sfl^{\op{lin}}(|\sfl|)\big)^\vee =\bigoplus_i\big(\cO_{\sfl,\zg_i}^{\op{lin}}(|\sfl|)\big)^\vee=:\bigoplus_i\mbf{L}_i\in\z{\tt FinVec},
\end{gather}
where $\mbf{L}_i$ has dimension $q_i$ ($q_0=p$). Further, in view of item $(iii)$ of Remark~\ref{RemLinFct}, if $\zvf\colon \sfl\to\sfl'$ is a morphism of $\z\tt LinMan$, then $^t\zvf^*$ is a degree preserving linear map
\[
\cV(\zvf):=\mbf{\Phi}:=\,^t\zvf^*\colon\ \cV(\sfl) =\big(\cO_{\sfl}^{\op{lin}}(|\sfl|)\big)^\vee\to\big(\cO_{\sfl'}^{\op{lin}}(|\sfl'|)\big)^\vee =\cV(\sfl').
\]
The definition of $\cV(\zvf)$ implies that $\cV$ is a functor.

\medskip\noindent
{\bf Compositions of the manifoldification and the vectorification functors.}

\smallskip
$(i)$ We first turn our attention to $\cV\circ\cM$. If
\[
\mbf{V}\in\z{\tt FinVec}\qquad\big(\op{dim}\mbf{V}=p|\ul q\big),
\]
its image
\[
\cM(\mbf{V})=V=V_0\times V_>\in\z{\tt LinMan}\qquad\big(\op{dim}V=p|\ul q\big)
\]
is the product of the linear $\z$-manifolds $V_0$ and $V_>$. Let $\big(b^i_\ell\big)_{i,\ell}$ be a basis of $\mbf{V}$ with dual $\big(\zb_i^\ell\big)_{i,\ell}$ and induced $\z$-vector space isomorphism $\mbf{b}\colon \mbf{V}\to\mbf{R}^{p|\ul q}$ (we denote the induced diffeomorphism from $V_0$ to $\R^p$ by $b_0$). As explained above, it defines a linear coordinate map
\begin{gather}\label{LinCooMapFroBas}
\rmh=\cM(\mbf{b})\colon\ V\to\R^{p|\ul q}
\end{gather}
with pullback morphism
\[
\rmh^*=(-\circ b_0)\widehat{\0}_\R\,\flat^{-1}\colon\ \Ci(\R^p)\widehat{\0}_\R\,\R[[\zx]]\to \Ci_{V_0}(V_0)\widehat{\0}_\R\,\bar S\big(\mbf{V}^\vee_*\big)
\]
(see~\eqref{GlobDiffcM2},~\eqref{PullProd} and~\eqref{GlobDiffcM1}). Using equation~\eqref{LinFun}, denoting the basis of $\big(\mbf{R}^{p|\ul q}\big)^\vee$ as usual by~$\big(\ze_i^\ell\big)_{i,\ell}$, and remembering the identifications~\eqref{Equiv}, we thus get
\begin{align*}
\cV(\cM(\mbf{V}))&=\big(\cO_V^{\op{lin}}(V_0)\big)^\vee=\big(\rmh^*\cO_{\R^{p|\ul q}}^{\op{lin}}(\R^p)\big)^\vee=\bigg(\rmh^*\bigg(\bigoplus_\ell\R\,\ze^\ell_0\oplus\bigoplus_{j,\ell} \R\,\ze^\ell_j\bigg)\bigg)^\vee
\\
&=\bigg(\bigoplus_\ell\R\,\big(\ze^\ell_0\circ b_0\big)\oplus\bigoplus_{j,\ell}\R\,\flat^{-1}(\zx^\ell_j)\bigg)^\vee =\bigg(\bigoplus_\ell\R\,\zb^\ell_0\oplus\bigoplus_{j,\ell}\R\,\zb^\ell_j\bigg)^\vee=\mbf{V}.
\end{align*}

$(ii)$ Regarding $\cM\circ\cV$, recall that if
\[
\sfl\in\z{\tt LinMan}\qquad \big(\op{dim}\sfl=p|\ul q\big).
\]
Definition~\eqref{DefVecFun} yields $\cV(\sfl)=\mbf{L}=\big(\cO_{\sfl}^{\op{lin}}(|\sfl|)\big)^\vee$ (notice that $\mbf{L}$ denotes a vector space here, and not a linear map) and Definition~\eqref{FunShV} leads to $\cM(\mbf{L}):=L:=(L_0,\cO_L)$, where $L_0$ is $\mbf{L}_0=\big(\cO_{\sfl,\zg_0}^{\op{lin}}(|\sfl|)\big)^\vee$ viewed as smooth manifold, and where $\cO_L(\zw)$ ($\zw\subset L_0$ open) is
\[
\cO_L(\zw)=\Ci_{L_0}(\zw)\widehat{\0}_\R\,\bar S\big(\mbf{L}_*^\vee\big)
\]
(see~\eqref{FunShV}). If we choose a basis $\big(\zb_j^\ell\big)_{j,\ell}$ of $\mbf{L}_*^\vee$, we have
\[
\bar S\big(\mbf{L}^\vee_*\big)=\Pi_k\bigoplus_{j_1\le\cdots\le j_k}\bigoplus_{\ell_1\lhd\cdots\lhd \ell_k}\R\,\zb_{j_1}^{\ell_1}\odot\cdots\odot \zb_{j_k}^{\ell_k}=\Pi_\za\R\,\zb^\za,
\]
where $\za_j^k\in\mathbb{N}$ (resp., $\za_j^k\in\{0,1\}$), if $\la\zg_j,\zg_j\rangle$ is even (resp., odd) (see~\eqref{SBarBasis}). Just as
\[
\Ci_{\R^p}(\zW)\widehat{\0}_\R\,\Pi_\za\,\R\,\zx^\za=\Ci_{\R^p}(\zW)\widehat{\0}_\R\,\R[[\zx]] =\Ci_{\R^p}(\zW)[[\zx]]=\zP_\za\,\Ci_{\R^p}(\zW)\,\zx^\za
\]
($\zW\subset\R^p$ open) (see~\cite{Bruce:2019}), we have
\begin{gather}\label{FctShVec}
\cO_L(\zw)=\Pi_\za\,\Ci_{L_0}(\zw)\,\zb^\za =\Pi_k\bigoplus_{j_1\le\cdots\le j_k}\bigoplus_{\ell_1\lhd\cdots\lhd \ell_k}\Ci_{L_0}(\zw)\,\zb_{j_1}^{\ell_1}\odot\cdots\odot \zb_{j_k}^{\ell_k} .
\end{gather}

\begin{Remark}
Let us mention that $\sfl$ and $L$ denote a priori different linear $\z$-manifolds and that our goal is to show that they do coincide.
\end{Remark}

Recall first that, for any $\z$-manifold $M$, there is a projection
\[
\zve_M\colon\ \cO_{M}\to\Ci_{|M|}
\]
of $|M|$-sheaves of $\z$-algebras and that $\zve_M$ commutes with pullbacks. In particular, if $\rmh\colon \sfl\to\R^{p|\ul q}$ is a linear coordinate map of $\sfl$ (a (linear) $\z$-diffeomorphism), its pullback is, for any open subset $|U|\subset|\sfl|$, a $\z$-algebra isomorphism
\[
\rmh^*\colon\ \cO_{\R^{p|\ul q}}(|h|(|U|))\to \cO_\sfl(|U|)
\]
and it restricts to a $\z$-vector space isomorphism
\[
\rmh^*\colon\ \cO^{\op{lin}}_{\R^{p|\ul q}}(|h|(|U|))\to \cO^{\op{lin}}_\sfl(|U|).
\]
Further, as just said, we have
\[
\zve_\sfl\circ \rmh^*=\rmh^*\circ\zve_{\R^{p|\ul q}}=(-\circ|\rmh|)\circ\zve_{\R^{p|\ul q}}
\]
on $\cO_{\R^{p|\ul q}}(|h|(|U|))$. Taking $|U|=|\sfl|$ and restricting the equality to degree zero linear functions
\[
\cO_{\R^{p|\ul q},\zg_0}^{\op{lin}}(\R^p)=(\mbf{R}^{p})^\vee
\]
(see~\eqref{LinFun}), we obtain
\begin{gather}\label{HStar1}
\zve_\sfl\circ\rmh^*=-\circ|\rmh|,
\end{gather}
or, equivalently,
\begin{gather}\label{HStar1'}
\zve_\sfl=(-\circ|\rmh|)\circ(\rmh^*)^{-1},
\end{gather}
where $(\rmh^*)^{-1}$ is a vector space isomorphism from $(\mbf{L}_0)^\vee=\cO_{\sfl,\zg_0}^{\op{lin}}(|\sfl|)$ to $(\mbf{R}^p)^\vee$ and where $-\circ|\rmh|$ is an algebra isomorphism from $\Ci(\R^p)$ to $\Ci(|\sfl|)$. In view of the diffeomorphism $|\rmh|\colon |\sfl|\to \R^p$, the smooth manifold $|\sfl|$ is linear. Hence, it is a finite dimensional vector space also denoted $|\sfl|$ and $|\rmh|$ is a vector space isomorphism, whose dual $^t|\rmh|=-\circ|\rmh|$ is a vector space isomorphism from $(\mbf{R}^p)^\vee\subset\Ci(\R^p)$ to $|\sfl|^{\vee}$. It~follows (see also equation~\eqref{HStar1'}) that the canonical map $\zve_\sfl$ is a vector space isomorphism from $(\mbf{L}_0)^\vee$ to $|\sfl|^\vee$. When identifying these vector spaces, we get $\zve_\sfl=\id$ and $|\sfl|=\mbf{L}_0$, hence the corresponding linear manifolds do also coincide: $|\sfl|=L_0$.

To~prove that the linear $\z$-manifolds $\sfl$ and $L$ coincide, it now suffices to show that their function sheaves coincide. The pullback of $\rmh$ is an isomorphism $\rmh^*\colon \cO_{\R^{p|\ul q}}\to\cO_\sfl$ of sheaves of~$\z$-algebras. Since $\rmh^*$ is a $\z$-vector space isomorphism
\[
\rmh^*\colon\ \big(\mbf{R}^{0|\ul q}\big)^\vee=\cO^{\op{lin}}_{\R^{0|\ul q}}(\R^p)\to \cO^{\op{lin}}_{\sfl,*}(|\sfl|)=\mbf{L}_*^\vee,
\]
the images $\big(\rmh^*\big(\ze_j^\ell\big)\big)_{j,\ell}$ are a basis $\big(\zb_j^\ell\big)_{j,\ell}$ of $\mbf{L}_*^\vee$. Moreover, we know that
\[
|\rmh|=\big(\dots,\zve_\sfl\big(\rmh^*\!\big(\ze_0^\ell\big)\big),\dots\big)=\big(\dots,\rmh^*\!\big(\ze_0^\ell\big),\dots\big),
\]
as $\zve_\sfl=\id$ on $(\mbf{L}_0)^\vee$. Therefore, if $f(x)\in\Ci(\R^p)$, we get
\begin{gather}\label{HStar2}
\rmh^*(f(x))=f(\rmh^*\!(x))=f\circ \big(\dots,\rmh^*\!\big(\ze^\ell_0\big),\dots\big)=f\circ|\rmh|\in\Ci(|\sfl|).
\end{gather}
Equation~\eqref{HStar2} (which generalizes equation~\eqref{HStar1}) shows that $\rmh^*$ is an algebra isomorphism $\rmh^*\colon \Ci(\R^p)\to\Ci(|\sfl|)$. Similarly, if $\zw\subset|\sfl|$ is open, $\zW:=|\rmh|(\zw)\subset\R^p$ and $f(x)\in\Ci(\zW)$, we have
\[
\rmh^*(f(x))=f\circ|\rmh||_{\zw}\in\Ci(\zw),
\]
so that
\begin{gather}\label{IsoBase}
\rmh^*\colon\ \Ci(\zW)\to \Ci(\zw)
\end{gather}
is also an algebra isomorphism. Finally, the $\z$-algebra isomorphism
\begin{gather}\label{CompIso1}
\rmh^*\colon\ \Pi_\za\Ci(\zW)\,\zx^\za\to\cO_\sfl(\zw)
\end{gather}
sends any series $\sum_\za f_\za(x)\zx^\za$ to
\begin{align*}
\sum_\za\rmh^*(f_\za(x))\big(\dots,\rmh^*\!\big(\zx_j^\ell\big),\dots\big)^\za &=\sum_\za\rmh^*(f_\za(x))\big(\dots,\rmh^*\!\big(\ze_j^\ell\big),\dots\big)^\za
\\
&=\sum_\za\rmh^*(f_\za(x))\zb^\za\in\cO_L(\zw)
\end{align*}
(see~\eqref{FctShVec}). The $\z$-algebra morphism
\begin{gather}\label{CompIso2}
\rmh^*\colon\ \Pi_\za\Ci(\zW)\zx^\za\to \cO_L(\zw)
\end{gather}
we get this way (notice that the targets of the arrows~\eqref{CompIso1} and~\eqref{CompIso2} are different) is visibly an isomorphism. Indeed, it is obviously injective, and it is surjective due to~\eqref{IsoBase}. It~follows from~\eqref{CompIso1} and~\eqref{CompIso2} that $\cO_\sfl(\zw)=\cO_L(\zw)$, for any open subset $\zw\subset|\sfl|$. Since $\rmh^*$ commutes with restrictions, the sheaves $\cO_\sfl$ and $\cO_L$ coincide and $\cM(\cV(\sfl))=\sfl$. An alternative way of saying what we just said is to observe that in view of~\eqref{CompIso1} every element of $\cO_\sfl(\zw)$ is the image by $\rmh^*$ of a unique series $\sum_\za f_\za(x)\zx^\za$ and therefore belongs to $\cO_L(\zw)$. Conversely, in view of~\eqref{IsoBase} every element $\sum_\za g_\za \zb^\za$ ($g_\za\in\Ci(\zw)$) of $\cO_{L}(\zw)$ uniquely reads $\sum_\za \rmh^*(f_\za(x))\zb^\za$, is therefore the image by $\rmh^*$ of $\sum_\za f_\za(x)\zx^\za$ and so belongs to $\cO_\sfl(\zw)$.
$(iii)$ We leave it to the reader to check that both functors, $\cV\circ\cM$ and $\cM\circ\cV$, coincide also on morphisms with the identity functors.

\begin{Theorem}
The functors
\[
\cM\colon\ \z{\tt FinVec}\rightleftarrows\z{\tt LinMan}\ \,{:}\, \cV
\]
are an isomorphism of categories.
\end{Theorem}

\noindent
{\bf Comparison of the functors of points.} Since $\z{\tt FinVec}\simeq\z{\tt LinMan}$, the fully faithful functors of points $\cF$ (see Proposition~\ref{FunPtsFinVecFF}) and $\mathcal{S}$ (see Theorem~\ref{FunPtsLinManFF}) of these categories should coincide. However, up till now, the functor $\cF$ is valued in ${\tt Fun}_0\big(\z{\tt Pts}^{\op{op}},{\tt AMod}\big)$, whereas the functor $\mathcal{S}$ is valued in ${\tt Fun}_0\big(\z{\tt Pts}^{\op{op}},{\tt FAMod}\big)$. Since $\tt FAMod$ is a subcategory of $\tt AMod$, the latter functor category is a subcategory of the former. Hence, if we show that the image $\cF(\mbf{V})$ of any object $\mbf{V}$ of $\z{\tt FinVec}$ is a functor of ${\tt Fun}_0\big(\z{\tt Pts}^{\op{op}},{\tt FAMod}\big)$ ($\star$) and that the image $\cF(\zvf)$ of any morphism $\zvf\colon \mbf{V}\to \mbf{W}$ of $\z{\tt FinVec}$ is a natural transforation of ${\tt Fun}_0\big(\z{\tt Pts}^{\op{op}},{\tt FAMod}\big)$ ($\ast$), we~can conclude that $\cF$ is a functor
\[
\cF\colon\ \z{\tt FinVec}\to{\tt Fun}_0\big(\z{\tt Pts}^{\op{op}},{\tt FAMod}\big).
\]

We start proving ($\star$). Since $\tt FAMod$ is a subcategory of $\tt AMod$, we just have to show that the image
\[
\cF(\mbf{V})(\zL)=\mbf{V}(\zL)=(\zL\0\mbf{V})_0
\]
of any object $\zL$ of $\z{\tt Pts}^{\op{op}}$ is a Fr\'echet $\zL_0$-module ($\bullet$) and that the image
\[
\cF(\mbf{V})(\zf^*)=\mbf{V}(\zf^*)=(\zf^*\0\Id_{\mbf{V}})_0
\]
of any morphism $\zf^*\colon \zL\to\zL'$ of $\z\tt Alg$ is a morphism of $\tt FAMod$ ($\circ$).

\newcommand{\rma}{\textrm{a}}
\newcommand{\rmv}{\textrm{v}}

To prove ($\bullet$), we consider a basis of $\mbf{V}$ ($\dim\mbf{V}=p|\ul q$), i.e., an isomorphism $\mbf{b}\colon \mbf{V}\rightleftarrows\mbf{R}^{p|\ul q}\colon \mbf{b}^{-1}$ of $\z$-vector spaces. Since $\cF(\mbf{b})=\mbf{b}_-$ is a natural isomorphism of ${\tt Fun}_0\big(\z{\tt Pts}^{\op{op}},{\tt AMod}\big)$, any of its $\zL$-components is an isomorphism
\[
\mbf{b}_\zL\colon\ \mbf{V}(\zL)\rightleftarrows\mbf{R}^{p|\ul q}(\zL)\ \,{:}\, \mbf{b}_\zL^{-1}
\]
of $\zL_0$-modules. We~use this isomorphism to transfer to $\mbf{V}(\zL)$ the Fr\'echet vector space struc\-ture~of
\begin{gather}\label{TrivId}
\mbf{R}^{p|\ul q}(\zL)=\big(\zL\0\mbf{R}^{p|\ul q}\big)_0=\bigoplus_i\bigoplus_k\zL_{\zg_i}=\Pi_i\Pi_k\zL_{\zg_i}=\zL_0^{\times p}\times\zL_{\zg_1}^{\times q_1}\times\cdots\times\zL_{\zg_N}^{\times q_N}
\end{gather}
(see proof of Proposition~\ref{FL0Mod} and equation~\eqref{FunPtsYon1}), thus obtaining a well-defined Fr\'echet structure and making $\mbf{b}_\zL$ a Fr\'echet vector space isomorphism, i.e., a continuous linear map with continuous inverse. Since $\mbf{b}_\zL$ is $\zL_0$-linear, the action $\cdot$ of $\zL_0$ on $\mbf{V}(\zL)$ is related to its action $\triangleleft$ on $\mbf{R}^{p|\ul q}(\zL)$ by
\[
\rma\cdot\rmv=\mbf{b}_\zL^{-1}(\rma\triangleleft\mbf{b}_\zL(\rmv)),
\]
for any $\textrm{a}\in\zL_0$ and any $\textrm{v}\in\mbf{V}(\zL)$. The action $\cdot$ is thus the composite of the continuous maps $\mathop{\rm id}\times\,\mbf{b}_\zL$, $\triangleleft$, and $\mbf{b}_\zL^{-1}$, hence, it is itself continuous. The $\zL_0$-module and the Fr\'echet vector space structures on $\mbf{V}(\zL)$ therefore define a Fr\'echet $\zL_0$-module structure on $\mbf{V}(\zL)$ and $\mbf{b}_\zL$ becomes an~isomorphism of Fr\'echet $\zL_0$-modules (for any basis $\mbf{b}$ of $\mbf{V}$).

As concerns ($\circ$), recall that $\mbf{V}(\zf^*)$ is a $(\zf^*)_0$-linear map, where the algebra morphism $(\zf^*)_0\colon \zL_0\to \zL'_0$ is the restriction of $\zf^*$. Observe now that, in view of~\eqref{TrivId}, we have
\[
\mbf{R}^{p|\ul q}(\zf^*)=(\zf^*\0\Id)_0=\Pi_i\Pi_k\,\zf^*,
\]
so that $\mbf{R}^{p|\ul q}(\zf^*)$ is continuous as product of continuous maps (indeed, the $\z\tt Alg$-morphism $\zf^*$ is continuous as pullback of the associated $\z$-morphism). As $\mbf{b}_-$ is a natural transformation of~${\tt Fun}_0\big(\z{\tt Pts}^{\op{op}},{\tt AMod}\big)$, we have
\[
\mbf{V}(\zf^*)=\mbf{b}_{\zL'}^{-1}\circ\mbf{R}^{p|\ul q}(\zf^*)\circ\mbf{b}_\zL,
\]
so that $\mbf{V}(\zf^*)$ is continuous (and $(\zf^*)_0$-linear), hence, is a morphism of $\tt FAMod$.

It remains to show that ($\ast$) holds. We~know that $\cF(\zvf)=\zvf_-$ is a natural transformation of~${\tt Fun}_0\big(\z{\tt Pts}^{\op{op}},{\tt AMod}\big)$, i.e., its $\zL$-components $\zvf_\zL$ are $\zL_0$-linear maps and the naturality condition is satisfied. It~thus suffices to explain that $\zvf_\zL=(\Id\0\zvf)_0$ is continuous. Since $\zL$ is a~Fr\'echet algebra, it is a locally convex topological vector space ({\small LCTVS}) and $\Id\colon \zL\to\zL$ is a~degree zero continuous linear map. Further, since $\mbf{V}$ and $\mbf{W}$ are finite dimensional $\z$-vector spaces, the deg\-ree zero linear map $\zvf\colon \mbf{V}\to\mbf{W}$ is automatically continuous for the canonical {\small LCTVS} structures on its source and target. It~follows that $\Id\0\zvf$ and $(\Id\0\zvf)_0$ are continuous linear maps.

\begin{Proposition}\label{FunPtsFinVecFF2}
The functor
\[
\cF\colon\ \z{\tt FinVec}\to{\tt Fun}_0\big(\z{\tt Pts}^{\op{op}},{\tt FAMod}\big)
\]
is fully faithful.
\end{Proposition}

\begin{proof}
The result is obvious in view of Proposition~\ref{FunPtsFinVecFF}, since ${\tt Fun}_0\big(\z{\tt Pts}^{\op{op}},{\tt FAMod}\big)$ is a subcategory of ${\tt Fun}_0\big(\z{\tt Pts}^{\op{op}},{\tt AMod}\big)$.
\end{proof}

\newcommand{\cS}{\mathcal{S}}

We are now ready to refine the idea expressed at the beginning of this subsection that the (fully faithful) functors of points
\[
\cF\colon\ \z{\tt FinVec}\to{\tt Fun}_0\big(\z{\tt Pts}^{\op{op}},{\tt FAMod}\big)
\]
(see Proposition~\ref{FunPtsFinVecFF2}) and
\[
\mathcal{S}\colon\ \z{\tt LinMan}\to{\tt Fun}_0\big(\z{\tt Pts}^{\op{op}},\tt FAMod\big)
\]
(see Theorem~\ref{FunPtsLinManFF}) of the isomorphic categories
\[
\cM\colon\ \z{\tt FinVec}\rightleftarrows\z{\tt LinMan}\ \,{:}\, \cV\
\]
should coincide.

\begin{Theorem}\label{NatIsoSMFGene}
The functors
\[
\mathcal{S}\circ\cM,\, \cF\colon\ \z{\tt FinVec}\to{\tt Fun}_0\big(\z{\tt Pts}^{\op{op}},{\tt FAMod}\big)
\]
are naturally isomorphic.
\end{Theorem}

\newcommand{\sfi}{\textsf{I}}

We first prove the theorem in the Cartesian case
\[
\mathcal{M}\colon\ \Z_2^n{\tt CarVec}\rightleftarrows\Z_2^n\tt CarMan\ \,{:}\, \cV
\]
(see Proposition~\ref{IsoCatCar}). More precisely, it follows from Proposition~\ref{FunPtsFinVecFF2} and Theorem~\ref{FunPtsLinManFF} that the functors $\cF$ and $\cS$ are (fully faithful) functors
\[
\cF\colon\ \z{\tt CarVec}\to{\tt Fun}_0\big(\z{\tt Pts}^{\op{op}},{\tt FAMod}\big)
\]
and
\[
\mathcal S\colon\ \z{\tt CarMan}\to{\tt Fun}_0\big(\z{\tt Pts}^{\op{op}},{\tt FAMod}\big).
\]
Actually:

\begin{Proposition}%\label{NatIsoSMFPart}
The functors
\[
\mathcal{S}\circ\cM,\, \cF\colon\ \z{\tt CarVec}\to{\tt Fun}_0\big(\z{\tt Pts}^{\op{op}},{\tt FAMod}\big)
\]
are naturally isomorphic.
\end{Proposition}

\newcommand{\rmx}{\mathrm{x}}

\begin{proof}
In order to construct a natural isomorphism $\textsf{I}\colon \mathcal{S}\circ\cM\to\cF$, we must define, for any~$\mbf{R}^{p|\ul q}$, a natural isomorphism
\[
\textsf{I}_{\mbf{R}^{p|\ul q}}\colon\ \mathcal{S}\big(\R^{p|\ul q}\big)\to\cF\big(\mbf{R}^{p|\ul q}\big)
\]
of ${\tt Fun}_0\big(\z{\tt Pts}^{\op{op}},{\tt FAMod}\big)$ that is natural in $\mbf{R}^{p|\ul q}$. To build $\textsf{I}_{\mbf{R}^{p|\ul q}}$, we have to define, for each $\zL$, an isomorphism
\[
\sfi_{\mbf{R}^{p|\ul q},\zL}\colon\ \mathcal{S}\big(\R^{p|\ul q}\big)(\zL)\to \cF\big(\mbf{R}^{p|\ul q}\big)(\zL)
\]
of Fr\'echet $\zL_0$-modules that is natural in $\zL$. Recalling that the source and target of this arrow are
\[
\R^{p|\ul q}(\zL)=\op{Hom}_{\z\tt Man}\big(\R^{0|\ul m},\R^{p|\ul q}\big)\qquad \big(\R^{0|\ul m}\simeq\zL\big)
\]
and
\[
\mbf{R}^{p|\ul q}(\zL)=\big(\zL\0\mbf{R}^{p|\ul q}\big)_0=\bigoplus_{i=0}^N\bigoplus_{k=1}^{q_i}\zL_{\zg_i}\0\mbf{R}e^i_k =\bigoplus_{i=0}^N\bigoplus_{k=1}^{q_i}\zL_{\zg_i}=\zL_0^{\times p}\times\zL_{\zg_1}^{\times q_1}\times\cdots\times\zL_{\zg_N}^{\times q_N},
\]
respectively, we set
\[
\sfi_{\mbf{R}^{p|\ul q},\zL}\colon\ \rmx\mapsto \sum_{i,k}\rmx^*\big(u^k_i\big)\0 e^i_k=\big(\rmx^*\big(u^k_i\big)\big)=\big(\rmx^*\big(x^k\big),\rmx^*\big(\zx^k_j\big)\big) =:\big(x^k_\zL,\zx^k_{j,\zL}\big),
\]
where $\big(u^k_i\big)=\big(x^k,\zx^k_j\big)$ are the coordinates of $\R^{p|\ul q}$ and where $\big(e^i_k\big)_{i,k}$ is the canonical basis of $\mbf{R}^{p|\ul q}$. Since we actually used this $1:1$ correspondence to transfer the Fr\'echet $\zL_0$-module structure from $\mbf{R}^{p|\ul q}(\zL)$ to $\R^{p|\ul q}(\zL)$ (see~\eqref{ActionCompWise}), the bijection $\sfi_{\mbf{R}^{p|\ul q},\zL}$ is an isomorphism of Fr\'echet $\zL_0$-modules. This isomorphism is natural with respect to $\zL$. Indeed, if $\zf^*\colon \zL\to\zL'$ is a $\z$-algebra map (with corresponding $\z$-morphism $\zf$), we have
\begin{gather*}
\sfi_{\mbf{R}^{p|\ul q},\zL'}\big(\R^{p|\ul q}(\zf^*)(\rmx)\big)=\sfi_{\mbf{R}^{p|\ul q},\zL'}(\rmx\circ\zf)=\big(\zf^*\big(\rmx^*\big(x^k\big)\big),\zf^*\big(\rmx^*\big(\zx^k_j\big)\big)\big) =(\zf^*\0\Id)_0\big(\sfi_{\mbf{R}^{p|\ul q},\zL}(\rmx)\big).
\end{gather*}
It now suffices to check that $\sfi_{\mbf{R}^{p|\ul q}}$ is natural with respect to $\mbf{R}^{p|\ul q}$ . Hence, let $\mbf{L}\colon \mbf{R}^{p|\ul q}\to\mbf{R}^{r|\ul s}$ be a~degree zero linear map and let $L\colon \R^{p|\ul q}\to\R^{r|\ul s}$ be the corresponding linear $\z$-mor\-phism~$\cM(\mbf{L})$. In~order to prove that
\begin{gather}\label{NatI}
\sfi_{\mbf{R}^{r|\ul s}}\circ\mathcal{S}(L)=\cF(\mbf{L})\circ\sfi_{\mbf{R}^{p|\ul q}},
\end{gather}
we have to show that the $\zL$-components of these natural transformations coincide. To find that these Fr\'echet $\zL_0$-module morphisms coincide, we must explain that they associate the same image to every $\rmx\in\R^{p|\ul q}(\zL)$. When denoting the coordinates of $\R^{r|\ul s}$ by $\big(u'^\ell_i\big)=\big(x'^\ell,\zx'^\ell_j\big)$, we~obtain
\[
\sfi_{\mbf{R}^{r|\ul s},\zL}\big(\mathcal{S}(L)_\zL(\rmx)\big)=\sfi_{\mbf{R}^{r|\ul s},\zL}(L\circ\rmx)=\big(\rmx^*\big(L^*\big(x'^\ell\big)\big), \rmx^*\big(L^*\big(\zx'^{\ell}_j\big)\big)\big)=\big(\rmx^*\big(L^*\big(u'^\ell_i\big)\big)\big),
\]
where
\[
L^*(u'^{\ell}_i)=\sum_{k=1}^{q_i} \mbf{L}^{\ell i}_{i k}\,u^k_i,
\]
in view of~\eqref{Pull1} and~\eqref{Pull2}. It~follows that
\begin{align*}
\big(\rmx^*\big(L^*\big(u'^\ell_i\big)\big)\big)&=\bigg(\sum_{k=1}^{q_i} \mbf{L}^{\ell i}_{i k}\,\rmx^*\big(u^k_i\big)\bigg)=\sum_{i,\ell}\bigg(\sum_k\mbf{L}^{\ell i}_{ik}\,\rmx^*\big(u_i^k\big)\bigg)\0 e'^i_\ell
\\&
=\sum_{i,k}\rmx^*\big(u_i^k\big)\0 \bigg(\sum_\ell\mbf{L}^{\ell i}_{ik}e'^i_\ell\bigg)
=\sum_{i,k}\rmx^*\big(u_i^k\big)\0 \mbf{L}\big(e^i_k\big)
\\
&=(\Id\0\mbf{L})_0\bigg(\sum_{i,k}\rmx^*\big(u_i^k\big)\0 e^i_k\bigg)=\cF(\mbf{L})_\zL\big(\sfi_{\mbf{R}^{p|\ul q},\zL}(\rmx)\big),
\end{align*}
where $\big(e'^i_\ell\big)_{i,\ell}$ is the basis of $\mbf{R}^{r|\ul s}$.
\end{proof}

\newcommand{\cI}{\mathcal{I}}

We are now able to prove Theorem~\ref{NatIsoSMFGene}.

\begin{proof}
For simplicity, we set
\[
{\tt T}:={\tt Fun}_0\big(\z{\tt Pts}^{\op{op}},{\tt FAMod}\big).
\]
In order to build a natural isomorphism $\cI\colon \cS\circ\cM\to\cF$, we must define, for any $\mbf{V}\in\z{\tt FinVec}$, a natural isomorphism
\[
\cI_{\mbf{V}}\colon \cS(V)\to \cF(\mbf{V})
\]
of ${\tt T}$ that is natural in $\mbf{V}$.
Set
\[
\dim\mbf{V}=p|\ul q
\]
and let $\mbf{b}$ be a basis of $\mbf{V}$, or, equivalently, a $\z$-vector space isomorphism $\mbf{b}\colon \mbf{V}\to\mbf{R}^{p|\ul q}$. In view of~\eqref{LinCooMapFroBas}, the morphism $\cM(\mbf{b})\colon V\to \R^{p|\ul q}$ is a linear $\z$-diffeomorphism. Using Proposition~\ref{FunPtsFinVecFF2} and Theorem~\ref{FunPtsLinManFF}, we obtain that
\[
\cF(\mbf{b})\colon\ \cF(\mbf{V})\to\cF\big(\mbf{R}^{p|\ul q}\big)
\]
and
\[
\cS(\cM(\mbf{b}))\colon\ \cS(V)\to\cS\big(\R^{p|\ul q}\big)
\]
are natural isomorphisms of ${\tt T}$. As
\[
\sfi_{\mbf{R}^{p|\ul q}}\colon\ \cS\big(\R^{p|\ul q}\big)\to\cF\big(\mbf{R}^{p|\ul q}\big)
\]
is a natural isomorphism of ${\tt T}$ as well, the transformation
\[
\cI_{\mbf{V}}:=\cF\big(\mbf{b}^{-1}\big)\circ\sfi_{\mbf{R}^{p|\ul q}}\circ\cS(\cM(\mbf{b}))
\]
is a natural isomorphism
\[
\cI_{\mbf{V}}\colon\ \cS(V)\to\cF(\mbf{V})
\]
as requested. In view of equation~\eqref{NatI}, the transformation $\cI_{\mbf{V}}$ is well-defined, i.e., is independent of the basis chosen.

It remains to show that $\cI_\mbf{V}$ is natural in $\mbf{V}$, i.e., that, for any degree zero linear map $\phi\colon \mbf{V}\to\mbf{W}$ ($\dim\mbf{W}=r|\ul s$) and for any basis $\mbf{b}$ (resp., $\mbf{c}$) of $\mbf{V}$ (resp., $\mbf{W}$), we have
\[
\cF(\phi)\circ\cF\big(\mbf{b}^{-1}\big)\circ\sfi_{\mbf{R}^{p|\ul q}}\circ\cS(\cM(\mbf{b}))=\cF\big(\mbf{c}^{-1}\big)\circ\sfi_{\mbf{R}^{r|\ul s}}\circ\cS(\cM(\mbf{c}))\circ\cS(\cM(\phi)),
\]
or, equivalently,
\[
\sfi_{\mbf{R}^{r|\ul s}}\circ\cS\big(\cM\big(\mbf{c}\circ\phi\circ\mbf{b}^{-1}\big)\big) =\cF\big(\mbf{c}\circ\phi\circ\mbf{b}^{-1}\big)\circ \sfi_{\mbf{R}^{p|\ul q}}.
\]
Since $\mbf{L}:=\mbf{c}\circ\phi\circ\mbf{b}^{-1}$ is a degree zero linear map $\mbf{L}\colon \mbf{R}^{p|\ul q}\to\mbf{R}^{r|\ul s}$, equation~\eqref{NatI} allows once more to conclude.
\end{proof}

\newcommand{\ccD}{\mathcal{D}}

\noindent
{\bf Internal Homs}. A topological property is a property of topological spaces that is invariant under homeomorphisms (isomorphisms of topological spaces). More intuitively, a ``topological property'' is a property that only depends on the topological structure, or, equivalently, that can be expressed by means of open subsets. Similarly, equivalences of categories (``isomorphisms'' of categories) preserve all ``categorical properties and concepts''. Hence, an equivalence should preserve products. It~turns out that this statement is actually correct. More precisely, if $\cE\colon {\tt S}\to {\tt T}$ is part of an equivalence of categories, then a functor $\ccD\colon {\tt I}\to{\tt S}$ has limit $s$ if and only if the functor $\cE\circ\ccD\colon {\tt I}\to {\tt T}$ has limit $\cE(s)$. Applying the statement to the discrete index category ${\tt I}$ with two objects $\{1,2\}$ and setting $\ccD(i)=s_i$ ($i\in\{1,2\}$), we get that $s_1$ and $s_2$ have product $s$ if and only if $\cE(s_1)$ and $\cE(s_2)$ have product $\cE(s)$. Now, the category $\z{\tt FinVec}$ has the obvious binary product $\times$. It~follows that, for any vector spaces $\mbf{V},\mbf{W}\in\z{\tt FinVec}$, the manifolds $\cM(\mbf{V}),\, \cM(\mbf{W})\in\z{\tt LinMan}$ have product
\[
\cM(\mbf{V})\times\cM(\mbf{W})=\cM(\mbf{V}\times\mbf{W}).
\]
If $\sfl,\sfl'\in\z{\tt LinMan}$, the categorical isomorphism implies that $\sfl=\cM(\cV(\sfl))$ and similarly for $\sfl'$, so that the product $\sfl\times\sfl'$ exists and is
\begin{gather}\label{TrafroVec}
\sfl\times\sfl'=\cM(\cV(\sfl)\times\cV(\sfl')).
\end{gather}
Hence, the category $\z{\tt LinMan}$ has finite products.

Equation~\eqref{TrafroVec} shows that we got the product of $\z{\tt LinMan}$ by transferring to ${\tt T}:=\z{\tt LinMan}$ the product of ${\tt S}:=\z{\tt FinVec}$. We can similarly transfer to ${\tt T}$ the closed symmetric monoidal structure of ${\tt S}$. Indeed, the category $\z{\tt Vec}$ is closed symmetric monoidal for the standard tensor product $-\0_{\z\tt Vec}-$ of $\z$-vector spaces and the standard internal Hom $\ul{\text{Hom}}_{\z\tt Vec}(-,-)$ of $\z$-vector spaces, which is defined, on objects for instance, by
\begin{gather*}%\label{IHomVec}
\ul{\text{Hom}}_{\z\tt Vec}(\mbf{V},\mbf{W}):=\bigoplus_i\ul{\text{Hom}}_{\z{\tt Vec}
,\zg_i}(\mbf{V},\mbf{W})\in\z{\tt Vec},
\end{gather*}
for any $\mbf{V},\mbf{W}\in\z\tt Vec$. Of course, if $\mbf{V},\mbf{W}\in\tt S$, then $\ul{\text{Hom}}_{\z\tt Vec}(\mbf{V},\mbf{W})\in\tt S$, and the same holds for $\mbf{V}\0_{\z\tt Vec}\mbf{W}$. It follows that ${\tt S}=\z{\tt FinVec}$ is also a closed symmetric monoidal category. If~we set now
\begin{gather}\label{CMSLinMan}
\sfl\0_{\tt T}\sfl':=\cM(\cV(\sfl)\0_{\z\tt Vec}\cV(\sfl')),\qquad \ul{\op{Hom}}_{\tt T}(\sfl,\sfl'):=\cM\big(\ul{\op{Hom}}_{\z{\tt Vec}}(\cV(\sfl),\cV(\sfl'))\big),
\end{gather}
and similarly for morphisms, we get a closed symmetric monoidal structure on ${\tt T}=\z{\tt LinMan}$:

{\sloppy\begin{Proposition}
The category $\z{\tt LinMan}$ is closed symmetric monoidal for the struc\-ture~\eqref{CMSLinMan}.
\end{Proposition}

}

Alternatively, we could have defined $\ul{\op{Hom}}_{\tt T}(\sfl,\sfl')\in\tt T$ using the fully faithful functor of points
\[
\cS\colon\ {\tt T}\ni\sfl\mapsto \op{Hom}_{\z{\tt Man}}(-,\sfl)=:\sfl(-)\in{\tt Fun}_0\big(\z{\tt Pts}^{\op{op}},{\tt FAMod}\big),
\]
i.e., defining first a functor $F_{\sfl,\sfl'}(-)$ in the target category, and then showing that this functor is representable by some $\ul{\op{Hom}}_{\tt T}(\sfl,\sfl')\in\tt T$:
\[
F_{\sfl,\sfl'}(-)=\op{Hom}_{\z\tt Man}(-,\ul{\op{Hom}}_{\tt T}(\sfl,\sfl'))=\ul{\op{Hom}}_{\tt T}(\sfl,\sfl')(-).
\]
This ``functor of points approach'' is often easier.

{\sloppy
To shed some light on our more abstract definition above, we now compute $\ul{\op{Hom}}_{\tt T}\big(\R^{p|\ul q},\R^{r|\ul s}\big)(\zL)$ ($\diamond$) assuming some familiarity with $\z$-graded matrices $\op{gl}(r|\ul s \times p|\ul q,\zL)$ with entries in $\zL\in\z{\tt Alg}$. Details can be found in Section~\ref{subsec:Zmatrices} which we leave in its natural place. However, we highly recommend reading it before working though the end of this section.

}

We observe first that
\[
\ul{\Hom}_{\z{\tt Vec},\zg_k}\big(\mbf{R}^{p|\ul q},\mbf{R}^{r|\ul s}\big)=\op{gl}_{\zg_k}\big(r|\ul s \times p|\ul q,\,\R\big)\in{\tt Vec}.
\]
In order to understand the gist here, we consider the case $n=2$, so that a matrix $X\in\op{gl}_{\zg_k}\big(r|\ul s \times p|\ul q,\,\R\big)$ has the block format
\begin{equation}\label{glR}
X = \left(\begin{array}{c|c|c|c}
X_{00} & X_{01} & X_{02} & X_{03} \\
\hline
X_{10} & X_{11} & X_{12} & X_{13} \\
\hline
X_{20} & X_{21} & X_{22} & X_{23} \\
\hline
X_{30} & X_{31} & X_{32} & X_{33} \\
\end{array}\right)\!,
\end{equation}
where the degree $x_{ij}$ of the block $X_{ij}$ is
\begin{gather}\label{DegBas}
x_{ij}=\zg_i+\zg_j+\zg_k.
\end{gather}
Since the entries of the $X_{ij}$ are real numbers and so of degree $\zg_0$, all the blocks with non-vani\-shing~$x_{ij}$ do vanish. For instance, if~$\zg_k=01\in\z$ (resp., $\zg_k=11$) (do not confuse with the row-column index $01$ in $X_{01}$ (resp., $11$ in $X_{11}$)), the degree $x_{ij}=0$ if and only if~$ij\in\{01, 10, 23, 32\}$ (resp., $ij\in\{03, 12, 21, 30\}$) (as in most of the other cases in this text, the $\z$-degrees are lexicographically ordered), so that only these $X_{ij}$ do not vanish. It~follows that
\[
\ul{\Hom}_{\z\tt Vec}\big(\mbf{R}^{p|\ul q},\mbf{R}^{r|\ul s}\big)=\op{gl}\big(r|\ul s \times p|\ul q,\R\big)\in\z{\tt FinVec}
\]
is made of the matrices~\eqref{glR}, where no block $X_{ij}$ vanishes a priori. The canonical basis of this $\z$-vector space are the obvious matrices $E_{ik,j\ell}$ ($i,j\in\{0,\dots,N\}$, $k\in\{1,\dots,s_i\}$, $\ell\in\{1,\dots, q_j\}$) with all entries equal to 0 except the entry $kl$ in $X_{ij}$ which is $1$. In view of equation~\eqref{DegBas}, the vectors of this basis have the degrees $\zg_i+\zg_j$. We can of course identify (up to renumbering) this $\z$-vector space with $\mbf{R}^{t|\ul u}$, where $u_n$ ($n\in\{0,\dots,N\}$) is equal to
\begin{gather}\label{UN}
u_n=\sum_{i,j\colon \zg_i+\zg_j=\zg_n}s_iq_j
\end{gather}
(we set $s_0=r$, $q_0=p$, $u_0:=t$). Hence:
\begin{gather}\label{IHomVecMat}
\ul{\Hom}_{\z\tt Vec}\big(\mbf{R}^{p|\ul q},\mbf{R}^{r|\ul s}\big)=\op{gl}\big(r|\ul s \times p|\ul q,\R\big)=\mbf{R}^{t|\ul u}\in\z{\tt CarVec}.
\end{gather}

Combining~\eqref{CMSLinMan} and~\eqref{IHomVecMat}, we get
\begin{gather}\label{IntHomLinManCarMan}
\ul{\op{Hom}}_{\z\tt LinMan}\big(\R^{p|\ul q},\R^{r|\ul s}\big)=\cM\big(\ul{\op{Hom}}_{\z\tt Vec}\big(\mbf{R}^{p|\ul q},\mbf{R}^{r|\ul s}\big)\big)=\R^{t|\ul u}\in\z{\tt CarMan}.
\end{gather}
We now come back to ($\diamond$). Setting as usual $\R^{0|\ul m}\simeq\zL$, we get the isomorphism
\[
\ul{\op{Hom}}_{\z\tt LinMan}\big(\R^{p|\ul q},\R^{r|\ul s}\big)(\zL)=\R^{t|\ul u}(\zL)\simeq \Pi_{n=0}^N\,\zL_{\zg_n}^{\times u_n}
\]
of Fr\'echet $\zL_0$-modules. On the other hand, the vector space $\op{gl}_0\big(r|\ul s\times p|\ul q,\zL\big)$ is a $\zL_0$-module and this module ``coincides'' obviously with
\[
\op{gl}_0\big(r|\ul s\times p|\ul q,\zL\big)=\Pi_{n=0}^N\,\zL_{\zg_n}^{\times u_n}.
\]
By transferring the Fr\'echet structure, we get an ``equality'' of Fr\'echet $\zL_0$-modules. Hence, the Fr\'echet $\zL_0$-module isomorphism
\begin{gather}\label{Man0Mat}
\ul{\op{Hom}}_{\z\tt LinMan}\big(\R^{p|\ul q},\R^{r|\ul s}\big)(\zL)=\R^{t|\ul u}(\zL)\simeq \Pi_{n=0}^N\,\zL_{\zg_n}^{\times u_n}=\op{gl}_0\big(r|\ul s\times p|\ul q,\zL\big)\in{\tt F}\zL_0{\tt Mod}.
\end{gather}

There is a natural upgrade that is independent of the internal Homs and makes $G:=\op{gl}_0(r|\ul s\times p|\ul q,-)$ a functor $G\in{\tt Fun_0}\big(\z{\tt Pts}^{\op{op}},{\tt FAMod}\big)$. Indeed, it suffices to define $G$ on a $\z{\tt Alg}$-morphism $\zf^*\colon \zL\to\zL'$ as
\[
G(\zf^*)\colon\ G(\zL)\ni X\mapsto \zf^*(X)\in G(\zL'),
\]
where $\zf^*(X)$ is defined entry-wise. The morphism $G(\zf^*)$ is clearly $(\zf^*)_0$-linear. It~is also continuous, as it can be viewed as a product of copies of $\zf^*$. Since $G$ respects compositions and identities it is actually a functor of the functor category mentioned. The functors $G$ and $\R^{t|\ul u}(-)=\cS(\R^{t|\ul u})$ are of course naturally isomorphic. Since $\cS$ is a fully faithful functor
\[
\cS\colon\ \z{\tt LinMan}\to{\tt Fun_0}\big(\z{\tt Pts}^{\op{op}},{\tt FAMod}\big),
\]
the functor $G$ can be viewed as represented by the linear $\z$-manifold $\R^{t|\ul u}$.
\begin{Proposition}\label{gl0}
The functor $\op{gl}_0(r|\ul s\times p|\ul q,-)$ is representable and the Cartesian $\z$-manifold
\[
\op{gl}_0\big(r|\ul s\times p|\ul q\big):=\R^{t|\ul u},
\]
with dimension $t|\ul u$ defined in equation~\eqref{UN}, is ``its'' representing object.
\end{Proposition}

\begin{Example}
For $n=2$, we find that $\op{gl}_0(1|1,1,1) = \R^{4|4,4,4}$.
\end{Example}

\section[Z2n-Lie groups and linear actions]%\label{LGLA}
{$\boldsymbol{\Z_2^n}$-Lie groups and linear actions}

\newcommand{\gl}{\op{gl}}
\newcommand{\rim}{\mathring}

\subsection[Z2n-matrices]
{$\boldsymbol{\Z_2^n}$-matrices}\label{subsec:Zmatrices}
We will consider matrices that are valued in some $\Z_2^n$-Grassmann algebra $\Lambda$, though everything we say generalizes to arbitrary $\Z_2^n$-commutative associative unital $\R$-algebras. A homogeneous matrix $X \in \gl_x\big(r|\ul s\times p|\ul{q},\Lambda\big)$ of degree $x\in\z$ is understood to be a block matrix
\begin{equation*}%\label{Z2nMatrices}
X =\begin{pmatrix}
X_{00} & \dotsc & X_{0N} \\
\hline
\vdots & \ddots & \vdots \\
\hline
X_{N0} & \dotsc & X_{NN}
\end{pmatrix}\!,
\end{equation*}
with the entries of each block $X_{ij}$ being elements of the $\Z_2^n$-Grassmann algebra $\Lambda$. Here the degree $x_{ij}\in\z$ of $X_{ij}$ is
\[
x_{ij} = \gamma_i + \gamma_j + x
\]
and the dimension of $X_{ij}$ is
\[
\dim(X_{ij}) = s_i \times q_j
\]
(setting $s_0=r$ and $q_0 = p$ as usual). Addition of such matrices and multiplication by reals are defined in the obvious way and they endow $\gl_x\big(r|\ul s\times p|\ul{q},\Lambda\big)$ with a vector space structure. We~set
\[
\gl\big(r|\ul s\times p|\ul{q},\Lambda\big):=\bigoplus_{x\in\z}\gl_x\big(r|\ul s\times p|\ul{q},\Lambda\big)\in\z{\tt Vec}.
\]
Multiplication by an element of $\Lambda$ requires an extra sign factor given by the row of the matrix, i.e., for any homogeneous $\lambda \in \Lambda_{\zg_k}$, we have that
\[
(\lambda\: X)_{ij} = (-1)^{\langle \zg_k, \gamma_i\rangle} \lambda \: X_{ij}.
\]
We thus obtain on $\gl\big(r|\ul s\times p|\ul{q},\Lambda\big)$ a $\z$-graded module structure over the $\z$-commutative algebra $\zL$. If $ r|\ul s=p|\ul q$, we write
\[
\gl\big(p|\ul q,\zL\big):=\gl(p|\ul q \times p|\ul q,\zL).
\]
Multiplication of matrices in $\gl(p|\ul{q},\Lambda)$ is via standard matrix multiplication~-- now taking care that the entries are from a $\Z_2^n$-commutative algebra. Equipped with this multiplication, the~$\z$-graded $\zL$-module $\gl\big(p|\ul q,\zL\big)$ is a $\z$-graded associative unital $\R$-algebra. In particular, the degree zero matrices $\gl_0\big(p|\ul q,\zL\big)$ form an associative unital $\R$-algebra. Since multiplication of matrices only uses multiplication and addition in $\zL$, we can replace $\zL$ not only, as said above, by any $\z$-commutative associative unital $\R$-algebra, but also by any $\z$-commutative ring $R$ and then get a ring $\gl_0(p|\ul q,R)$. We~denote by $\GL(p|\ul q,R)$ the group of invertible matrices in~$\gl_0(p|\ul q,R)$. For further details the reader may consult~\cite{Covolo:2012}.

\subsection[Invertibility of Z2n-matrices]{Invertibility of $\boldsymbol{\z}$-matrices}
Let $R$ be a $\z$-commutative ring which is Hausdorff-complete in the $J$-adic topology, where
$J$ is the (proper) homogeneous ideal of $R$ that is generated by the elements of non-zero degree $\gamma_j\in\Z_2^n$, $j\in\{1,\dots,N\}$. The $\Z_2^n$-graded ring morphism
$\varepsilon\colon R \to \,R/J$, where
\[
R/J=\bigoplus_iR_i/(R_i\cap J)=R_0/(R_0\cap J)
\]
vanishes in all non-zero degrees, induces a ring morphism
\[
\tilde\varepsilon\colon\ \gl_0\big(p|\ul{q},R\big)\ni X\mapsto \tilde\varepsilon(X)\in\op{Diag}\big(p|\ul{q},R/J\big),
\]
where $\tilde\varepsilon(X)$ is the
block-diagonal matrix with diagonal blocks $\tilde\varepsilon(X_{ii})$ (with commuting entries).

The following proposition appeared as Proposition~5.1 in~\cite{CKP2020}:

\begin{Proposition}\label{invertibleblockdiag}
Let $R$ be a $J$-adically Hausdorff-complete $\z$-commutative ring and let $X\in\gl_0\big(p|\ul{q},R\big)$ be a degree zero $p|{\ul q}\times p|{\ul q}$ matrix with entries in $R$, written in the
standard block format
\begin{equation*}
X = \left(\begin{array}{c|c|c}
X_{00} & \dotsc & X_{0N} \\
\hline
\vdots & \ddots & \vdots \\
\hline
X_{N0} & \dotsc & X_{NN} \\
\end{array}\right)\!.
\end{equation*}
We have:
\begin{align*}
X\in\GL(p|\ul q,R)& \Leftrightarrow X_{ii}\in\GL(q_i,R),\quad\forall i\\
&\Leftrightarrow\tilde\varepsilon(X)\in\GL(p|\ul q,R/J)
\Leftrightarrow \tilde\varepsilon(X_{ii})\in\GL(q_i,R/J),\quad\forall i.
\end{align*}
\end{Proposition}

In this work, we are of course mainly interested in the case $R:=\zL=\R\oplus\rim{\zL}$ and $J=\rim{\zL}$, so~that $R/J=\R$.

\subsection[Z2n-Lie groups and their functor of points]
{$\boldsymbol{\Z_2^n}$-Lie groups and their functor of points}

Groups, or, better, group objects can easily be defined in any category with finite products, i.e., any category $\tt C$ with terminal object $1$ and binary categorical products $c\times c'$ ($c,c'\in {\tt C}$).

\newcommand{\cG}{\mathcal{G}}

If $\tt C$ is a concrete category, the definition of a group object is very simple. For instance, if $\tt C$ is the concrete category $A\tt FM$ of Fr\'echet manifolds over a Fr\'echet algebra $A$, a group object $\cG$ in~$\tt C$ is just an object $\cG\in\tt C$ that is group whose structure maps $\zm\colon \cG\times \cG\to \cG$ and $\op{inv}\colon \cG\to \cG$ are $\tt C$-morphisms, i.e., $A$-smooth maps. We~refer of course to a group object in $A\tt FM$ as a {\it Fr\'echet $A$-Lie group}.

If $\tt C$ is the category $\z\tt Man$ of $\z$-manifolds, the definition of a group object is similar, but all the (natural) requirements (above) have to be expressed in terms of arrows (since there are no points here). More precisely, a group object $G$ in $\tt C$ is an object $G\in\tt C$ that comes equipped with $\tt C$-morphisms
\[
\zm\colon\ G\times G\to G,\qquad \op{inv}\colon\ G\to G\qquad\text{and}\qquad e\colon\ 1\to G
\]
(the terminal object $1$ is here the $\z$-manifold $\R^{0|\ul 0}=(\{\star\},\R)$), which are called multiplication, inverse and unit, and satisfy the standard group properties (expressed by means of arrows): $\zm$ is associative, $\op{inv}$ is a two-sided inverse of $\zm$ and $e$ is a two-sided unit of $\zm$. To understand the arrow expressions of these properties, we need the following notations. We~denote by $\zD\colon G\to G\times G$ the canonical diagonal $\tt C$-morphism and we denote by $e_G\colon G\to G$ the composite of the unique $\tt C$-morphism $1_G\colon G\to 1$ and the unit $\tt C$-morphism $e\colon 1\to G$. The left inverse condition now reads
\[
\zm\circ(\op{inv}\times\id_G)\circ \zD=e_G
\]
and the left unit condition reads
\[
\zm\circ (e_G\times\id_G)\circ\zD =\id_G
\]
(and similarly for the right conditions). The associativity of $\zm$ is of course encoded by \begin{gather}\label{AssG}
\zm\circ(\zm\times \id_G)=\zm\circ(\id_G\times\zm).
\end{gather}
We refer to a group object in $\z\tt Man$ as a {\it $\z$-Lie group}.

A morphism $F\colon \cG\to \cG'$ of Fr\'echet $A$-Lie groups is of course defined as an $A$-smooth map that is a group morphism. Analogously, a morphism $F\colon \cG\to \cG'$ from a Fr\'echet $A$-Lie group to a Fr\'echet $A'$-Lie group is a morphism of $\tt AFM$ that is also a group morphism. We~denote the {\it category of Fr\'echet $A$-Lie groups} by $A\tt FLg$ and we write $\tt AFLg$ for the {\it category of Fr\'echet Lie groups over any Fr\'echet algebra}.

Further, a morphism $\Phi\colon G\to G'$ of $\z$-Lie groups is a $\z$-morphism that respects the multiplications, the inverses and the units (obvious arrow definitions). The {\it category of $\z$-Lie groups} we denote by $\z\tt Lg$.

The functor of points of $\z$-manifolds
\begin{gather}\label{SMan}
\cS\colon\ \z{\tt Man}\to {\tt Fun}_0\big(\z{\tt Pts}^{\op{op}},{\tt AFM}\big)
\end{gather}
induces a fully faithful functor of points of $\z$-Lie groups:

\begin{Theorem}\label{FunPtsZLg}
The functor
\begin{gather}\label{ResSLg}
\cS\colon\ \z{\tt Lg}\to {\tt Fun}_0\big(\z{\tt Pts}^{\op{op}},{\tt AFLg}\big)
\end{gather}
is fully faithful. Moreover, if $M\in\z{\tt Man}$ and
\[
\cS(M)=M(-)\in{\tt Fun}_0\big(\z{\tt Pts}^{\op{op}},{\tt AFLg}\big),
\]
then $M\in\z{\tt Lg}$. \end{Theorem}

This theorem was announced as~\cite[Theorem 3.30]{Bruce:2019b} without proper explanation or proof.

\begin{proof}
It is clear that we have subcategories
\[
{\tt AFLg}\subset{\tt AFM},\qquad {\tt Fun}_0\big(\z{\tt Pts}^{\op{op}},{\tt AFLg}\big)\subset{\tt Fun}_0\big(\z{\tt Pts}^{\op{op}},{\tt AFM}\big)\qquad\text{and}\qquad{\z\tt Lg}\subset\z{\tt Man}.
\]
Therefore, in order to prove that the functor~\eqref{SMan} restricts to a functor~\eqref{ResSLg}, it suffices to show that $\cS$ sends objects $G$ and morphisms $\Phi$ of $\z{\tt Lg}$ to objects and morphisms of the functor category with target ${\tt AFLg}$.
Observe first that, for any $M,N\in\z{\tt Man}$, we have the functor equality
\begin{gather}\label{SPRPRS}
\cS(M\times N)=(M\times N)(-)=M(-)\times N(-)=\cS(M)\times\cS(N),
\end{gather}
in view of the universal property of $M\times N$. Further, if $\phi\colon M\to M'$ and $\psi\colon N\to N'$ are two $\z$-morphisms, the natural transformation
\[
\cS(\phi\times\psi)=(\phi\times\psi)_- \colon\ (M\times N)(-)\to (M'\times N')(-)
\]
 becomes $\phi_-\times\psi_-$, if we read it through the identification~\eqref{SPRPRS}.

Now, if $G\in\z{\tt Lg}$ with structure $\z$-morphisms $\zm$, $\op{inv}$ (and $e$), then the ${\tt AFM}$-valued functor $\cS(G)=G(-)$ is actually ${\tt AFLg}$-valued. This means that it sends any $\z$-Grassmann algebra $\zL$ and any $\z{\tt Alg}$-morphism $\zf^*\colon \zL\to\zL'$ to an object $G(\zL)$ and a morphism $G(\zf^*)$ of $\tt AFLg$.

For $G(\zL)\in\zL_0{\tt FM}$, notice that the natural transformations $\cS(\zm)=\zm_-,\; \cS(\op{inv})=\op{inv}_-$ (and $\cS(e)=e_-$) have $\zL_0$-smooth $\zL$-components
\[
\zm_\zL\colon\ G(\zL)\times G(\zL)\to G(\zL),\qquad \op{inv}_\zL\colon\ G(\zL)\to G(\zL)\qquad(\text{and } e_\zL\colon\ 1(\zL)\to G(\zL))
\]
 (the Fr\'echet $\zL_0$-manifold $1(\zL)$ is the singleton that consists of the $\z{\tt Alg}$-morphism $\iota_\zL$ that sends any real number to itself viewed as an element of $\zL$) that define a group structure on~$G(\zL)$ (with unit $1_\zL:=e_\zL(\iota_\zL)$), which is therefore a Fr\'echet $\zL_0$-Lie group. The group properties of these structure maps are consequences of the group properties of the structure maps of $G$. For~instance, when we apply $\cS$ to the associativity equation~\eqref{AssG} and then take the $\zL$-component of the resulting natural transformation, we get
 \[
 \zm_\zL\circ\big(\zm_\zL\times\id_{G(\zL)}\!\big)=\zm_\zL\circ\big(\id_{G(\zL)}\times\zm_\zL\big).
 \]

As for $G(\zf^*)\colon G(\zL)\to G(\zL')$, we know that it is an ${\tt AFM}$-morphism and have to show that it respects the multiplications $\zm_\zL$ and $\zm_{\zL'}$, i.e., that
\begin{gather}\label{MultRespNat}
\zm_{\zL'}\circ(G(\zf^*)\times G(\zf^*))=G(\zf^*)\circ\zm_\zL.
\end{gather}
However, this equality is nothing other than the naturalness property of $\zm_-$.

Finally, let $\Phi\colon G\to G'$ be a $\z{\tt Lg}$-morphism and denote the multiplications of the source and target by $\zm$ and $\zm'$, respectively. In order to prove that the natural transformation $\cS(\Phi)=\Phi_-\colon G(-)\to G'(-)$ of the functor category with target ${\tt AFM}$ is a natural transformation of the functor category with target ${\tt AFLg}$, it suffices to show that $\Phi_\zL$ is a morphism of ${\tt AFLg}$, which results from the application of the functor $\cS$ to the commutative diagram
\begin{gather}\label{MultResp}
\zm'\circ(\Phi\times\Phi)=\Phi\circ\zm.
\end{gather}

The next task is to show that the functor~\eqref{ResSLg} is fully faithful, i.e., that the map \begin{gather}\label{FFResSLg}
\cS_{G,G'}\colon\ \op{Hom}_{\z{\tt Lg}}(G,G')\ni\Phi\mapsto \Phi_-\in\op{Hom}_{{\tt Fun}_0(\z{\tt Pts}^{\op{op}},\,{\tt AFLg})}(G(-),G'(-))
\end{gather}
is a $1:1$ correspondence, for any $\z$-Lie groups $G,G'$. Since the functor~\eqref{SMan} is fully faithful, any natural transformation in the target set of~\eqref{FFResSLg} is implemented by a unique $\z$-morphism $\phi\colon G\to G'$ and it suffices to show that $\phi$ respects the group operations, for instance, that is satisfies equation~\eqref{MultResp}. However, equation~\eqref{MultResp} is satisfied if and only if
\[
\zm'_\zL\circ(\phi_\zL\times\phi_\zL)=\phi_\zL\circ \zm_\zL,
\]
for all $\zL$. The latter condition holds, since $\phi_\zL$ is, by assumption, a group morphism.

We must still prove the last statement of Theorem~\ref{FunPtsZLg}. The assumption implies that, for any $\z$-Grassmann algebra $\zL$ and any $\z$-algebra morphism $\zf^*\colon \zL\to\zL'$, we get a Fr\'echet $\zL_0$-Lie group $M(\zL)$ and a $(\zf^*)_0$-smooth group morphism $M(\zf^*)\colon M(\zL)\to M(\zL')$. We denote by $1_\zL$ (resp., $\zm_\zL,\op{inv}_\zL$) the unit element (resp., the $\zL_0$-smooth multiplication, the $\zL_0$-smooth inverse) of the group structure on the Fr\'echet $\zL_0$-manifold $M(\zL)$. We have already observed (see~\eqref{MultRespNat}) that the fact that $M(\zf^*)$ respects the multiplications $\zm_\zL$ and $\zm_{\zL'}$ is equivalent to that of $\zm_-$ being natural. The natural transformation $\zm_-\colon (M\times M)(-)\to M(-)$ is implemented by a~unique $\z$-morphism $\zm\colon M\times M\to M$. We obtain similarly a $\z$-morphism $\op{inv}\colon M\to M$. As~for $e\colon 1\to M$, we notice that the maps
\[
e_\zL\colon \ 1(\zL)\ni\iota_\zL\mapsto 1_\zL\in M(\zL)\quad (\zL\in\z{\tt GrAlg})
\]
 define visibly a natural transformation with $\zL_0$-smooth $\zL$-components. Hence, it is implemented by a unique $\z$-morphism $e\colon 1\to M$. We leave it to the reader to check that $\zm$, $\op{inv}$ and $e$ satisfy~\eqref{AssG} and the other group properties. \end{proof}

\subsection[The general linear Z2n-group]{The general linear $\boldsymbol{\Z_2^n}$-group}
%\label{SubSec:GLZGrp}
We want to define the {\it general linear $\z$-group} of order $p|\ul q$ so that it is a $\z$-Lie group $\op{GL}\big(p|\ul q\big)$. In view of Theorem~\ref{FunPtsZLg}, it suffices to define a functor
\[
\op{GL}\big(p|\ul q\big)(-)\in{\tt Fun}_0\big(\z{\tt Pts}^{\op{op}},{\tt AFLg}\big)\;
\]
 that is represented by a $\z$-manifold $\op{GL}\big(p|\ul q\big)$.

\begin{Definition}\label{def:GLp|q}
The \emph{general linear $\Z_2^n$-group} $\op{GL}\big(p|\ul q\big)$ is defined, for any $\zL\in\z{\tt GrAlg}$, by
\[
\GL\big(p|\ul{q}\big)(\Lambda) := \op{GL}\big(p|\ul q,\zL\big)=\big\{ X \in \gl_0(p|\ul{q},\zL)\colon X \textnormal{ is invertible} \big \},
\]

and, for any $\z{\tt Alg}$-morphism $\zf^* \colon \Lambda \rightarrow \Lambda'$ and any $X\in\GL\big(p|\ul{q}\big)(\Lambda)$, by
\begin{align*}
\GL\big(p|\ul{q}\big)(\varphi^*)(X):=\tilde\varphi^*X,
\end{align*}
where $\tilde\zf^*$ is $\zf^*$ acting on $X$ entry-by-entry.\end{Definition}

\begin{Theorem}\label{thm:GLRep}
The maps $\op{GL}\big(p|\ul q\big)(-)$ of Definition~$\ref{def:GLp|q}$ define a representable functor. We~refer to the representing object $\op{GL}\big(p|\ul q\big)\in\z{\tt Lg}$ as the \emph{general linear $\z$-group of dimension} $p|\ul q$. \end{Theorem}

\begin{proof}

Recall that:
\begin{enumerate}\itemsep=0pt
\item
It follows from equation~\eqref{Man0Mat} that
\[
\gl_0\big(p|\ul{q},\zL\big)=\Pi_{n=0}^N\zL_{\zg_n}^{\times u_n}=\zL_0^{\times t}\times \Pi_{j=1}^N \zL_{\zg_j}^{\times u_j}\simeq\R^{t|\ul u}(\zL),
\]
 where $u_n$ is given by~\eqref{UN} ($t=u_0$).

\item
It follows from Proposition~\ref{invertibleblockdiag} that $X\in\op{gl}_0\big(p|\ul q,\zL\big)$ is invertible if and only if $\tilde{\ze}(X)\in\op{GL}\big(p|\ul q,\R\big)$, if and only if $\tilde{\ze}(X_{ii})\in\op{GL}(q_i,\R)$, for all $i\in\{0,\dots, N\}$, if and only if $X_{ii}\in\op{GL}(q_i,\zL)$, for all $i\in\{0,\dots, N\}$.\end{enumerate}

In particular, a matrix
\[
X\in\op{gl}_0\big(p|\ul q,\R\big)=\R^t=\R^{p^2+\sum_jq_j^2}=\op{Diag}\big(p|\ul q,\R\big)
\]
is invertible if and only if $X_{ii}\in\op{GL}(q_i,\R)$, for all $i$. It follows that \begin{gather}\label{GL1}
\cU^t:=\op{GL}\big(p|\ul q\big)(\R)=\Pi_{i=0}^N\op{GL}(q_i,\R)\subset \R^t.
\end{gather}
As $\cU^t\subset\R^t$ is open, we can consider the $\z$-domain
\begin{gather}\label{GL2}
\cU^{t|\ul u}:=(\cU^t,\cO_{\R^{t|\ul u}}|_{\cU^t}),
\end{gather}
as well as its functor of points
\[
\cU^{t|\ul u}(-)\in{\tt Fun}_0\big(\z{\tt Pts}^{\op{op}},{\tt AFM}\big),
\]
with value on $\zL$
\[
\cU^{t|\ul u}(\zL)\simeq\cU^t\times\rim{\zL}^{\times t}_0\times\Pi_{j=1}^N\zL_{\zg_j}^{\times u_j}\;
\]
(see~\cite{Bruce:2019b}).

On the other hand, we get
\begin{align*}
\op{GL}\big(p|\ul q\big)(\zL)&=\big\{X\in\R^t\times \rim{\zL}^{\times t}_0\times\Pi_{j=1}^N\zL_{\zg_j}^{\times u_j}\colon (
\dots,\tilde{\ze}(X_{ii}),\dots)\in\Pi_{i=0}^N\op{GL}(q_i,\R)\big\}
\\
&=\cU^t\times\rim{\zL}^{\times t}_0\times\Pi_{j=1}^N\zL_{\zg_j}^{\times u_j},
\end{align*}
so that $\cU^{t|\ul u}(-)$ and $\op{GL}\big(p|\ul q\big)(-)$ ``coincide'' on objects $\zL$: if we denote the coordinates of $\R^{t|\ul u}$ as usually by $(u^{\mathfrak{a}})=\big(x^a,\zx^A\big)$, this ``equality'' reads
\[
\cU^{t|\ul u}(\zL)\ni{\rm x}^*\simeq ({\rm x}^*(u^{\mathfrak{a}}))_{\mathfrak{a}}\in \op{GL}\big(p|\ul q\big)(\zL).
\]
Moreover, $\cU^{t|\ul u}(-)$ and $\op{GL}\big(p|\ul q\big)(-)$ coincide on morphisms $\zf^*\colon \zL\to\zL'$. Indeed, the map $\op{GL}\big(p|\ul q\big)(\zf^*)$ acts on a matrix
\[
({\rm x}^*(u^{\mathfrak{a}}))_{\mathfrak{a}}\in\op{GL}\big(p|\ul q\big)(\zL)\subset\zL^{\times t}_0\times\Pi_{j=1}^N\zL_{\zg_j}^{\times u_j}\;
\]
by acting on all its entries ${\rm x}^*(u^{\mathfrak{a}})$ by $\zf^*$, whereas the map $\cU^{t|\ul u}(\zf^*)$ acts on a $\z{\tt Alg}$-morphism ${\rm x}^*\in\cU^{t|\ul u}(\zL)$ by left composition $\zf^*\circ{\rm x}^*$; if we identify ${\rm x}^*$ with the tuple $({\rm x}^*(u^{\mathfrak{a}}))_{\mathfrak{a}}$, then $\cU^{t|\ul u}(\zf^*)$ acts by acting on each ${\rm x}^*(u^{\mathfrak{a}})$ by $\zf^*$, which proves the claim.

It follows that $\op{GL}\big(p|\ul q\big)(-)$ is a functor
\[
\op{GL}\big(p|\ul q\big)(-)\in{\tt Fun}_0\big(\z{\tt Pts}^{\op{op}},{\tt AFM}\big)
\]
that is represented by
\begin{gather}\label{GL3}
\op{GL}\big(p|\ul q\big):=\cU^{t|\ul u}\in\z{\tt Man},
\end{gather}
so that it now suffices to prove that this functor is valued in ${\tt AFLg}$, i.e., it suffices to show that $\op{GL}\big(p|\ul q\big)(\zL)\in\zL_0{\tt FLg}$ and that $\op{GL}\big(p|\ul q\big)(\zf^*)$ is an ${\tt AFLg}$-morphism.

Recall that $\op{gl}_0\big(p|\ul q,\zL\big)$ is an associative unital $\R$-algebra for the standard matrix multiplication $\cdot$ (standard matrix addition, standard matrix multiplication by reals and standard unit matrix $\mathbb{I}$) (see Section~\ref{subsec:Zmatrices}). It~is clear that the subset $\op{GL}\big(p|\ul q\big)(\zL)\subset\op{gl}_0\big(p|\ul q,\zL\big)$ is closed under~$\cdot$:
\begin{gather}\label{MultGL}
\zm_\zL\colon\ \op{GL}\big(p|\ul q\big)(\zL)\times\op{GL}\big(p|\ul q\big)(\zL)\ni(X,Y)\mapsto X\cdot Y\in\op{GL}\big(p|\ul q\big)(\zL)
\end{gather}
is an associative unital multiplication on $\op{GL}\big(p|\ul q\big)(\zL)$. Therefore, $\zm_\zL$ and
\begin{gather}\label{InvGL}
\op{inv}_\zL\colon\ \op{GL}\big(p|\ul q\big)(\zL)\ni X\mapsto X^{-1}\in \op{GL}\big(p|\ul q\big)(\zL)
\end{gather}
endow $\op{GL}\big(p|\ul q\big)(\zL)$ with a group structure (with unit $\mathbb{I}$). Finally, the Fr\'echet $\zL_0$-manifold $\mathop{\rm GL}\big(p|\ul q\big)(\zL)$ together with its group structure $\zm_\zL,\op{inv}_\zL$ (and $\mathbb{I}$) is a Fr\'echet $\zL_0$-Lie group, if its structure maps $\zm_\zL$ and $\op{inv}_\zL$ are $\zL_0$-smooth. This condition is actually satisfied (see below).

As for $\op{GL}\big(p|\ul q\big)(\zf^*)$, we know that it is an ${\tt AFM}$-morphism and need to show that it respects the multiplications $\zm_\zL$, $\zm_{\zL'}$. This condition is clearly met because $\op{GL}\big(p|\ul q\big)(\zf^*)$ acts entry-wise by the $\z{\tt Alg}$-morphism $\zf^*$.

It remains to explain why $\zm_\zL$ and $\op{inv}_\zL$ are $\zL_0$-smooth.

Notice first that the source of the multiplication~\eqref{MultGL} is the open subset $\zW(\zL):=\cU^{t|\ul u}(\zL)\times\cU^{t|\ul u}(\zL)$ of the Fr\'echet space $F(\zL):=\R^{t|\ul u}(\zL)\times\R^{t|\ul u}(\zL)$ (see~\cite{Bruce:2019b}) and that we can choose the Fr\'echet vector space (and Fr\'echet $\zL_0$-module) $\R^{t|\ul u}(\zL)$ as its target. Since $\zL$ is the ($\z$-commutative nuclear) Fr\'echet $\R$-algebra of global $\z$-functions of some $\z$-point $\R^{0|\ul m}$, its addition and internal multiplication (its multiplication by reals and subtraction) are continuous maps. It~follows that each component function of the standard matrix multiplication $\zm_\zL$ is continuous, so that $\zm_\zL$ is itself continuous. We~must now explain why all directional derivatives of $\zm_\zL$ exist everywhere and are continuous, and why the first derivative is $\zL_0$-linear. Let $(X,Y)\in\zW(\zL)$ and $(V,W)\in F(\zL)$. We get
\[
{\rm d}_{(X,Y)}\zm_\zL\,(V,W)=\lim_{t\to 0}\frac{(X+tV)\cdot(Y+tW)-X\cdot Y}{t}=X\cdot W+V\cdot Y.
\]
Hence, the first derivative exists everywhere, is continuous and $\zL_0$-linear. Indeed, for any ${\rm a}\in\zL_0$, we have
\[
{\rm d}_{(X,Y)}\zm_\zL\,({\rm a}\cdot V,{\rm a}\cdot W)={\rm a}\cdot {\rm d}_{(X,Y)}\zm_\zL\,(V,W).
\]
It is easily checked that
\[
{\rm d}^2_{(X,Y)}\zm_\zL(V_1,W_1,V_2,W_2)=V_2\cdot W_1+V_1\cdot W_2\qquad\text{and}\qquad
{\rm }^{k\ge 3}_{(X,Y)}\zm_\zL(V_1,W_1,\dots,V_k,W_k)=0,
\]
so that $\zm_\zL$ is actually $\zL_0$-smooth.

As for
\[
\op{inv}_\zL\colon\ \cU^{t|\ul u}(\zL)\subset\R^{t|\ul u}(\zL)\to \R^{t|\ul u}(\zL),
\]
we start computing the directional derivative of
\[
\mathbb{I}_\zL:=\zm_\zL \circ (\op{inv}_\zL\times\id_\zL)\circ\zD_\zL\colon\ \cU^{t|\ul u}(\zL)\subset\R^{t|\ul u}(\zL)\ni X\mapsto X^{-1}\cdot X=\mathbb{I}\in\R^{t|\ul u}(\zL)
\]
($\zD_\zL$ is the diagonal map), assuming continuity of $\op{inv}_\zL$, for the time being. For any $V\in\R^{t|\ul u}(\zL)$, we have
\[
{\rm d}_X\mathbb{I}_\zL(V)=\lim_{t\to 0}\frac{(X+tV)^{-1}\cdot (X+tV)-X^{-1}\cdot X}{t}=\lim_{t\to 0}\left(f_{XV}(t)\cdot X +g_{XV}(t)\cdot V\right)=0,
\]
where
\[
f_{XV}(t)=\frac{(X+tV)^{-1}-X^{-1}}{t}\qquad\text{and}\qquad g_{XV}(t)=(X+tV)^{-1}.
\]
It follows that
\begin{align*}
{\rm d}_X\op{inv}_\zL(V)&=\lim_{t\to 0}f_{XV}(t)=\lim_{t\to 0}\big((f_{XV}(t)\cdot X+g_{XV}(t)\cdot V)\cdot X^{-1}-g_{XV}(t)\cdot V\cdot X^{-1}\big)
\\
&=-X^{-1}\cdot V\cdot X^{-1},
\end{align*}
so that the first derivative is defined everywhere, is continuous, as well as $\zL_0$-linear. Also the higher order derivatives exist everywhere and are continuous. For instance, the second order derivative is given by
\begin{align*}
{\rm d}^2_X\op{inv}_\zL(V,W)&=-\lim_{t\to 0}\big(f_{XW}(t)\cdot V\cdot X^{-1}+g_{XW}(t)\cdot V\cdot f_{XW}(t)\big)\nonumber \\ &=X^{-1}\cdot W\cdot X^{-1}\cdot V\cdot X^{-1}+X^{-1}\cdot V\cdot X^{-1}\cdot W\cdot X^{-1}.
\end{align*}
Finally, the inverse map $\op{inv}_\zL$ is $\zL_0$-smooth, provided we prove its still pending continuity.

We will show that the continuity of~\eqref{InvGL} boils down to the continuity of the inverse map $\iota_\zL\colon \zL^{\times}\ni\zl\mapsto \zl^{-1}\in\zL^{\times}$ in $\zL$. Here $\zL^\times\subset\zL$ is the group of invertible elements of $\zL$. Since $\zL$ is a (unital) Fr\'echet $\R$-algebra, its inverse map $\iota_\zL$ is continuous if and only if $\zL^\times$ is a $G_\zd$-set, i.e., if and only if it is a countable intersection of open subsets of $\zL$~\cite{Waelbroeck:1971}. We~will show that $\zL^\times$ is actually open in the specific Fr\'echet $\R$-algebra $\zL$ considered. In view of Equation (16) in~\cite{Bruce:2018}, the topology of $\zL=\R[[\theta]]$ ($\zL\simeq\R^{0|\ul m}$) is induced by the countable family of seminorms
\[
\zr_\zb(\zl)=\frac{1}{\zb!}|\ze(\partial^\zb_\theta \zl)|=|\zl_\zb|\qquad \bigg(\zb\in\mathbb{N}^{\times|\ul m|},\ \zl=\sum_\za \zl_\za \theta^\za\in\zL\bigg),
\]
where $\ze$ is the projection $\ze\colon \zL\to \R$. This means that the topology is made of the unions of finite intersections of the open semiballs
\[
B_\zb(\zn,\ze)=\{\zl\in\zL\colon \zr_\zb(\zl-\zn)<\ze\}=\{\zl\in\zL\colon |\zl_\zb-\zn_\zb|<\ze\}=\{\zl\in\zL\colon \zl_\zb\in b(\zn_\zb,\ze)\}
\]
$\big(\zb\in\mathbb{N}^{\times |\ul m|}$, $\zn=\sum_\za \zn_\za\theta^\za\in\zL$, $\ze>0$ and $b(\zn_\zb,\ze)$ is the open ball in $\R$ with center $\zn_\zb$ and radius~$\ze\big)$. Since
\[
\zL^\times=\{\zl\in\zL\colon \zl_0\in\R\setminus\{0\}\}\qquad\text{and}\qquad \R\setminus\{0\}=\bigcup_{r\in\R\setminus\{0\}} b(r,\ze_r)\qquad (\text{for some}\ \ze_r>0),
\]
we get
\[
\zL^\times=\bigcup_{r\in\R\setminus\{0\}}\{\zl\in\zL\colon \zl_0\in b(r,\ze_r)\}=\bigcup_{r\in\R\setminus\{0\}}B_0(r,\ze_r),
\]
which implies that $\zL^\times$ is open and that $\iota_\zL$ is continuous, as announced.

Before we are able to deduce from this that $\op{inv}_\zL$ is continuous, we need an inversion formula for $X\in\op{GL}\big(p|\ul q\big)(\zL)$. Notice first that, in view of~\cite[Proposition 4.7]{Covolo:2012}, an invertible $2\times 2$ block matrix
\begin{equation*}
X =\begin{pmatrix}
 A & B\\
 C & D
\end{pmatrix}
\end{equation*}
with square diagonal blocks $A$ and $D$ and entries (of all blocks) in a ring, has a block UDL decomposition if and only if $D$ is invertible. In this case, the UDL decomposition is
\begin{gather*} %\label{UDL2}
\begin{pmatrix} A & B \\[2pt] C & D\end{pmatrix}
= \begin{pmatrix} \mathbb{I} & BD^{-1}\\[2pt] 0 & \mathbb{I}\end{pmatrix}
\begin{pmatrix} A-BD^{-1}C & 0 \\[2pt] 0 & D\end{pmatrix}
\begin{pmatrix} \mathbb{I} & 0 \\[2pt] D^{-1}C & \mathbb{I}\end{pmatrix}\!.
\end{gather*}
As upper and lower unitriangular matrices are obviously invertible, it follows that the diagonal matrix is invertible, hence that $A-BD^{-1}C$ is invertible. Similarly, the invertible matrix $X$ has a block LDU decomposition if and only if $A$ is invertible and in this case $D-CA^{-1}B$ is invertible. Moreover, in view of Proposition~\ref{invertibleblockdiag}, a matrix $X\in\op{gl}_0\big(p|\ul q,\zL\big)$ is invertible if and only if all its diagonal blocks $X_{ii}$ are invertible. Let now
\begin{equation*}
X =\begin{pmatrix}
 A & B\\
 C & D
\end{pmatrix}
\end{equation*}
be a $2\times 2$ block decomposition of $X\in\op{gl}_0\big(p|\ul q,\zL\big)$ that respects the $(N+1)\times (N+1)$ block decomposition
\begin{equation*}
X = \left(\begin{array}{c|c|c}
X_{00} & \dotsc & X_{0N} \\
\hline
\vdots & \ddots & \vdots \\
\hline
X_{N0} & \dotsc & X_{NN} \\
\end{array}\right)\!.
\end{equation*}
Since $A$ (resp., $D$) is invertible if and only if
\begin{equation*}
\tilde{A} =\begin{pmatrix}
 A & 0\\
 0 & \mathbb{I}
\end{pmatrix}\quad\quad\left(\text{resp.,}\;\tilde{D} =\begin{pmatrix}
 \mathbb{I} & 0\\
 0 & D
\end{pmatrix}\right)
\end{equation*}
is invertible, hence, if and only if the $X_{kk}$ on the diagonal of $A$ (resp.,~$D$) are invertible, we get that $X$ is invertible if and only if $A$ and $D$ are invertible. If we combine everything we have said so far in this paragraph, we find that if $X\in\op{GL}\big(p|\ul q\big)(\zL)$, then $A,D,A-BD^{-1}C,D-CA^{-1}B$ are all invertible. Therefore, we can use the formula
\begin{equation}\label{eqn:InverForm}
X^{-1} =\begin{pmatrix}
 \big(A - B D^{-1}C\big) ^{-1}& - A^{-1}B\big(D - C A^{-1}B\big)^{-1}\\
 - D^{-1}C\big(A- BD^{-1}C\big)^{-1} & \big(D - C A^{-1}B\big)^{-1}
\end{pmatrix}\! ,
\end{equation}
for any $X\in\op{GL}\big(p|\ul q\big)(\zL)$.

In order to simplify proper understanding, we consider for instance the case $n=2$,
\[
p|\ul q = p|q_1,q_2,q_3 = 1|2,1,1
\]
and
\begin{equation*}
X = \begin{pmatrix}
 A & B\\
 C & D
\end{pmatrix}=\left(\begin{array}{c|c|c|c}
a & b\;c & d & e\\
\hline
f & g\;h & i & j\\
k & l\;m& n& p\\
\hline
q & r\; s & t & u\\
\hline
v & w\;x& y& z\\
\end{array}\right)\in\op{GL}(1|2,1,1)(\zL),\qquad\text{where}\quad A=\left(\begin{array}{c|c}
a & b\;c\\
\hline
f & g\;h\\
k & l\;m\\
\end{array}\right)\!,
\end{equation*}
and so on. We~focus for instance on the first of the four block matrices in $X^{-1}$, i.e., on $\big(A-BD^{-1}C\big)^{-1}$. The matrix $D$ is a $2\times 2$ invertible matrix with square diagonal blocks and entries in~$\zL$. Since the four diagonal block matrices in $X$ are invertible, it follows from what we have said above that the inverse $D^{-1}$ is given by equation~\eqref{eqn:InverForm} with $A=t\in\zL$, $B=u\in\zL$, $C=y\in\zL$ and $D=z\in\zL$. Hence all entries of $D^{-1}$ are composites of the addition, the subtraction, the multiplication and the inverse in $\zL$, and so are all entries in the invertible $2\times 2$ block matrix
\begin{equation}\label{Comp0}
A-BD^{-1}C=\left(\begin{array}{c|c}
\za & \zb\;\zg\\
\hline
\zd & \ze\;\zeta\\
\zh & \zvy\;\zx\\
\end{array}\right)\;
\end{equation}
with square diagonal blocks (which are invertible) and with entries in $\zL$ (the square diagonal blocks have entries in $\zL_0$). Hence, the inverse $(A-BD^{-1}C)^{-1}$ can again be computed by~\eqref{eqn:InverForm}. We~focus on its entry
\begin{gather}\label{Comp1}
\zk:=\left(\za-(\zb\,\zg)\begin{pmatrix}
 \ze & \zeta \\
 \zvy & \zx
\end{pmatrix}^{-1}\begin{pmatrix}
 \zd\\
 \zh
\end{pmatrix}\right)^{-1}\in\zL.
\end{gather}
Notice that here we cannot conclude that $\ze$ and $\zx$ are invertible and apply~\eqref{eqn:InverForm} to compute the internal inverse. However, this inverse is the inverse of a square matrix with entries in the commutative ring $\zL_0$, for which the standard inversion formula holds (recall that a square matrix with entries in a commutative ring is invertible if and only if its determinant is invertible):
\begin{gather}\label{Comp2}\begin{pmatrix}
 \ze & \zeta \\
 \zvy & \zx
\end{pmatrix}^{-1}=(\ze\zx-\zeta\zvy)^{-1}\begin{pmatrix}
 \zx & -\zeta \\
 -\zvy & \ze
\end{pmatrix}\!.
\end{gather}
Since all the entries of~\eqref{Comp0} are composites of the addition, subtraction, multiplication and inverse in $\zL$, it follows from~\eqref{Comp1} and~\eqref{Comp2} that the same is true for the entry $\zk$ of $X^{-1}$. More precisely the entry $\zk$ corresponds to a map $\tilde{\zk}$ that is a composite of the inclusion of~$\op{GL}\big(p|\ul q\big)(\zL)$ into its topological supspace $\zL^{\times (t+|\ul u|)}$ (continuous), the projection of $\zL^{\times(t+|\ul u|)}$ onto $\zL^{\times v}$ ($v\le t+|\ul u|$) (continuous) and of products of the identity map $\id$ of $\zL$ (continuous), the diagonal map~$\zD$ of $\zL$ (continuous), the switching map $\zs$ of $\zL\times\zL$ (continuous), the addition $a$ of $\zL$ (continuous), the scalar multiplication $e$ of $\zL$ (continuous), its subtraction $s$ (continuous), multiplication $m$ (continuous) and its inverse $\iota$ (continuous). Indeed, it is for instance easily seen that the map
\[
\zL^{\times 4}\ni(t,u,y,z)\mapsto -z^{-1}y\big(t-uz^{-1}y\big)^{-1}\in\zL
\]
is a (continuous) composite of products of these continuous maps. We~thus understand that the entry $\zk$ of $X^{-1}$ corresponds to a continuous map $\tilde\zk\colon \op{GL}\big(p|\ul q\big)(\zL)\to \zL$. The same holds of course also for all the other entries of $X^{-1}$. Finally, the inverse map
\[
\op{inv}_\zL\colon\ \op{GL}\big(p|\ul q\big)(\zL)\ni X\mapsto X^{-1}\in\zL^{\times(t+|\ul u|)}
\]
is continuous and it remains continuous when view as valued in the subspace $\op{GL}\big(p|\ul q\big)(\zL)$.
\end{proof}

\begin{Example}
In view of equations~\eqref{GL3},~\eqref{GL2} and~\eqref{GL1}, the general linear $\Z^2_2$-group of order $1|1,1,1$ is
\[
\op{GL}(1|1,1,1) \simeq \left((\R^\times)^4, \cO_{\R^{4|4,4,4}}|_{(\R^\times)^4}\right),
\]
where $\R^\times=\R\setminus\{0\}$. \end{Example}
\subsection{Smooth linear actions}\label{SubSec:SmoLinAct} In this section we define linear actions of $\z$-Lie groups $G$ on finite dimensional $\z$-vector spaces $\mbf{V}\simeq V$ (we identify the isomorphic categories $\z{\tt FinVec}$ and $\z{\tt LinMan}$). The definition can be given in the category of $\z$-manifolds, but it is slightly more straightforward if we use the functor of points. Notice that the functors of points of $G\in\z{\tt Lg}\subset\z{\tt Man}$ and $V\in\z{\tt LinMan}\subset\z{\tt Man}$ are functors
\[
\cS(G)=G(-)\in{\tt Fun}_0\big(\z{\tt Pts}^{\op{op}},{\tt AFLg}\big)\subset{\tt Fun}_0\big(\z{\tt Pts}^{\op{op}},{\tt AFM}\big)
\]
and
\[
\cS(V)=V(-)\in{\tt Fun}_0\big(\z{\tt Pts}^{\op{op}},{\tt FAMod}\big)\subset{\tt Fun}_0\big(\z{\tt Pts}^{\op{op}},{\tt AFM}\big).
\]

\begin{Definition}\label{SmoLinActDef}
Let $G\in\z{\tt Lg}$ and $V\in\z{\tt LinMan}$. A \emph{smooth linear action} of $G$ on $V$ is a~natural transformation
\[
\sigma_{-} \colon\ (G\times V)(-)=G(-) \times V(-) \rightarrow V(-)
\]
in ${\tt Fun}_0\big(\z{\tt Pts}^{\op{op}},{\tt AFM}\big)$ (natural transformation with $\zL_0$-smooth $\zL$-components) that satisfies the following conditions:
\begin{enumerate}\itemsep=0pt
\item[$(i)$] Identity: for all $v_\zL\in V(\zL)$, we have
\[
\sigma_\Lambda(1_\Lambda, v_\Lambda) = v_\Lambda,
\]
where $1_\Lambda$ is the unit of $G(\zL)$. %
\item[$(ii)$] Compatibility: for all $g_\Lambda, g'_\Lambda \in G(\Lambda)$ and all $v_\zL\in V(\zL)$, we have
\[
\sigma_\Lambda\big(g_\Lambda, \sigma_\Lambda(g'_\Lambda, v_\Lambda)\big) = \sigma_\Lambda\big(\zm_\zL(g_\Lambda, g'_\Lambda),v_\Lambda\big),
\]
where $\zm_\zL$ is the multiplication of $G(\zL)$. %
\item[$(iii)$] $\Lambda_0$-linearity: for all $g_\zL\in G(\zL)$, all $v_\Lambda, v'_\Lambda \in V(\Lambda)$ and all $\mathrm{a} \in \Lambda_0$, we have
\begin{enumerate}\itemsep=0pt
\item[$(a)$] $\sigma_\Lambda(g_\Lambda, v_\Lambda + v'_\Lambda) = \sigma_\Lambda(g_\Lambda, v_\Lambda) + \sigma_\Lambda(g_\Lambda, v'_\Lambda)$,
\item[$(b)$] $\sigma_\Lambda (g_\Lambda, \mathrm a\cdot v_\Lambda) = \mathrm{a} \cdot \sigma_\Lambda(g_\Lambda , v_\Lambda)$,
 \end{enumerate}
where $\cdot$ is the action of $\zL_0$ on $V(\zL)$. \end{enumerate}
\end{Definition}

Since
\begin{gather}\label{FunPts1}
\cS\colon\ \z{\tt Man}\to{\tt Fun}_0\big(\z{\tt Pts}^{\op{op}},{\tt AFM}\big)
\end{gather}
is fully faithful (for more details, see~\cite{Bruce:2019b,Bruce:2018,Bruce:2019}), there is a $1:1$ correspondence between natural transformations $\zs_-$ as above and $\z$-morphisms
\[%\label{Act0}
\zs\colon\ G\times V\to V.
\]
This correspondence implies in particular that condition~($ii$) is equivalent to the equality
\begin{gather}\label{Act1}
\zs\circ(\id_G\times\zs)=\zs\circ(\zm\times\id_V)
\end{gather}
of $\z$-morphisms from $G\times G\times V\to V$ ($\zm\colon G\times G\to G$ is the multiplication of $G$). The same holds for condition~($i$) and the equality
\begin{gather}\label{Act2}
\zs\circ(e\times \id_V)=\id_V
\end{gather}
of $\z$-morphisms from $V\simeq 1\times V\to V$ ($e\colon 1\to G$ is the two-sided unit of $\zm$).

\subsubsection{Canonical action of the general linear group} We will now define the canonical action of the general linear $\z$-group $\op{GL}\big(p|\ul q\big)=\mathcal{U}^{t|\ul u}\in\z{\tt Lg}$ on the Cartesian $\z$-manifold $\R^{p|\ul q}\in\z{\tt LinMan}$. To do this, we use both, the fully faithful functor~\eqref{FunPts1} and the fully faithful functor
\begin{gather}\label{FunPts2}
\mathcal{Y}\colon\ \z{\tt Man}\ni M\mapsto\op{Hom}_{\z{\tt Man}}(-,M)\in{\tt Fun}\big(\z{\tt Man}^{\op{op}},{\tt Set}\big).
\end{gather}

We start defining a natural transformation $\zs_-$ of ${\tt Fun}\big(\z{\tt Man}^{\op{op}},{\tt Set}\big)$ from $\mathcal{U}^{t|\ul u}(-)\times\R^{p|\ul q}(-)$ to $\R^{p|\ul q}(-)$. We~will denote the coordinates of $\mathcal{U}^{t|\ul u}$ (resp., $\R^{p|\ul q}$) here by $\mathcal{X}^{\mathfrak{a}}_\mathfrak{b}$ (resp., $\mathcal{x}^{\mathfrak{c}}$), where $\mathfrak{a},\mathfrak{b}\in\{1,\dots,p+|\ul q|\}$ (resp., where $\mathfrak{c}\in\{1,\dots,p+|\ul q|\}$). For this, we must associate to any $S\in\z{\tt Man}$, a set-theoretical map $\zs_S$ that assigns to any
\[
(X,\zvf)\in\mathcal{U}^{t|\ul u}(S)\times\R^{p|\ul q}(S)=\op{Hom}_{\z{\tt Man}}\big(S,\mathcal{U}^{t|\ul u}\big)\times\op{Hom}_{\z{\tt Man}}\big(S,\R^{p|\ul q}\big),
\]
i.e., to any (appropriate) coordinate pullbacks
\begin{gather}\label{CoorPull}
\big(\mathcal{X}^\mathfrak{a}_{S,\mathfrak{b}},\mathcal{x}_S^{\mathfrak{c}}\big):= (X^*(\mathcal{X}^\mathfrak{a}_{\mathfrak{b}}),\zvf^*(\mathcal{x}^{\mathfrak{c}}))
\in\cO(S)^{\times (p+|\ul q|)^2}\times\cO(S)^{\times(p+|\ul q|)},
\end{gather}
a unique element $\zs_S(X,\zvf)\in\R^{p|\ul q}(S)$, i.e., unique (appropriate) coordinate pullbacks
\[
\zs_S\big(\mathcal{X}^\mathfrak{a}_{S,\mathfrak{b}},\mathcal{x}_S^{\mathfrak{c}}\big)\in\cO(S)^{\times(p+|\ul q|)}.
\]
Since $\big(\mathcal{x}_S^{\mathfrak{c}}\big)_{\mathfrak{c}}$ is viewed as a tuple (horizontal row), the natural definition of this image (horizontal row) is
\begin{gather}\label{NatTraCanAct}
\zs_S\big(\mathcal{X}^\mathfrak{a}_{S,\mathfrak{b}},\mathcal{x}_S^{\mathfrak{c}}\big) =\big(\mathcal{x}_S^{\mathfrak{b}}\;\mathcal{X}^{\mathfrak{a}}_{S,\mathfrak{b}}\big)_\mathfrak{a},
\end{gather}
where the sum and products are taken in the global $\z$-function algebra $\cO(S)$ of $S$. It~is clear that the elements of this target-tuple have the required degrees, as the same holds for the elements of~the source-tuple. The transformation $\zs_-$ we just defined is clearly natural. Indeed, for any $\z$-morphism $\psi\colon S'\to S$, the induced set-theoretical mapping between the Hom-sets with source~$S$ and the corresponding ones with source~$S'$ is $-\circ \psi$, so that the induced set-theoretical mapping between the tuples of global $\z$-functions of $S$ and $S'$ is $\psi^*$. The naturalness of $\zs_-$ follows now from the fact that $\psi^*$ is a $\z{\tt Alg}$-morphism.

Since~\eqref{FunPts2} is fully faithful, the natural transformation $\zs_-$ is implemented by a unique $\z$-morphism
\begin{gather}\label{zMorAct}
\zs\colon\ \op{GL}\big(p|\ul q\big)\times\R^{p|\ul q}\to\R^{p|\ul q},
\end{gather}
which in turn implements, via~\eqref{FunPts1}, a unique natural transformation in ${\tt Fun}_0\big(\z{\tt Pts}^{\op{op}},{\tt AFM}\big)$ between the same functors, but restricted to $\z{\tt Pts}^{\op{op}}$. Since this transformation is the restriction of $\zs _-$ to $\z{\tt Pts}^{\op{op}}$, we use this symbol for both transformations (provided that any confusion can be excluded). It~is easily seen that
\[
\zs_\zL\big(\mathcal{X}^{\mathfrak{a}}_{\zL,\mathfrak{b}},\mathcal{x}_\zL^{\mathfrak{c}}\big) =\big(\mathcal{x}_\zL^{\mathfrak{b}}\;\mathcal{X}^{\mathfrak{a}}_{\zL,\mathfrak{b}}\big)_\mathfrak{a},
\]
with sum and products in $\zL$, has the properties ($i$), ($ii$) and ($iii$) of Definition~\ref{SmoLinActDef}, so that we defined a smooth linear action of $\op{GL}\big(p|\ul q\big)$ on $\R^{p|\ul q}$.

The interesting aspect here is that we are able to compute the $\z$-morphism~\eqref{zMorAct}. Indeed, in view of the proof of the full faithfulness of the standard Yoneda embedding $c\mapsto \op{Hom}_{\tt C}(-,c)$ of an arbitrary locally small category $\tt C$ into the functor category ${\tt Fun}\big({\tt C}^{\op{op}},{\tt Set}\big)$, the morphism $\zs\in\op{Hom}_{\tt C}(c,c')$ that implements a natural transformation
\[
\zs_-\colon\ \op{Hom}_{\tt C}(-,c)\to \op{Hom}_{\tt C}(-,c')
\]
is
\[
\zs=\zs_c(\id_c)\in\op{Hom}_{\tt C}(c,c').
\]
In our case of interest ${\tt C}=\z{\tt Man}$, the previous Yoneda embedding is the functor~\eqref{FunPts2} and the morphism
\[
\zs\in\op{Hom}_{\z\tt Man}\big(\!\op{GL}\big(p|\ul q\big)\times\R^{p|\ul q},\R^{p|\ul q}\big)
\]
is
\[
\zs=\zs_c(\id_c),\qquad\text{with}\quad c=\op{GL}\big(p|\ul q\big)\times\R^{p|\ul q}.
\]
Since the pullback of the identity $\z$-morphism $\id_c$ is identity and the coordinate pullbacks~\eqref{CoorPull} are
\[
\big(\mathcal{X}^{\mathfrak{a}}_{\mathfrak{b}},\mathcal{x}^{\mathfrak{c}}\big)\in\cO(c)^{\times(p+|\ul q|)(p+|\ul q|+1)}.
\]
Equation~\eqref{NatTraCanAct} yields
\[
\zs=\zs_c(\id_c)\simeq\zs_c\big(\mathcal{X}^{\mathfrak{a}}_{\mathfrak{b}},\mathcal{x}^{\mathfrak{c}}\big) =\big(\mathcal{x}^{\mathfrak{b}}\,\mathcal{X}^{\mathfrak{a}}_{\mathfrak{b}}\big)_{\mathfrak{a}},
\]
with sum and products in $\cO(c)$. In other words:

\begin{Proposition}
The canonical action $\zs$ of the general linear $\z$-group $\op{GL}\big(p|\ul q\big)$ on the linear $\z$-manifold or $\z$-graded vector space $\R^{p|\ul q}$, is the $\z$-morphism that is defined by the coordinate pullbacks
\begin{gather}\label{zActGLR}
\zs^*(\mathcal{x}^{\mathfrak{a}})=\mathcal{x}^{\mathfrak{b}}\,\mathcal{X}^{\mathfrak{a}}_{\mathfrak{b}},
\end{gather}
where we denoted the coordinates of $\op{GL}\big(p|\ul q\big)$ $($resp., $\R^{p|\ul q})$ by $\mathcal{X}^{\mathfrak{a}}_\mathfrak{b}$ $($resp., $\mathcal{x}^{\mathfrak{c}})$.
\end{Proposition}

\begin{Example}
We know that the general linear $\Z_2^2$-group $\GL(1|1,1,1)$ can be identified with the open $\Z^2_2$-submanifold $\mathcal{U}^{4|4,4,4}$ of $\R^{4|4,4,4}$. We~denote the global coordinates of this Cartesian $\Z^2_2$-manifold by $\big(x^\alpha, \zx^\beta, \theta^\gamma,z^\delta\big)$. The indices run over $\{1,2,3,4\}$ and the $\Z_2^2$-degrees of these coordinates are $(0,0)$, $(0,1)$, $(1,0)$ and $(1,1)$, respectively. We~already mentioned that if we view $\mathcal{U}^{4|4,4,4}$ as $\op{GL}(1|1,1,1)$, we must rearrange the coordinates:
\[
\mathcal{X}=(\mathcal{X}^{\mathfrak{a}}_{\mathfrak{b}})_{\mathfrak{a},\mathfrak{b}}=\left(
\begin{matrix}
 x^1 & \zx^1 & \theta^1 & z^1 \\
 \zx^2 &x^2 & z^2 & \theta^2 \\
 \theta^3 & z^3 & x^3 & \zx^3\\
 z^4 & \theta^4& \zx^4 & x^4\\
\end{matrix}
\right)\!.
\]
In view of~\eqref{zActGLR}, the action $\zs$ of $\op{GL}(1|1,1,1)$ on $\R^{1|1,1,1}$ with global co\-or\-di\-nates
\[
\mathcal{x}=\big(\mathcal{x}^{\mathfrak{a}}\big)_{\mathfrak{a}}=\big(x^0, \zx^0, \theta^0 , z^0\big),
\]
is given as
\begin{align*}
& \sigma^*\big(x^0\big) = x^0x^1 + \zx^0 \zx^1 + \theta^0 \theta^1 + z^0 z^1,\\
& \sigma^*\big(\zx^0\big) = x^0 \zx^2 + \zx^0 x^2 + \theta^0 z^2 + z^0 \theta^2,\\
& \sigma^*\big(\theta^0\big) = x^0 \theta^3 + \zx^0 z^3 + \theta^0 x^3 + z^0 \zx^3,\\
& \sigma^*\big(z^0\big) = x^0 z^4 + \zx^0 \theta^4 + \theta^0 \zx^4 + z^0 x^4.
\end{align*}
\end{Example}

\subsubsection{Connection between the canonical action and the internal Hom} Since
\[
\op{GL}\big(p|\ul q\big)=\mathcal{U}^{t|\ul u}
\]
(see equation~\eqref{GL3}) is an open $\z$-submanifold (see equation~\eqref{GL2}) of
\[
\op{gl}_0\big(p|\ul q\big)=\R^{t|\ul u}=\ul{\op{Hom}}_{\z{\tt LinMan}}\big(\R^{p|\ul q},\R^{p|\ul q}\big)
\]
(see Proposition~\ref{gl0} and equation~\eqref{IntHomLinManCarMan}), we can expect a connection between the canonical action of $\op{GL}\big(p|\ul q\big)$ on $\R^{p|\ul q}$ and $\ul{\op{Hom}}_{\z{\tt LinMan}}\big(\R^{p|\ul q},\R^{p|\ul q}\big)$. It turns out that this link becomes apparent as soon as we understand the connection between the internal Hom of linear $\z$-manifolds and the internal Hom of arbitrary $\z$-manifolds. Indeed, for any $\zL\simeq\R^{0|\ul m}$, we have
\[
\ul{\op{Hom}}_{\z{\tt Man}}\big(\R^{p|\ul q},\R^{p|\ul q}\big)(\zL):=\op{Hom}_{\z{\tt Man}}\big(\R^{p|\ul q},\R^{p|\ul q}\big)
\]
(see~\cite{Bruce:2019b}). If we denote the coordinates of $\R^{0|\ul m}$ by $\zvy=(\zvy^{\mathfrak{d}})$ and those of $\R^{p|\ul q}$ by $\mathcal{x}=(\mathcal{x}^{\mathfrak{a}})=\big(x^a,\zx^A\big)$, the {\small RHS} Hom-set can be identified with the set of (degree respecting) coordinate pullbacks:
\[
\ul{\op{Hom}}_{\z{\tt Man}}\big(\R^{p|\ul q},\R^{p|\ul q}\big)(\zL)=\bigg\{\mathcal{x}^{\mathfrak{a}}=\mathcal{x}^{\mathfrak{a}}(x,\zx,\zvy) =\sum_{\za\zb}f^{\mathfrak{a}}_{\za\zb}(x)\zx^\za\zvy^\zb\bigg\}.
\]
On the other hand, when denoting the coordinates of $\R^{t|\ul u}$ as above by $\mathcal{X}=(\mathcal{X}^{\mathfrak{a}}_{\mathfrak{b}})$, we get similarly
\begin{align}
\ul{\op{Hom}}_{\z{\tt LinMan}}\big(\R^{p|\ul q},\R^{p|\ul q}\big)(\zL)&=\R^{t|\ul u}(\zL)=\op{Hom}_{\z{\tt Man}}\big(\zL,\R^{t|\ul u}\big)\nonumber
\\
&=\bigg\{\mathcal{X}^{\mathfrak{a}}_{\mathfrak{b}}=\mathcal{X}^{\mathfrak{a}}_{\mathfrak{b}}(\zvy)=\sum_\zd r^{\mathfrak{a}}_{\mathfrak{b},\zd}\,\zvy^\zd\bigg\}
=\op{gl}_0\big(p|\ul q,\zL\big).\label{IntHomZMan}
\end{align}
An obvious identification leads now to
\begin{gather}
\ul{\op{Hom}}_{\z{\tt LinMan}}\big(\R^{p|\ul q},\R^{p|\ul q}\big)(\zL)\nonumber
\\ \qquad
{}=\bigg\{\mathcal{x}^{\mathfrak{a}}=\mathcal{x}^{\mathfrak{a}}(x,\zx,\zvy)
=\sum_{\mathfrak{b}}\mathcal{x}^{\mathfrak{b}}\,\mathcal{X}^{\mathfrak{a}}_{\mathfrak{b}}(\zvy)=
\sum_{b}x^b\,\mathcal{X}^{\mathfrak{a}}_{b}(\zvy)+\sum_B\zx^B\,\mathcal{X}^{\mathfrak{a}}_{B}(\zvy)\bigg\}.
\label{IntHomLinZMan}
\end{gather}
When comparing~\eqref{IntHomLinZMan} and~\eqref{IntHomZMan}, we see that the internal Hom of linear $\z$-manifolds consists of the pullbacks of the internal Hom of arbitrary $\z$-manifolds which are defined by the canonical action of $\op{gl}_0\big(p|\ul q\big)(\zL)$ on $\R^{p|\ul q}$, in the sense of~\eqref{zActGLR}.

\subsubsection{Equivalent definitions of a smooth linear action}
Section~\ref{SubSec:SmoLinAct} already implicitly contained the idea that a smooth linear action of a $\z$-Lie group~$G$ on a linear $\z$-manifold $V$ in the sense of Definition~\ref{SmoLinActDef}, is equivalent to a $\z$-morphism $\zs\colon G\times V\to V$ that satisfies the conditions~\eqref{Act1} and~\eqref{Act2} and additionally has a certain linearity property with respect to $V$. A natural idea is that $\zs^*$ should send linear $\z$-functions of $V$ to $\z$-functions of $G\times V$ that are linear along the fibers. The meaning of this concept becomes clear when we think of the classical differential geometric case in which the functions of a trivial vector bundle $E=M\times\R^r$ are
\[
\Ci(E)=\zG(\vee E^*)=\Ci(M)\otimes\,\vee\,(\mbf{R}^r)^*
\]
($\vee$ is the symmetric tensor product), i.e., are the functions that are smooth in the base and polynomial along the fiber. Hence, linear functions of $E$ are the functions that are smooth in~the base and linear along the fiber, i.e.,
\[
\Ci_{\op{lin}}(E)=\Ci(M)\otimes(\mbf{R}^r)^*=\Ci(M)\otimes\Ci_{\op{lin}}(\R^r).
\]
We can choose the same definition in the case of the trivial $\z$-vector bundle $E=G\times V$:
\[
\cO^{\op{lin}}_{E}(|G|\times|V|):=\cO_G(|G|)\otimes\cO^{\op{lin}}_V(|V|).
\]
This definition is of course in particular valid for $G=\op{GL}\big(p|\ul q\big)\in\z{\tt Lg}$. However, let us mention that the linear functions (``linear along the fibers'') of the trivial $\z$-vector bundle $E=$ \mbox{$\op{GL}\big(p|\ul q\big)\times V$} that are defined on $|\op{GL}\big(p|\ul q\big)|\times|V|$ do not coincide with the linear functions (``globally linear'') of the linear $\z$-manifold $M=\R^{t|\ul u}\times V$ (see~\eqref{TrafroVec}) that are defined on the open subset $|\op{GL}\big(p|\ul q\big)|\times |V|$ of its base $\R^t\times|V|$:
\[
\cO^{\op{lin}}_E\big(\big|\op{GL}\big(p|\ul q\big)\big|\times|V|\big)\neq \cO_M^{\op{lin}}\big(\big|\op{GL}\big(p|\ul q\big)\big|\times |V|\big).
\]
Given what we have just said, we expect the following proposition to hold:

\begin{Proposition}\label{EquivSmoLinActFunPtsMan} A smooth linear action $\zs_-$ of the $\z$-Lie group $G=\op{GL}\big(p|\ul q\big)$ on a linear $\z$-manifold $V$ in the sense of Definition~$\ref{SmoLinActDef}$, is equivalent to a $\z$-morphism $\zs\colon G\times V\to V$ that satisfies the conditions~\eqref{Act1} and~\eqref{Act2} and has the linearity property
\begin{gather}\label{LinCondActZMan}
\zs^*(\cO_V^{\op{lin}}(|V|))\subset \cO_{G}(|G|)\otimes\cO_V^{\op{lin}}(|V|).
\end{gather}
\end{Proposition}

Notice first that the pullback $\zs^*$ is a morphism of $\z$-algebras
\[
\zs^*\colon\ \cO_V(|V|)\to\cO_{G\times V}(|G|\times|V|).
\]
Since $G$ and $V$ have global coordinates, it follows from~\cite{Bruce:2019} that the target of $\zs^*$ is given by
\[
\cO_{G\times V}(|G|\times|V|)=\cO_G(|G|)\widehat{\otimes}\,\cO_V(|V|),
\]
which shows that it contains
\[
\cO_G(|G|)\0\,\cO_V^{\op{lin}}(|V|)
\]
and that the requirement~\eqref{LinCondActZMan} actually makes sense.

Another fact is also worth noting. We~know from standard supergeometry that the classical Berezinian defines a super-Lie group morphism
\[
\op{Ber}\colon\ \op{GL}(p|q)\to\op{GL}(1|0),
\]
so that we get a linear action of $\op{GL}(p|q)$ on $\R^{1|0}$. The point here is that linear actions of $\op{GL}(p|q)$ are not limited to actions on $\R^{p|q}$.
\begin{proof}
In the light of the observations that follow Definition~\ref{SmoLinActDef}, it suffices to prove that the $\zL_0$-linearity requirement ($iii$) in~Definition~\ref{SmoLinActDef} is equivalent to the linearity condition~\eqref{LinCondActZMan} in~Proposition~\ref{EquivSmoLinActFunPtsMan}. Hence, let $\zs_-$ be a smooth action of $G$ on $V$ and let $\zs$ be the corresponding $\z$-morphism. If~$\rmh\colon V\to\R^{r|\ul s}$ is a linear coordinate map of $V$, the $\z$-morphism
\[
\mathfrak{S}:=\rmh\circ\zs\circ\big(\id_G\times\rmh^{-1}\big)\colon\ G\times\R^{r|\ul s}\to \R^{r|\ul s}
\]
satisfies~\eqref{LinCondActZMan} if and only if the $\z$-morphism $\zs$ does. Indeed, if $\zs$ has the property~\eqref{LinCondActZMan}, then
\[
\mathfrak{S}^*=\big(\id_G\times\rmh^{-1}\big)^*\circ\zs^*\circ\rmh^* =\big(\id_G^*\widehat{\0}\,\big(\rmh^{-1}\big)^*\big)\circ\zs^*\circ\rmh^*
\]
has obviously the same property. We~similarly find that the converse implication holds. On the other hand, if we denote the coordinates of $\R^{r|\ul s}$ by $\mathcal{y}=(\mathcal{y}^{\mathfrak{c}})=\big(y^c,\eta^C\big)$, the $\zL$-components of the natural transformations $\zs_-$ and $\mathfrak{S}_-$ satisfy
\[
\mathfrak{S}_\zL=\rmh_\zL\circ\zs_\zL\circ\big(\id_{G(\zL)}\times\rmh_\zL^{-1}\big)
\]
and
\[
\mathfrak{S}_\zL\bigg(g_\zL,\sum_k\zl^k\mathcal{y}_{\zL,k}\bigg) =\rmh_\zL\bigg(\zs_\zL\bigg(g_\zL,\rmh_\zL^{-1}\bigg(\sum_k\zl^k\mathcal{y}_{\zL,k}\bigg)\bigg)\bigg),
\]
for any $g_\zL\in G(\zL)$, any $\mathcal{y}_{\zL,k}\in\R^{r|\ul s}(\zL)$ and any $\zl^k\in\zL_0$ (where $k$ runs through a finite set). Since
\[
\rmh_\zL\colon\ V(\zL)\to\R^{r|\ul s}(\zL)
\]
is an isomorphism of Fr\'echet $\zL_0$-modules, the $\zL_0$-smooth map $\mathfrak{S}_\zL$ is $\zL_0$-linear in $\mathcal{y}_\zL$ if and only if the $\zL_0$-smooth map $\zs_\zL$ is $\zL_0$-linear in $v_\zL\in V(\zL)$. It~is therefore sufficient to prove the equivalence ``($iii$) if and only if~\eqref{LinCondActZMan}'' for $V=\R^{r|\ul s}$.

We refrain from writing down the proof of the implication ``if ($iii$) then~\eqref{LinCondActZMan}''. It~is technical and partially reminiscent of a part of the proof of Theorem~\ref{FunPtsLinManFF} (for the super-case, see~\cite{Carmeli:2018} and the references it contains).

We now prove the converse implication from scratch. Assume that
\begin{gather}\label{LinAssumption}
\mathfrak{S}^*\big(\cO_{\R^{r|\ul s}}^{\op{lin}}(\R^r)\big)\subset \cO_{G}(|G|)\otimes\cO_{\R^{r|\ul s}}^{\op{lin}}(\R^r).
\end{gather}
In view of the universal property of the product of $\z$-manifolds, we have
\[
G(\zL)\times \R^{r|\ul s}(\zL)\ni(g_\zL,\mathcal{y}_\zL)\simeq u_\zL\in\big(G\times \R^{r|\ul s}\big)(\zL).
\]
If we identify the $\z$-morphisms $g_\zL$, $\mathcal{y}_\zL$, $u_\zL$ with the corresponding continuous $\z$-algebra morphisms
\begin{align*}
g_\zL^*&\in\op{Hom}_{\z\tt Alg}\big(\cO_G(|G|),\zL\big),\mathcal{y}^*_\zL
\in\op{Hom}_{\z\tt Alg}\big(\cO_{\R^{r|\ul s}}(\R^r),\zL\big),u_\zL^*
\\
&\in\op{Hom}_{\z\tt Alg}\big(\cO_G(|G|)\widehat{\0}\,\cO_{\R^{r|\ul s}}(\R^r),\zL\big),
\end{align*}
we get
\[
u_\zL^*=\widehat{m}_\zL\circ\big(g_\zL^*\widehat\0\,\mathcal{y}_\zL^*\big),
\]
where $\widehat{m}_\zL\colon \zL\widehat{\0}\,\zL\to\zL$ is continuous $\z$-algebra morphism that extends the multiplication \mbox{$m_\zL=\cdot_\zL$} of $\zL$ (see~\cite{Bruce:2019}). We~denote the coordinates of $G\times \R^{r|\ul s}$ as $\mathcal{z}=(\mathcal{z}^{\mathfrak{d}})=(\mathcal{X}^{\mathfrak{a}}_{\mathfrak{b}},\mathcal{y}^{\mathfrak{c}})$. Since
\begin{align*}
\mathfrak{S}_\zL\colon\ \big(G\times\R^{r|\ul s}\big)(\zL)\ni u_\zL&\simeq (g_\zL,\mathcal{y}_\zL)\simeq u^*_\zL\simeq u^*_\zL(\mathcal{z}^{\mathfrak{d}})\mapsto \mathfrak{S}\circ u_\zL\simeq u_\zL^*\circ\mathfrak{S}^*
\\
&\simeq u_\zL^*(\mathfrak{S}^*(\mathcal{y}^{\mathfrak{c}}))\in\R^{r|\ul s}(\zL),
\end{align*}
we find that
\[
\mathfrak{S}_\zL\bigg(g_\zL,\sum_k\zl^k\mathcal{y}_{\zL,k}\bigg)\simeq\widehat{m}_\zL \bigg(\bigg(g_\zL^*\widehat{\0}\bigg(\sum_k\zl^k\mathcal{y}_{\zL,k}\bigg)^*\bigg)\big(\mathfrak{S}^*
(\mathcal{y}^{\mathfrak{c}})\big)\bigg).
\]
In view of the definition of the $\zL_0$-module structure on $\R^{r|\ul s}(\zL)$, we have
\[
\bigg(\sum_k\zl^k\mathcal{y}_{\zL,k}\bigg)^*=\sum_k\zl^k\mathcal{y}_{\zL,k}^*
\]
and in view of the assumption~\eqref{LinAssumption}, we get for any {\it fixed} $\mathfrak{c}$ that
\[
\mathfrak{S}^*(\mathcal{y}^{\mathfrak{c}}) =\sum_{n=1}^{M}s_G^n\0_\R\bigg(\sum_{\mathfrak{c'}}r^n_{\mathfrak{c'}}\mathcal{y}^{\mathfrak{c'}}\bigg),
\]
where $M\in\mathbb{N}$, where $s_G^n\in\cO_G(|G|)$ and where $r^n_{\mathfrak{c'}}\in\R$. Since $\zl^k\in\zL_0$, what we just said yields
\begin{align*}
\mathfrak{S}_\zL\bigg(\!g_\zL,\sum_k\zl^k\mathcal{y}_{\zL,k}\bigg)& \simeq\sum_{n=1}^{M}g^*_\zL(s_G^n) \cdot_\zL\sum_k\zl^k\cdot_\zL\!\bigg(\!\sum_{\mathfrak{c'}}r^n_{\mathfrak{c'}}\,
\mathcal{y}_{\zL,k}^*(\mathcal{y}^{\mathfrak{c'}})\bigg)\\
& \simeq\sum_k\zl^k\, \mathfrak{S}_\zL(g_\zL,\mathcal{y}_{\zL,k}). \tag*{\qed}
\end{align*}
\renewcommand{\qed}{}
\end{proof}

\section{Future directions}
We view the current paper as the first steps towards understanding actions of $\Z_2^n$-Lie groups on~$\Z_2^n$-manifolds and we claim that it will be vital in carefully constructing the total spaces of~$\Z_2^n$-vector bundles, for example. In both these settings, the functor of points, and in particular $\Lambda$-points, are expected to be of fundamental importance. In particular, the typical fibres of~$\Z_2^n$-vector bundles cannot be $\Z_2^n$-graded vector spaces, but rather they are linear $\Z_2^n$-manifolds. Moreover, the transition functions will correspond to an action of the general linear $\Z_2^n$-group and as such a careful understanding of linear actions is needed. This paper provides some of this technical background. We~plan to explore the algebraic and geometric definitions of vector bundles in the category of $\Z_2^n$-manifolds in a future publication.

\appendix
\section{The category of modules over a variable algebra}\label{app:CatModAlg}
We define the category $\tt AMod$ (resp., $\tt FAMod$) of modules (resp., Fr\'echet modules) over any (unital) algebra (resp., any (unital) Fr\'echet algebra) $A$. The algebra $A$ can vary from object to object. The objects are the modules over some $A$ (resp., the Fr\'echet vector spaces that come equipped with a (compatible) continuous $A$-action). We~denote such modules by $M_{A}$. Morphisms consist of pairs $(\varphi , \Phi)$, where
\[
\varphi \colon\ A \longrightarrow B
\]
is an algebra morphism (resp., a continuous algebra morphism), and
\[
\Phi \colon\ M_A \longrightarrow M_B
\]
is a map (resp., a continuous map) that satisfies
\[
\Phi(a m + a'm') = \varphi(a) \Phi(m) + \varphi(a')\Phi(m'),
\]
for all $a,a' \in A$ and $m,m' \in M_A$. It~is evident that we do indeed obtain a category in this way.

The preceding categories $\tt AMod$ and $\tt FAMod$ are similar to the category $\tt AFMan$ that we used in~\cite{Bruce:2019b}. They naturally appear when considering the zero degree rules functor or the functor of~points. See for instance equations~\eqref{CatAMod} and~\eqref{QL0L}.
\section[Basics of Z2n-geometry]
{Basics of $\boldsymbol{\Z_2^n}$-geometry}\label{app:ZZnGeom}

\subsection[Z2n-manifolds and their morphisms]
{$\boldsymbol{\Z_2^n}$-manifolds and their morphisms}

The locally ringed space approach to $\Z_{2}^{n}$-manifolds was pioneered in~\cite{Covolo:2016}. We~work over the field~$\R$ of real numbers and set $\Z_2^n := \Z_2 \times \Z_2 \times \dots\times\Z_2$ ($n$-times). A \emph{$\Z_2^n$-graded algebra} is an $\R$-algebra $\mathcal{A}$ with a decomposition into vector spaces $\mathcal{A} := \oplus_{\gamma \in \Z_2^n} \mathcal{A}_\gamma$, such that the multiplication, say $\cdot$, respects the $\Z_2^n$-grading, i.e., $\mathcal{A}_\alpha \cdot \mathcal{A}_\beta \subset \mathcal{A}_{\alpha + \beta}$. We~will always assume the algebras to be associative and unital. If for any pair of homogeneous elements $a \in \mathcal{A}_\alpha$ and $b \in \mathcal{A}_\beta$ we have that
\begin{gather}\label{eq:Z2ncommrules}
a \cdot b = (-1)^{\langle \alpha , \beta\rangle } b \cdot a,
\end{gather}
where $\langle -, -\rangle$ is the standard scalar product on $\Z_2^n$, then $\mathcal{A}$ is a \emph{$\Z_2^n$-commutative algebra}.

Essentially, $\Z_2^n$-manifolds are ``manifolds'' equipped with both, standard commuting coordinates and formal coordinates of non-zero $\Z_2^n$-degree that $\Z_2^n$-commute according to the general sign rule~\eqref{eq:Z2ncommrules}. Note that in general we need to deal with formal coordinates that are {\it not} nilpotent.

In order to keep track of the various formal coordinates, we need to introduce a convention on how we fix the order of elements in $\Z_{2}^{n}$ and we choose the \emph{lexicographical order}. For example, with this choice of ordering
\[
\Z_{2}^{2} = \{ (0,0), \: (0,1), \: (1,0), \: (1,1)\}.
\]
Note that other choices of ordering have appeared in the literature. A tuple $\ul{q} = (q_{1}, q_{2}, \dots , q_{N})\in\mathbb{N}^{\times N}$ ($N = 2^{n}-1)$ provides the number of formal coordinates in each $\z$-degree. We~can now recall the definition of a $\Z_2^n$-manifold.
\begin{Definition}\label{def:Z2nMan}
A (smooth) $\Z_{2}^{n}$-\emph{manifold} of dimension $p |\ul{q}$ is a locally $\Z_{2}^{n}$-ringed space $ M := \big(|M|, \cO_M \big)$, which is locally isomorphic to the $\Z_{2}^{n}$-ringed space $\mathbb{R}^{p |\ul{q}} := \big( \mathbb{R}^{p}, C^{\infty}_{\mathbb{R}^{p}}[[\zx]] \big)$. Local sections of $M$ are formal power series in the $\Z_{2}^{n}$-graded variables $\zx$ with smooth coefficients,
\[
\cO_M(|U|) \simeq C^{\infty}(|U|)[[\zx]] := \bigg \{ \sum_{\alpha \in \mathbb{N}^{\times N}}^{\infty} f_{\alpha}\,\zx^{\alpha}\colon \, f_{\alpha} \in C^{\infty}(|U|)\bigg \},
\]
for ``small enough'' opens $|U|\subset |M|$. \emph{Morphisms} between $\Z_{2}^{n}$-manifolds are morphisms of $\Z_{2}^{n}$-ringed spaces, that is, pairs $\Phi = (\phi, \phi^{*}) \colon \big(|M|, \cO_M\big) \rightarrow (|N|, \cO_N)$ consisting of a continuous map $\phi\colon |M| \rightarrow |N|$ and a sheaf morphism $\phi^{*} \colon \cO_N(-) \rightarrow \cO_M\big(\phi^{-1}(-)\big)$, i.e., a family of $\Z_2^n$-algebra morphisms $\phi^*_{|V|}\colon \cO_N(|V|) \rightarrow \cO_M\big(\phi^{-1}(|V|)\big)$ ($|V| \subset |N|$ open) that commute with restrictions. We~sometimes denote $\z$-manifolds by $\mathcal{M}=(M,\cO_M)$ instead of $M = \big(|M|, \cO_M \big)$ and we sometimes denote $\z$-morphisms by $\phi=(|\phi|,\phi^*)$ instead of $\Phi = (\phi, \phi^{*})$. \end{Definition}

\begin{Example}[the local model]%\label{exp:SuperDom}
The locally $\Z_{2}^{n}$-ringed space $\mathcal{U}^{p|\ul{q}} := \big(\mathcal{U}^p , C^\infty_{\mathcal{U}^p}[[\zx]] \big)$ ($\mathcal{U}^p \subset \R^p$ open) is naturally a $\Z_2^n$-manifold~-- we refer to such $\Z_2^n$-manifolds as \emph{$\Z_2^n$-domains} of dimension~$p|\ul{q}$. We~can employ (natural) coordinates $\big(x^a, \zx^A\big)$ on any $\Z_2^n$-domain, where the $x^a$ form a~coordinate system on $\mathcal{U}^p$ and the $\zx^A$ are formal coordinates.
\end{Example}

Many of the standard results from the theory of supermanifolds pass over to $\Z_2^n$-manifolds. For example, the topological space $|M|$ comes with the structure of a smooth manifold of dimension~$p$, hence our suggestive notation. Moreover, there exists a canonical projection $\varepsilon \colon \cO_M \rightarrow C_{|M|}^{\infty}$. What makes the category of $\Z_{2}^{n}$-manifolds a very tractable form of noncommutative geometry is the fact that we have local models. Much like in the theory of smooth manifolds, one can construct global geometric concepts via the gluing of local geometric concepts. That is, we can consider a $\Z_{2}^{n}$-manifold as being covered by $\Z_2^n$-domains together with specified gluing data. More precisely, a $p|\ul q$-\emph{chart} (or $p|\ul q$-coordinate-system) over a (second-countable Hausdorff) smooth manifold $|M|$ is a $\Z_2^n$-domain
\[
\mathcal{U}^{p|\ul q}=\big(\cU^p,C^\infty_{\cU^p}[[\zx]]\big),
\]
together with a diffeomorphism $|\psi|\colon |U|\to \cU^p$, where $|U|$ is an open subset of $|M|$. Given two $p|\ul q$-charts
 \begin{gather}\label{Charts}
\big(\mathcal{U}^{p|\ul q}_\za,|\psi_\za|\big)\qquad\text{and}\qquad
\big(\mathcal{U}^{p|\ul q}_\zb,|\psi_\zb|\big)
 \end{gather}
 over $|M|$, we set $V_{\alpha \beta} : = |\psi_\alpha|(|U_{\alpha \beta}|)$ and $V_{ \beta \alpha} : = |\psi_\beta|(|U_{\alpha \beta}|)$, where $|U_{\alpha \beta}| := |U_{\alpha}|\cap |U_{\beta}| $. We~then denote by $|\psi_{\zb\za}|$ the diffeomorphism
 \begin{gather}\label{Diffeo}
 |\psi_{\zb\za}|:=|\psi_\zb|\circ|\psi_\za|^{-1}\colon\ V_{\za\zb} \to V_{\zb\za}.
 \end{gather}
Whereas in classical differential geometry the coordinate transformations are completely defined by the coordinate systems, in $\Z_2^n$-geometry, they have to be specified separately. A \emph{coordinate transformation} between two charts, say the ones of~\eqref{Charts}, is an isomorphism of $\Z_2^n$-manifolds
\begin{gather}\label{CoordTrans}
\psi_{\zb\za}=\big(|\psi_{\zb\za}|,\psi^*_{\zb\za}\big)\colon\ \mathcal{ U}^{p|\ul q}_{\za}|_{V_{\za\zb}}\to \mathcal{ U}^{p|\ul q}_{\zb}|_{V_{\zb\za}},
\end{gather}
where the source and target are the open $\Z_2^n$-submanifolds
\[
\cU^{p|\ul q}_\za|_{V_{\za\zb}}=\big(V_{\za\zb}, C^\infty_{V_{\za\zb}}[[\zx]]\big)
\]
(note that the underlying diffeomorphism is~\eqref{Diffeo}).\;A $p|\ul q$-\emph{atlas} over $|M|$ is a covering $\big(\mathcal{ U}^{p|\ul q}_\za\!,\!|\psi_\za|\big)_\za$ by charts together with a coordinate transformation~\eqref{CoordTrans} for each pair of charts, such that the usual {cocycle} condition $\psi_{\zb\zg}\psi_{\zg\za}=\psi_{\zb\za}$ holds (appropriate restrictions are understood).

Moreover, we have the \emph{chart theorem}~\cite[Theorem 7.10]{Covolo:2016} that says that $\Z_2^n$-morphisms from a~$\Z_2^n$-manifold $(|M|,\cO_M)$ to a $\Z_2^n$-domain $\big(\mathcal{U}^p,C^\infty_{\mathcal{U}^p}[[\zx]]\big)$ are completely described by the pullbacks of the coordinates $\big(x^a,\zx^A\big)$. In other words, to define a $\Z_2^n$-morphism valued in a~$\Z_2^n$-domain, we only need to provide total sections $\big(s^a,s^A\big)\in\cO_M(|M|)$ of the source structure sheaf, whose degrees coincide with those of the target coordinates $\big(x^a,\zx^A\big)$. Let us stress the condition
\[
(\dots,\ze s^a,\dots)(|M|)\subset \mathcal{U}^p,
\]
where $\ze$ is the canonical projection, is often understood in the literature.

\subsection[Z2n-Grassmann algebras, Z2n-points and the Schwarz--Voronov embedding]
{$\boldsymbol{\Z_2^n}$-Grassmann algebras, $\boldsymbol{\Z_2^n}$-points and the Schwarz--Voronov embedding}

It is clear that $\Z_2^n$-manifolds, as they are locally ringed spaces, are not fully determined by their topological points. To ``claw back'' a fully useful notion of a point, one can employ Grothendieck's functor of points. This is, of course, an application of the Yoneda embedding (see~\cite[Chapter~III, Section~2]{MacLane1998}). For the case of supermanifolds, it is well-known, via the seminal works of Schwarz and~Voronov~\cite{Schwarz:1982,Schwarz:1984,Voronov:1984}, that superpoints are sufficient to act as ``probes'' for the functor of points. That is, we only need to consider supermanifolds that have a single point as their underlying topological space. Dual to this, we may consider finite dimensional Grassmann algebras $\zL=\zL_0\oplus \zL_1$ as parameterizing the ``points'' of a supermanifold. One can thus view supermanifolds as functors from the category of finite dimensional Grassmann algebras to sets. However, it turns out that the target category is not just sets, but (finite dimensional) $\Lambda_0$-smooth manifolds. That is, the target category consists of smooth manifolds that have a~$\Lambda_0$-module structure on their tangent spaces. Morphisms in this category respect the module structure and are said to be $\Lambda_0$-smooth (we will explain this further later on). In~\cite{Bruce:2019b}, it was shown how the above considerations generalize to the setting of $\Z_2^n$-manifolds. We~will use the notations and results of~\cite{Bruce:2019b} rather freely. We~encourage the reader to consult this reference for the subtleties compared to the standard case of supermanifolds.

A \emph{$\Z_2^n$-Grassmann algebra} we define to be a formal power series algebra $\R[[\theta]]$ in $\Z_2^n$-graded, $\Z_2^n$-commutative parameters $\theta_j^\ell$. All the information about the number of generators is specified by the tuple $\ul{q}$ as before. We~will denote a $\Z_2^n$-Grassmann algebra by $\Lambda$, as usually we do not need to specify the number of generators. A \emph{$\Z_2^n$-point} is a $\Z_2^n$-manifold (that is isomor\-phic~to)~$\R^{0|\ul q}$. It~is clear, from Definition~\ref{def:Z2nMan}, that the algebra of global sections of a $\Z_2^n$-point is precisely a~$\Z_2^n$-Grassmann algebra. There is an equivalence between $\Z^n_2$-Grassmann algebras and $\Z_2^n$-points:
\[
\Z_2^n\catname{GrAlg} \cong \Z_2^n\catname{Pts}^\textnormal{op}.
\]

The Yoneda functor of points of the category $\z\tt Man$ of $\z$-manifolds is the fully faithful embedding
\[
\mathcal{Y}\colon\ \z{\tt Man}\ni M\mapsto \op{Hom}_{\z{\tt Man}}(-,M)\in{\tt Fun}\big(\z{\tt Man}^{\op{op}},{\tt Set}\big).
\]
In~\cite{Bruce:2019b}, we showed that $\mathcal Y$ remains fully faithful for appropriate restrictions of the source and target of the functor category, as well as of the {\it resulting} functor category. More precisely, we proved that the functor
\[
\mathcal{S}\colon\ \z{\tt Man}\ni M\mapsto \op{Hom}_{\z{\tt Man}}(-,M)\in{\tt Fun}_0\big(\z{\tt Pts}^{\op{op}},\tt A(N)FM\big)
\]
is fully faithful. The category $\tt A(N)FM$ is the category of (nuclear) Fr\'echet manifolds over a (nuc\-lear) Fr\'echet algebra, and the functor category is the category of those functors that send a~$\z$-Grassmann algebra $\zL$ to a (nuclear) Fr\'echet $\zL_0$-manifold, and of those natural transformations that have $\zL_0$-smooth $\zL$-components.

\pdfbookmark[1]{References}{ref}
\LastPageEnding

\end{document}